\documentclass[12pt, reqno]{amsart}
\usepackage{main}
\usepackage{setspace}
\usepackage{natbib}
\usepackage[normalem]{ulem}
\usepackage[all=normal, paragraphs=tight, bibbreaks=tight, floats=tight, bibnotes=tight]
{savetrees}
\usepackage{microtype}
\usepackage{comment}
\usepackage{appendix}
 \usepackage[T1]{fontenc}
  \usepackage{lmodern}
\usepackage{titletoc}

\newcommand\DoToC{%
  \startcontents
  \printcontents{}{0}{\textbf{Contents}\vskip1em\hrule\vskip1em}
  \vskip1em\hrule\vskip5pt
}

\newgeometry{margin=1.25in}

\title{Compound Selection Decisions: An Almost SURE Approach}

\author{
  Jiafeng Chen \and Lihua Lei \and Timothy Sudijono \and Liyang Sun \and Tian Xie
}
\date{\footnotesize \today. Chen: Department of Economics, Stanford University,
  jiafeng@stanford.edu; Lei: Stanford GSB, lihualei@stanford.edu; Sudijono: Department of
  Statistics, Stanford University, tsudijon@stanford.edu; Sun: UCL and CEMFI,
  liyang.sun@ucl.ac.uk; Xie: UCL, tian.xie.20@ucl.ac.uk. The authors thank Isaiah Andrews,
  Xu Cheng, Nikos Ignatiadis, Guido Imbens, Patrick Kline, Charles Manski, and seminar and
  workshop participants at Berkeley, Duke, Georgetown, Harvard/MIT, Northwestern, NTU,
  Princeton, Stanford, SMU, UC Davis, UC Santa Cruz, the 2026 Interaction Conference, and NASMES 2026
  for helpful comments. The authors thank Chris Walters and Toru Kitagawa in particular for discussing this
  paper. Jiafeng Chen is grateful for the support of the Stanford and AWS Marketing Science
  Lab. Lihua Lei is grateful for the support of the National Science Foundation grant
  DMS-2338464. Liyang Sun gratefully acknowledges support from the Economic and Social
  Research Council (new investigator grant UKRI607). OpenAI GPT provided valuable insights
  and superb research assistance. Refine.ink and Claude were used to check the paper for
  consistency and clarity.}
\numberwithin{equation}{section}

\makeatletter
\renewcommand{\paragraph}{\@startsection{paragraph}{4}%
  \z@\z@{-\fontdimen2\font}%
{\normalfont\bfseries}}
\makeatother

\begin{document}
\maketitle
\allowdisplaybreaks
\begin{abstract}
 
  This paper proposes methods for \emph{compound selection decision problems} in a
  Gaussian sequence model. Inspired by Stein's unbiased risk estimate (SURE), we introduce
  \assure{}, a family of estimators for \emph{welfare}, defined as the expected utility of
  a selection decision rule. 
  \assure{} enables robust evaluation of selection decisions. For empirical Bayes decisions derived from random-effects models of unknown parameters, ASSURE estimates their welfare even when the models are misspecified. Optimizing ASSURE-estimated welfare over a pre-specified class of decision rules further borrows strength across noisy estimates and yields decision rules with favorable regret properties. These regret properties are again robust to potential misspecification of random effects models, thereby robustifying empirical Bayes methods.
  We apply
  \assure{} to selecting Census tracts for economic
  mobility, identifying discriminating firms, and evaluating $p$-value
  decision rules in A/B testing.
\end{abstract}
{\footnotesize



}
\pagestyle{plain}

\newpage

\section{Introduction}

Applied researchers often face the following \emph{compound selection} problem: using noisy estimates of parameters, \emph{select} units with large underlying parameters.
Many such problems
can be modeled as follows: For each parallel setting $i=1,\ldots,n$, a researcher
observes a noisy estimate $Y_i$, a standard error $\sigma_i$, some covariates
$x_i$, and costs $k_i$. The estimates $(Y_i, \sigma_i)$ are for certain unknown
parameters $\mu_i$, which represent the unknown payoffs to selection decisions. Given $
(Y_i, \sigma_i, x_i, k_i)$, the researcher chooses thresholds $\delta_i$ and selects all
units with $Y_i >
\delta_i$. They would like to set these thresholds to identify units with $\mu_i > k_i$
and avoid units with $\mu_i < k_i$.

As an example, \citet{bergman2024creating} seek neighborhoods with high economic
mobility to recommend to housing voucher holders. In their context, $Y_i$ is a noisy
estimate of economic mobility of Census tract $i$, $\mu_i$ is true economic mobility,
$\sigma_i$ is the estimated standard error of $Y_i$, $x_i$ collects other contextual
information for a tract, and $k_i$ is the cost of selection.\footnote{We assume $k_i$ is
  known---amounting to assuming that the decision maker knows their decision problem.
  Sometimes, when $\mu_i$ measures net payoffs, we can assume $k_i = 0$. In the case of
  \citet{bergman2024creating}, it may be reasonable to imagine that the costs $k_i$ are
  equal to some value $k$ that represents the economic mobility level for which a social
  planner is indifferent between incentivizing a low-income family to move to a tract
with mobility $k$ and not doing so.} For a fixed set of thresholds, the decision maker's
utility is the average net payoff of the selection decisions:\footnote{Compound selection problems appear in many
  economic applications such as the meta-analysis of experiments \citep {azevedo2020b},
  teacher value-added \citep{chetty2014measuring,kwon2023optimal,cheng2025optimal},
  identifying discrimination \citep {kline2022systemic}, and moving to opportunity
  \citep{chetty2018opportunity,bergman2024creating}. This
  problem also appears in the statistical literature under the name \emph{empirical Bayes
  testing with linear loss} \citep{liang1988convergence, liang2000empirical,
liang2004optimal, karunamuni1996optimal}.}
\begin{equation}
  \label{eq:welfare}
  u(\delta_{1:n}, \mu_{1:n}) := \frac{1}{n} \sum_{i=1}^n \one(Y_i >
  \delta_i) \left (\mu_i -
    k_i
  \right).
\end{equation}
We would like to choose $\delta_{1:n}$ from the data to maximize expected utility. The
statistical challenge is that expected utility depends on the unknown parameters $\mu_
{1:n}$ and cannot be directly optimized.

To proceed, we impose a distributional assumption, motivated by central limit theorems for
the procedure generating $Y_i$: Conditional on $(\mu_i,\sigma_i,x_i)$, the estimate $Y_i$
is Gaussian with known variance \[Y_i \mid \mu_i, \sigma_i, x_i \sim \Norm(\mu_i, \sigma_i^2) \numberthis \label{eq:normal_model}.\]
Following the compound decision and empirical Bayes literatures
\citep{robbins1985asymptotically,efron2012large,walters2024empirical}---for which
\eqref{eq:normal_model} is a standard starting point---we call these statistical decision
problems \emph{compound selection problems}.

Casting the problem as compound allows one to ``borrow strength'' by learning collective
structure across $\mu_{1:n}$, thereby improving relative to conditional empirical success
rules $\one(Y_i > k_i)$ that treat the decision problems separately \citep
{manski2004statistical,stoye2012minimax,hirano2009asymptotics}. For instance, if $\mu_i -
k_i$ is on average negative over $i$, decision makers should favor higher thresholds and
thereby stricter selections, as a randomly drawn unit  from the sample is likely to have
$\mu_i<k_i$---a fundamental insight of \citet{robbins1985asymptotically}.

Existing approaches for compound selection problems are either indirect for \eqref
{eq:welfare}---targeting optimal estimation of $\mu_i$ or testing hypotheses
on $\mu_i$---or they rely on potentially misspecified random-effects models for $\mu_
{1:n}$. For instance,
\citet{bergman2024creating} select high-mobility Census tracts by thresholding
shrinkage estimators for $\mu_i$. These shrinkage estimators can be motivated by optimal
estimation of $\mu_{1:n}$, \`a la \citet{jamesstein}, but it is not obvious that better
prediction necessitates better selection 
\citep{mehta2019measuring,manski2021econometrics}. Alternatively, these shrinkage estimators can be rationalized as empirical Bayes posterior means under a random-effects
model in which $\mu_i$ is i.i.d. Gaussian. If the model is correct, then
screening on these shrinkage estimators is optimal. However,
this Gaussian random-effects model is often misspecified
\citep{chen2022empirical,bonhomme2024estimating,muller2024spatial}.

This paper develops robust methods both for evaluating and for improving compound
selection decisions.  First, for \emph{evaluating} selection decisions, we
propose a novel estimator $\hat W$ of the expected utility---which we term \emph
{compound welfare}---of a given $\delta_{1:n} := (\delta(z_1),\ldots, \delta(z_n))$: \[ W
(\delta_ {1:n}, \mu_ {1:n}) :=
  \frac{1}{n} \sum_{i=1}^n \E_{Y_i \sim
    \Norm(\mu_i,
  \sigma_i^2)} [
    \one(Y_i > \delta(z_i)) \underbrace{\left (\mu_i - k_i
      \right)}_{\text{Net payoff in case selected}} ] \label{eq:welfare_expected} \quad z_i
      := (\sigma_i, x_i, k_i)
  \numberthis.
\] We show that $\hat W$ converges to $W$ for fixed $\mu_{1:n}$ and $z_{1:n}$, where
randomness solely comes from $Y_i$ under \eqref{eq:normal_model}. Thus, the validity of
evaluating welfare using our estimators does not depend on correctly specifying any model
for the unknown parameters $\mu_{1:n}$. This allows practitioners to robustly and
accurately assess the performance of selection decisions possibly derived from restrictive
models for $\mu$.

Second, this welfare estimator also enables practitioners to \emph{improve} selection
decisions
by optimizing within a user-chosen class of decision rules. So long as our estimate of
welfare $W$ is accurate, optimizing the estimated welfare leads to approximately optimal
decisions within the user's chosen class. Such \emph{regret} guarantees are again robust
to misspecification, in the sense that they are valid for fixed $\mu_{1:n}$ rather than
relying on a model of their distribution.

Our strategy is inspired by Stein's unbiased risk estimate \citep[SURE,][]
{stein1981estimation}. In estimation contexts, SURE is an unbiased estimator of the
squared error risk of given predictions $\delta(Y_i, z_i)$ for $\mu_i$. Since it is
unbiased when integrating solely over the
random draws of $Y_{1:n}$, SURE robustly evaluates prediction decisions. SURE
also generates good decisions by optimizing the estimated objective over a parametrized
class of predictions
\citep{xie2012sure,kwon2023optimal,cheng2025optimal,ghosh2025stein}. Yet, SURE targets
mean-squared error prediction and is not obviously suitable for compound
selection problems.
Therefore, extending SURE to selection contexts is both important and nontrivial.

Our first contribution is a new estimator $\hat W$ of \eqref{eq:welfare_expected} that extends SURE to compound selection.  We propose the following class of
estimators, which we term \assure{} (\textbf{A}lmo\textbf{s}t \textbf{SURE}): For a kernel $\kernel_h(y) := \frac{1}{h} \kernel\pr{
\frac{y}{h}}$, define $F_h(y) :=  \int_{-\infty}^{y/h} \kernel(t) \,dt$ so that $F_h\pr{\frac{y-\delta}\sigma}$
is a smooth approximation of $\one(y>\delta)$. We then define
\[
  \hat W_{\assure} = \frac{1}{n} \sum_{i=1}^n (Y_i - k_i) F_h\pr{\frac{Y_i - \delta_i}
  {\sigma_i}}
  -  \sigma_i \kernel_h\pr{\frac{Y_i - \delta_i}{\sigma_i}}. \numberthis
  \label{eq:assure_general}
\] Like SURE, \assure{} is constructed from Stein's identity
\citep{stein1981estimation}. Unlike SURE, we apply Stein's identity by first smoothing
$\one(Y > \delta)$, which results in bias for \assure. This smoothing is necessary,
however, since an exactly unbiased estimator of welfare does not exist
\citep{stefanski1989unbiased}.

Different smoothing kernels produce different bias properties for the welfare estimator.
We focus on two choices.  First, our recommended choice of $\kernel (\cdot)$ is the sinc
kernel $\kernel(x) :=
\frac{\sin x} {\pi x}$ \citep{stefanski1990deconvolving}. We call the resulting estimator \assure{*}.
\assure{*} has bias exponentially vanishing in the tuning parameter $h$.  Under bandwidth
$h=1/\sqrt{2\log n}$, the squared bias and variance are both $O(1/n)$, up to log
factors; thus the \assure{*} estimator converges at near-parametric rates.

Second, an alternative smoothing uses the Gaussian kernel, which results in a welfare
estimator that has a natural sample-splitting interpretation: It is equivalent to a
Rao--Blackwellized version of \emph{coupled bootstrap}
\citep{oliveira2024unbiased,chen2022empirical,ignatiadis2025empirical}.\footnote{Coupled bootstrap \citep{oliveira2024unbiased} is an instance of data thinning \citep{neufeld2024data} or data
  fission \citep{leiner2023data, lei2025discussion}, where independent variates are constructed by adding noise to data. With
  Gaussian $Y$, we can let
  $Y_{i1} = Y_i + \epsilon Q_i$ and $Y_{i2} = Y_i - \frac{1}{\epsilon} Q_i$ for
  independent $Q_i
  \sim \Norm(0,\sigma_i^2)$. Then $Y_{i1} \indep Y_{i2} \mid \mu_i,
  \sigma_i^2$. When $Y_i$ is the sample mean of micro-data $Y_{ij}$, the
  constructed $Y_{i1}, Y_{i2}$ can be viewed as sample means on disjoint subsamples of $
  \br{Y_
  {ij}}_j$. Proposition 1 in \citet{chen2022empirical} uses this \emph
  {coupled bootstrap} idea to estimate \eqref{eq:welfare} for selection decisions that
depend on $Y_{i1}$.} We thus call this estimator \assurecb. Because the Gaussian kernel is
only a second-order kernel, it does
not fully exploit the smoothness of the estimand, resulting in slower convergence for 
\assurecb{} compared to \assure{*}.

Evaluation of decision rules naturally leads to optimizing the estimated welfare within
user-selected class of decision rules $\mathcal D$. This workflow operates with arbitrary
$\mathcal D$ and is thus versatile. In practice, these decisions classes can come from
external preference for simple decisions, such as in
\citet{kitagawa2018should,sudijono2024optimizing,crippa2025regret}. They may also come
from a \emph{working} (empirical) Bayesian model on the parameters $\mu_{1:n}$---similar
to \citet{kwon2023optimal} and \citet{cheng2025optimal}. A flexible class, which we call
$\mathcal D_{\text{\closegauss}}$ following \citet{chen2022empirical}, is motivated by a
correlated Gaussian random effects model for $\mu_{1:n}$ \citep{weinstein2018group}, which
nests many linear shrinkage estimators.

Our second contribution derives statistical guarantees for \assure{*}, in the regime where
the number of parallel settings diverges ($n\to\infty$), but each
setting contains a noisy estimate ($\sigma_i^2 \not\to 0$). We show that the
performance gap between the
estimated optimal decision and the population optimizer of
\eqref{eq:welfare_expected} within $\mathcal D$---i.e., \emph{regret}---converges to zero at
minimax optimal rates, up to log factors.  Applied to $\mathcal D_{\text{\closegauss}}$, 
decisions tuned by
\assure{*} would asymptotically match the performance of the best
member of $\mathcal D_{\text{\closegauss}}$, and thus outperform alternative methods for
selecting from $\mathcal D_{\closegauss}$, including plug-in empirical Bayes.

These favorable properties rely on the Gaussian distributional assumption
\eqref{eq:normal_model}, which, though standard, assumes the
central limit theorem provides a perfect approximation. To assess robustness to the 
approximation quality, we
 analyze \assure{*} as an estimator of welfare in a regime where $Y_i$ are sample means of
 $m$ i.i.d. non-Gaussian random variables \citep {hirano2009asymptotics,hirano2016panel}.
 We consider a local asymptotic regime, where both the number of parallel estimates $n$
 and the sample sizes for each estimate $m$ diverge to infinity, but the underlying
 parameters are within some $1/\sqrt {m}$ neighborhood and thus have meaningful
 statistical ambiguity. We find that the performance of
 \assure{*} degrades gracefully when Gaussianity becomes approximate. Moreover, because
 \assure{*} exploits the Gaussianity of $Y_i$, it estimates welfare more accurately than
 competing approaches that only exploit the consistency of the noisy estimates. This analysis partially addresses questions raised by \citet
 {https://doi.org/10.3982/ECTA20364} for EB and compound decision approaches.

This paper is closely related to the empirical Bayes literature
\citep[among many others,][]
{jiang2020general,gu2023invidious,walters2024empirical,kline2024discrimination,chen2022empirical}.
EB analyses treat $\mu_ {1:n}$ as random effects and integrate over their distribution,
evaluating decisions through the Bayes risk instead:\[W_{\text{EB}}(\delta_{1:n}) = \E_
  {\mu_{1:n} \sim P}[W (\delta_ {1:n},
\mu_{1:n})] \text{ for } \mu_{1:n} \mid \sigma_{1:n}, k_{1:n}, x_{1:n} \sim P.\]
EB is versatile for many decision problems---simply requiring  estimating the
distribution $P$ and computing Bayes-optimal decisions relative to the estimated $P$. So
long as this estimated prior is accurate for $P$, the resulting decision rule is
approximately optimal over \emph{all} decisions.

That versatility, however, depends critically on modeling the distribution of $\mu_ {1:n}$
well. Failing to do so---either misspecifying how $\mu_i$ correlates with contextual
information \citep{chen2022empirical}, how $\mu_i$ correlates with $\mu_j$
\citep{bonhomme2024estimating}, or whether $\mu_i$ is distributed parametrically \citep
{gilraine2020new}---could harm the selection decision. The full \emph{joint} distribution
of $\mu_{1:n}$ conditional on \emph{all} covariates is exceedingly difficult to
faithfully model and flexibly estimate; feasible empirical Bayes models are thus usually at risk of
misspecification.\footnote{A similar robustness concern motivates the fixed effects panel
literature \citep {dano2025binary}, relative to correlated random effects.} 

Instead, \assure{} enables using empirical Bayes models as \emph{working approximations}. A (parametric) empirical Bayes model generates a decision rule class $\mathcal {D}$ through its posterior.\footnote{Empirical Bayesians solve selection problems by thresholding on the
  (estimated) posterior mean $\E_P[\mu_i \mid Y_i, x_i, \sigma_i] > k_i$. For Gaussian $Y$
  the posterior mean is monotone in $Y_i$, so the rule is equivalent to $\one(Y_i >
  \delta(z_i; P))$. If $P$ is parametric, then this yields a parametric class $\mathcal D$
  of decision rules. } Instead of plugging in an estimate $\hat P$, one could choose within this $\mathcal D$ by \assure. Since \assure's validity does not depend on whether the EB model is correctly specified, doing so asymptotically recovers the best decision within $\mathcal D$. Under correct specification of the EB model, \assure-based decisions do little harm. Under misspecification, if $\mathcal D$ includes simple decision rules like $\delta_i = k_i$,
selection via \assure{} ensures that decisions are no worse than the simple
benchmarks. In this way, \assure{} robustifies empirical Bayes by prioritizing features of $P$ that are most
relevant for the selection problem under a given $\mu_{1:n}$.
We present Monte Carlo exercises in \cref{sec:simulations,sub:discussion} and
three empirical applications in \cref{sec:empirical_application} that illustrate this
point.

This paper is structured as follows. \Cref{sec:unbiased_welfare_estimate} defines the
\assure{} estimators and discusses their bias and variance.
\Cref{sec:theory} presents several theoretical guarantees on selecting a decision using
\assure{*}, showing upper and lower bounds on regret. \Cref{sec:extensions} presents two
extensions on complicated decision procedures and on settings where Gaussianity holds
approximately. Simulation results and empirical applications are discussed in \cref
{sec:simulations} and \cref{sec:empirical_application}. Proofs are found in the Appendix
alongside additional computational details and an extension of our framework to Poisson
observations.

\paragraph{Notation} We use $u\star v$ to denote convolution. We use $\tilde O
(\cdot), \tilde O_P(\cdot)$ as the analogues of big-$O$ notation that ignores
logarithmic factors. $\polylog(n)$ denotes logarithmic factors that are bounded by
$(\log n)^C$ for some fixed $C>0$. We say $f = \Omega (g)$ if $g = O (f)$. $\lesssim,
\gtrsim$ record
inequalities up to an absolute constant. We use $[n] := \br{1,\ldots,n}$.

\section{Setup and Methodology}
\label{sec:unbiased_welfare_estimate}

\subsection{Preliminaries and decision rules}

Let $Y_i \sim \Norm(\mu_i,\sigma_i^2), i  \in [n]$ with known $\sigma_i$, known costs
$k_i$, with known covariates $x_i$. We treat both the contextual
information $z_i := (x_i,k_i,\sigma_i)$ and the parameters $\mu_i$ as fixed, so all
probability statements are taken over the Gaussian draws of $Y_i$. For simplicity, we
assume that $Y_i$ are mutually independent across $i$.\footnote{If $ (Y_1,\ldots,  Y_n)
\sim \Norm((\mu_1,\ldots,\mu_n)', \Sigma)$ for some known $\Sigma$, then our arguments are
generalizable by studying the conditional distributions $Y_i \mid Y_{-i}$
\citep{fithian2022conditional, luo2025improving, luo2024estimating}. } We also restrict
attention to \emph{separable} thresholds of the form $\delta_i(Y_{-i}) =
\delta(z_i;\beta)$, so unit $i$'s threshold depends on other units only through the
estimated parameter $\beta$. With this parametrization, \eqref{eq:welfare} and
\eqref{eq:welfare_expected} can be rewritten as follows:
\begin{align}
  u(\beta) &:= \frac{1}{n} \sum_{i=1}^n \one(Y_i > \delta(z_i;\beta)) \cdot (\mu_i - k_i),
  \label{eq:u_beta}\\
  W(\beta) &:= \frac{1}{n}\sum_{i=1}^n (\mu_i - k_i) \P\left(Y_i \geq \delta
    (z_i;\beta)
  \right) = \frac1n \sum_{i=1}^n (\mu_i - k_i) \Phi\left( \frac{ \mu_i - \delta(z_i;\beta)}
  {\sigma_i} \right) \label{eq:general_welfare}.
\end{align}

These quantities are defined with respect to some parametrization of decision rules
$\delta(\cdot; \beta)$, which we collect in a class $\mathcal D = \br{\delta
( \cdot ;\beta) : \beta
\in \mathcal B \subset \R^d}$. Below, we enumerate several practically relevant classes
 that fall within this framework:
\begin{example}
\label{sec:examples}
\begin{enumerate}[wide]
  \item Simple truncation rules: $\cl{D}_{\text{threshold}} = \br{\delta(z;\beta) = k +
    \beta : \beta \in [-M,M]}$.
  \item Thresholding $t$-statistics: $\cl{D}_{\text{$t$-stat}} = \br{
    \delta(z;\beta) = k + \beta \sigma : \beta \in [-M, M ]}$. Such rules nest selection via
    one-sided $p$-values \citep{mogstad2024inference}, which is common in digital experimentation
    \citep{kohavi2020trustworthy} and medicine \citep{manski2019treatment}.
  \item Finite-dimensional thresholds: $\mathcal D_{\text{finite}} = \mathrm{conv}(
    \set{\delta^{(1)}, \dots, \delta^{(p)}})$
  \item Gaussian empirical Bayes models: Consider Gaussian priors on $\mu_i \mid z_i
  \sim \Norm (m_0
(z_i;\beta), s_0^2(z_i; \beta))$ and generate corresponding Bayes decision rules: 
\[\mathcal D_{\text{\closegauss}} =
      \br{\delta (z_i;\beta) = k_i + \frac{\sigma_i^2}
    {s_0^2(z_i;\beta)}(k_i - m_0(z_i;\beta)): \beta \in \mathcal B \subset \R^d}.
    \numberthis \label{eq:close_gauss_example}\] This produces  a family of decision rule
     classes by parametrizing $s_0, m_0$ as different functions of $\beta$. \Cref{tab:close-gauss-fam} names some of them.
\end{enumerate}
 \begin{table}[htb]
  \small
  \begin{tabular}{lll}
    \toprule
    \text{Name}
    & \text{$m_0$ specification}
    & \text{$s_0^2$ specification}
    \\
    \midrule
    \text{Linear shrinkage}
    & $\mu_0$
    & $\tau^2$ \\
    \text{Fay-Herriot  \citep{fay1979estimates}}
    & $x_i' \beta$
    & constant $A$ \\
    \text{Parametric \closegauss{} \citep{chen2022empirical,walters2024empirical} }
    & $a_1 + a_2 \log \sigma_i$
    & $\exp(b_1 + b_2 \log \sigma_i)$\\
    \bottomrule
  \end{tabular}
  \caption{Examples of Gaussian empirical Bayes models}
  \label{tab:close-gauss-fam}
\end{table}
\end{example}

\subsection{ASSURE-based welfare estimators} We now turn to estimators $\hat W_n
(\beta)$ for
 $W(\beta)$ and study their estimation error $\hat W_n - W$. Naturally, an estimator
 $\hat W_n$ with small estimation error is valuable for evaluating a fixed decision rule.
 Moreover, because $\hat W_n$ provides an empirical objective, practitioners can also
 choose $\beta$ by optimizing $\hat W_n(\cdot)$:
\[
  \hat\beta \in \argmax_{\beta \in \mathcal B} \hat W_n(\beta).\numberthis
  \label{eq:sure_type}
\]
Accurate estimation of $W$  also enables the estimated decision rule
$\hat\beta$
to perform well---using bounds on $\hat W - W$, our subsequent theoretical results
control the difference, or
\emph{regret}, between the
welfare of the best decision rule in $\mathcal D$ and the expected utility of
$\hat\beta$: \[
  \regret_n = \sup_{\beta \in \mathcal B} W(\beta) - \E_{\mu_{1:n}}[u(\hat\beta)].
  \numberthis \label{eq:regret}
\]
For these two reasons, a good welfare estimator for the compound selection problem is
the focus of this paper.

\smallskip
\paragraph{The \assure{} functional form \eqref{eq:assure_general} as kernel smoothing
Stein's
identity. } Let $F (\cdot)$ be differentiable with derivative $f (\cdot)$. Fix some $k,
\delta, \mu \in \R$ and $\sigma^2 > 0$, and consider $Y \sim \Norm(\mu,\sigma^2)$. Simple manipulations of Stein's identity reveal
that $(Y-k) F(Y-\delta) - \sigma^2
f(Y-\delta)$ is an unbiased estimator for $(\mu-k) \E[F(Y-\delta)]$:
\[
  \E[(Y-k) F(Y-\delta) - \sigma^2 f(Y-\delta)] = (\mu-k) \E[F(Y-\delta)]. \tag{Stein's identity}
  \label{eq:stein_identity}
\]
On the other hand, with \eqref{eq:general_welfare}, estimating $W(\beta)$ reduces to
estimation of a particular function of $\mu$, which can be written as $(\mu-k)\E[ H
(Y-\delta)]$ for the \emph{non-differentiable} $H(t):=\one(t > 0)$: \[(\mu-k) \Phi\pr{
\frac{\mu-\delta} {\sigma}} = (\mu - k)
\E[\one(Y > \delta)].\]

Accordingly, we seek a differentiable $F$ whose expectation approximates $\Phi((\mu-\delta)/\sigma)$: \[\E[F
  (Y-\delta)]\approx \Phi
  \pr{\frac{\mu-\delta}{\sigma}} = \E\bk{\one\pr{\frac{Y - \delta }{\sigma} > 0}}.
  \numberthis
\label{eq:F_criterion}\]
For finding a suitable $F$, a natural idea is to smooth the step
function $H(y) := \one(y > 0)$ with some kernel $\kernel_h(y) = \frac{1}{h} \kernel\pr{\frac{y}
{h}}$ so that the resulting function is
differentiable: \[ F_h(y) := (H
  \star \kernel_h)(y) = \int_{-\infty}^{y/h} \kernel(t) \,dt.
\]
This leads to the proposed class of estimators in~\eqref{eq:assure_general}, since
$\sigma^2 \frac{d}{dy} F_h\pr{\frac{y-\delta}{\sigma}} = \sigma \kernel_h\pr{\frac{y-\delta}{\sigma}}$.

\smallskip
\paragraph{Choice of kernel and bandwidth.} Which kernel should we choose? It turns out
that the
\emph{sinc kernel} for
\eqref{eq:assure_general} leads to optimal estimation and regret properties, building
on the deconvolution literature
\citep{kolmogorov1950unbiased,tate1959unbiased,liang2000empirical,pensky2017minimax,zhou2019fourier,stefanski1990deconvolving}.
We call this variant \assure{*}. 

Choosing bandwidth $h = 1/\sqrt{2\log n}$, \assure{*} is
the following estimator:
\begin{align}
  \label{eq:assure_term}
  \hat{W}_n(\beta) &:= \frac{1}{n}\sum_{i=1}^n w_{1/\sqrt{2\log n}} (Y_i; z_i, \beta)
  \\ \text{ for }
  w_h(Y_i; z_i, \beta)  &:= (Y_i - k_i) \Csinc\pr{
    \frac{Y_i - \delta(z_i;\beta)}{\sigma_i h}
  } - \frac{\sigma_i}{h} \sinc\pr{
    \frac{Y_i - \delta(z_i;\beta)}{\sigma_i h}
  }.
\end{align}
where $
\sinc(x) = \frac{\sin x}{\pi x}$, $\Si(x) = \int_0^x \frac{\sin t}{t} dt$, and $
\Csinc(x) = \int_{-\infty}^x \sinc(t)\,dt = \frac{1}{2} + \frac{1}{\pi} \Si(x).
$

The next proposition shows that, pointwise in $\beta$, \assure{*} achieves squared error of
order $O(\sqrt{\log n}/n)$.  
\begin{proposition}
  \label{thm:assure_bias}
  \begin{enumerate}[wide]
    \item Fix $\mu, k, z$ where $\delta(z;\cdot)$ is defined. Let $Y \sim \Norm(\mu,
    \sigma^2)$. The bias of \assure{*} is exponentially
    decreasing in $1/h$: For all $h > 0$, 
    \[
        \abs{\E_\mu w_h(Y; z, \beta)  - (\mu - k) \Phi\pr{\frac{\mu - \delta
        (z;\beta)}
        {\sigma}}} \le |\mu - k| h^2 e^{-\frac{1}{2h^2}}.
      \]
      The right-hand side converges to zero as $h \to 0$.
    \item For $h = 1/\sqrt{2\log n}$, for fixed $\beta$, the  squared error of 
      \assure{*}
      satisfies
      \begin{align*}
        \E\bk{
          \pr{
            \hat W_n (\beta) - W(\beta)
          }^2
        } & \lesssim
        \frac{1}{n}\left(\frac{1}{n}\sum_{i=1}^n(\mu_i-k_i)^2
        +\frac{\sqrt{\log n}}{n}\sum_{i=1}^n\sigma_i^2\right).
        \numberthis \label{eq:pointwise_mse}
      \end{align*}
      In particular, under uniformly bounded second moments, the pointwise
      mean-squared error is $O(\sqrt{\log n}/n)$.
  \end{enumerate}
\end{proposition}

The choice of the sinc kernel is motivated by its bias properties. Since the kernels are scaled so that the variances of $\hat W$ are $O((nh)^{-1})$, bias properties of different kernels determine estimation rates for $W$. 
\Cref{fig:assure_function}
illustrates \assure{*}'s rapidly vanishing bias as a function of $h$. This bias performance leads to attractive mean-squared error for \assure*; in fact, the rate \eqref{eq:pointwise_mse} is exactly minimax optimal for estimating $W
(\beta)$ pointwise (\cref{thm:minimax_welfare_lower_bound}). 
\begin{figure}[h]
  \centering
  \includegraphics[width=0.7\textwidth]{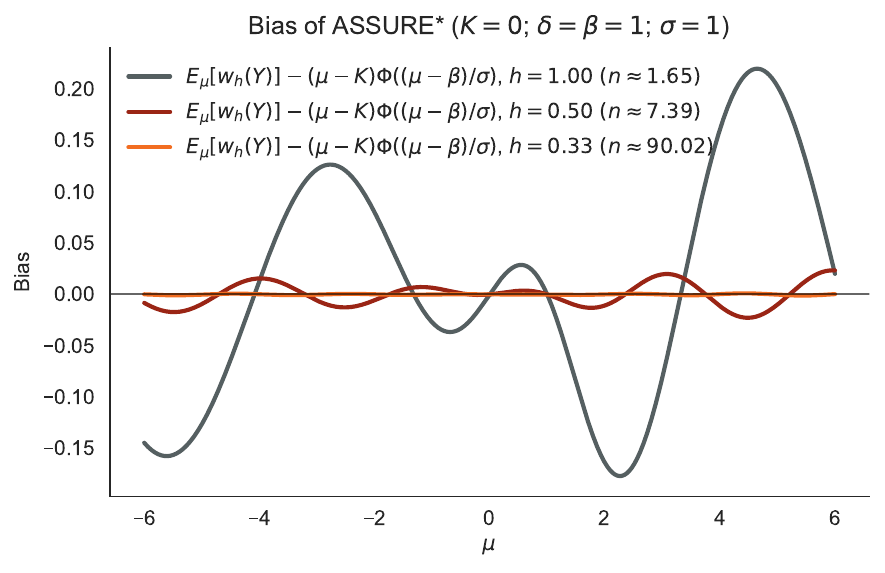}
  \caption{Illustration for the bias of \assure{*}. As $h$ decreases from $1$ to $0.33$, the bias of \assure*
rapidly vanishes.}
  \label{fig:assure_function}
\end{figure}

Unlike SURE, \assure{*} and the broader \assure{} family are not exactly unbiased for welfare. This lack of exact unbiasedness is unavoidable. We formalize in \cref{prop:unbiased_stefanski} an argument
in \citet{stefanski1989unbiased} that shows unbiased estimators (that grow slower than
$e^{y^2/(2\sigma_i^2 + \epsilon^2)}$ as functions of $Y_i$) do not exist. Short of an
unbiased
estimator, \assure{*} implements the next-best through the sinc kernel \citep
{davis1975mean, liang2000empirical, Tsybakov2009IntroNonparametricEstimation}.

To see why the sinc kernel results in low bias, return to \eqref{eq:F_criterion} and
consider the case with $\delta = 0, \sigma=1$ for simplicity. Our kernel-smoothed $F_h$
has expectation
\[
  \E_{Y \sim \Norm(\mu,1)}[F_h(Y)] = (H \star \kernel_h \star \varphi) (\mu).
\]
On the other hand, the target parameter is a convolution without $\kernel_h$, $ \Phi(\mu) = (H\star
\varphi) (\mu).$
Thus, a low-bias kernel is one that barely smoothes the Gaussian density: $\kernel_h \star
\varphi \approx \varphi$. 

The sinc kernel is exactly motivated by this approximation in
frequency space:\footnote{A large literature on Gaussian estimation and deconvolution
follows similar ideas
\citep{kolmogorov1950unbiased,zhou2019fourier,pensky2017minimax,tate1959unbiased,stefanski1989unbiased}.}
By taking the Fourier transform of both sides and recalling $(\mathcal F \varphi)
(\omega) = e^{-\omega^2/2}$, we would like a kernel such that \[(\mathcal F
\br{\kernel_h}) e^ {-\omega^2/2} \approx e^{-\omega^2/2} \text{, for } \mathcal F \br{\kernel_h}
  :=
  \int_\R
  \kernel_h (t)
  e^{-it\omega}\,dt.
\]
The sinc kernel has Fourier transform $\mathcal F\br{\kernel_h}(\omega) = \one(|\omega| <
1/h)$, which
only truncates the
high frequency signals in $\varphi$. Since $\mathcal F\br{\varphi}$ has Gaussian
tails, this
truncation leaves it essentially unchanged---allowing the bias to decay exponentially in
$1/h$.

\begin{remark}[\assurecb]
\label{sec:heuristic} 

Another natural kernel---though with worse
bias behavior than \assure{*}---is the Gaussian kernel $\kernel
(\cdot) = \varphi (\cdot)$.
Interestingly, the resulting \emph{\assurecb{} estimator} has an
interpretation as \emph {coupled bootstrap}
\citep{oliveira2024unbiased,leiner2023data,ignatiadis2025empirical,chen2022empirical}. For
$Y \sim \Norm(\mu, \sigma^2)$, we can construct two independent Gaussians by adding and
subtracting an independent Gaussian noise $Q$ \[ Y_1 = Y +  \sigma h Q \quad Y_2 = Y -
  \frac{\sigma}{h} Q \quad Q \sim \Norm(0,1)
  \implies Y_1 \indep Y_2.
  \numberthis \label{eq:cb}
\] This \emph{coupled bootstrap} procedure is a synthetic version of
sample-splitting.\footnote{If $Y$ is a sample mean of Gaussian micro-data, $Y
  = \frac
  {1}{m} \sum_{j=1}^m Y_ {(j)}$ for $Y_ {(j)} \sim \Norm(\mu, m \sigma^2)$, then
  sample-split means over $Y_{(1)},\ldots, Y_{(m)}$ can be represented in the form
\eqref{eq:cb}.}

With \eqref{eq:cb}, the welfare of selection decisions based on $Y_1 > \delta$ can be
unbiasedly estimated by $(Y_2 - k)\one(Y_1 > \delta)$, since $Y_2$ acts as fresh ``testing
data'' that is independent of the ``training data'' $Y_1$. Because $Y_1$ is slightly
noisier than the original estimate $Y$, this estimator is biased for the welfare of $\one
(Y >
\delta)$.

The estimator
\assurecb{}---choosing $\kernel(\cdot) = \varphi(\cdot)$ in \eqref{eq:assure_general}---is exactly the
Rao--Blackwellization of coupled bootstrap:
\[
  \E[(Y_2 - k)\one(Y_1 > \delta) \mid Y] = (Y-k) \Phi\pr{\frac{Y-\delta}{\sigma h}} - \frac{\sigma}{h}\varphi\pr{\frac{Y-\delta}{\sigma h}}.
\]
Because the Gaussian kernel is only a second-order kernel, \assurecb{} has $O(h^2)$ bias
and therefore converges more slowly than \assure{*}.\footnote{\assure{*} is $\tilde
  O(n^{-1/2})$-consistent for welfare, whereas \assurecb{} is only $\tilde
  O(n^{-2/5})$-consistent.
\Cref{appx sec:cb} provides a more formal comparison.} 
\end{remark}

\subsection{Compound decisions, policy learning, and empirical Bayes}
\label{sub:discussion}

We now contextualize our approach in the literature and compare it against a few
alternatives. We also disambiguate our regret notion \eqref{eq:regret} from the regret
notions in these literatures.  We illustrate these comparisons using a classic example
from \citet{robbins1985asymptotically}.

\begin{example}[\citet{robbins1985asymptotically}]
  \label{ex:robbins}
  The parameters are binary, $\mu_ {1:n} \in \br{-1,+1}^n$, of which $\rho n$ have $\mu_i
  = +1$ and the remaining units have $\mu_i = -1$.  For each unit $i$, suppose we observe
  $Y_i\sim \Norm(\mu_i,1)$ and set $k_i = 0$ for simplicity. Restrict decision rules to
  the class $\delta_i(\beta) = \beta$ in the decision problem ~\eqref{eq:welfare}. The
  welfare in this setting is equal to $W(\beta) = \rho \Phi (1-\beta) -
  (1-\rho)\Phi(-1-\beta)$. If $\mu_ {1:n}$---or equivalently $\rho$---is known, then the
  optimal decision for \eqref{eq:welfare_expected} sets $\beta^* = \frac{1}{2} \log
  \frac{1-\rho}{\rho}$.
\end{example}

\begin{figure}
  \centering
  \includegraphics[width=\textwidth]{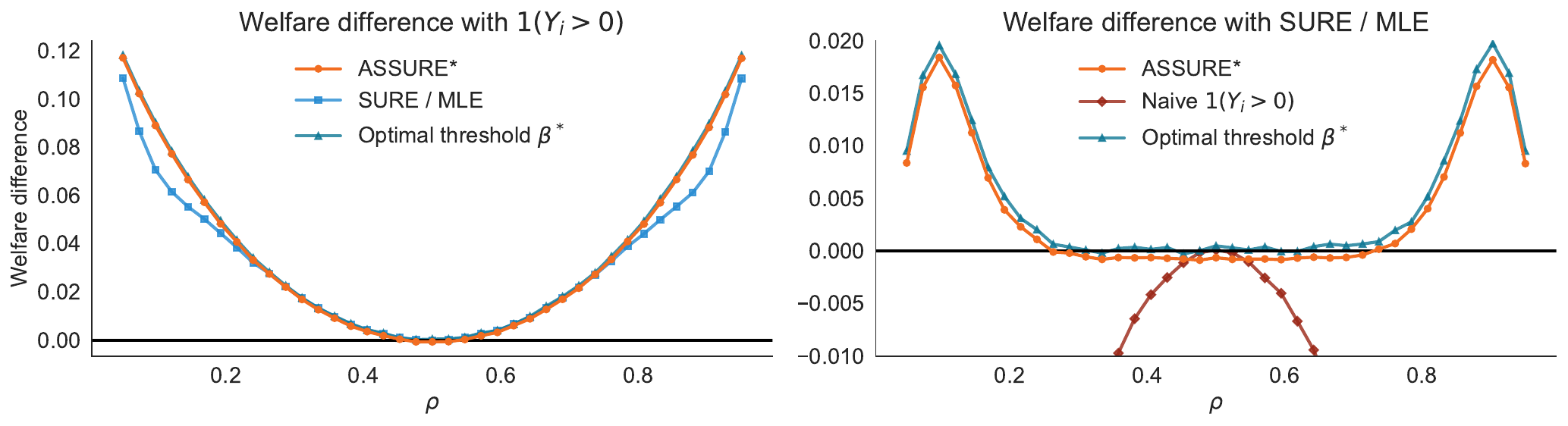}
  \caption{Welfare comparisons of different methods for fitting decision rules for
  \cref{ex:robbins}, as a function of the proportion  $\rho$ of $\mu_i = +1$.}

  \begin{proof}[Notes]

    We consider decision rules---\assure{*}, MLE/SURE, oracle $\beta^* = \frac{1}{2} \log
    \frac{1-\rho}{\rho}$, and $\one(Y_i > 0)$---and compare their performances. The left
    panel shows differences against $\one(Y_i > 0)$ and the right panel shows differences
    against fitting $\beta$ with SURE/MLE. In particular, SURE/MLE rationalizes the
    selection $\one(Y_i > \beta)$ with a Gaussian random effects model $\mu_i \sim \Norm
    (m_0, s_0^2)$ and fits $(m_0, s_0^2)$ through MLE or through SURE; they imply identical
    selection decisions.
  \end{proof}

  \label{fig:Robbins_skew_example}
\end{figure}

\smallskip 
\paragraph{Treatment choice.}
\Cref{ex:robbins} can be thought of as a compound version of the ``testing an innovation''
problem in \citet {manski2009identification}, where $\mu_i$ could represent the average
treatment effect of a new treatment relative to a known status quo.  Modeling this
problem as \emph{compound} changes our statistical perspective relative to the treatment choice
literature. That literature often seeks decision rules that optimize for worst-case
$\mathsf{FreqRegret}$ \citep [e.g., in][]{ishihara2022shrinkage}, where
\[
  \mathsf{FreqRegret}_n := \E_{\mu_{1:n}}\bk{
    \frac{1}{n} \sum_{i=1}^n \max(0, \mu_i - k_i) - u(\hat\beta)
  } \overset{\text{\cref{ex:robbins}}}{=} \rho - \E[u(\hat \beta)].
\]
By contrast, our goal is asymptotically vanishing regret in \eqref{eq:regret}, uniformly over suitable sequences $(\mu_i,z_i)_{i=1}^n$. Specializing to \cref{ex:robbins},
\eqref{eq:regret} is instead \[
  \regret_n = \rho\Phi\pr{
    1-\beta^*
  } - (1-\rho) \Phi\pr{
    -1-\beta^*
  }
- \E[u(\hat\beta)].\]

These desiderata are distinct. Our later results show that \assure{*} ensures vanishing
$\regret_n$, but its worst-case $
\mathsf{FreqRegret}$ for any finite $n$ may not be optimal. Conversely, the empirical
success rule $\one (Y_i > 0)$ achieves minimax-$\mathsf{FreqRegret}$ \citep
[Proposition 1,][]{stoye2012minimax}, but its regret
\eqref{eq:regret} does not vanish: $\sup_ {\mu_ {1:n}} \regret_n = \Omega (1)$. Thus,
along many sequences of $\mu_{1:n}$, eventually \assure{*} achieves higher welfare than
the empirical success rule. \Cref{fig:Robbins_skew_example} illustrates this with \cref
{ex:robbins}: The welfare of \assure{*} closely tracks that of the oracle $\beta^*$ and
exceeds that of the empirical success rule for all $\rho \not\approx 0.5$. \assure
{*} slightly underperforms the empirical success rule when $\rho$ is close to $0.5$.

This improvement comes from a key insight of \citet{robbins1985asymptotically}.
Intuitively, if $n$ is large, certain aggregate quantities such as
$
\frac{1} {n} \sum_ {i=1}^n (\mu_i - k_i)^q
$
are precisely estimable from the data and carry information about the configuration of
$\mu_{1:n}$. For instance, the average $\frac{1}{n}\sum_{i=1}^n (\mu_i-k_i)$ can be
informative: A positive estimate provides evidence that selection is
net-beneficial on average. If we are confident that $\frac{1}{n}\sum_ {i=1}^n (\mu_i-k_i)
> 0$, more lenient decision rules such as $\mathbf{1}(Y_i \geq - \e)$  can improve 
compound welfare, since any mildly negative realizations of $Y_i$ are then more likely due
to chance. This is especially stark in \cref{ex:robbins}, where $\frac{1} {n}\sum_i Y_i
\approx \frac{1}{n}\sum_{i=1}^n(\mu_i-k_i) = 2\rho - 1$ is directly informative of the key
quantity $\rho$.
\assure{*} precisely exploits this kind of information---even in settings without simple
statistics like $\rho$---through estimating the objective function.

\smallskip 
\paragraph{Empirical Bayes ($G$-modeling).} A popular framework for compound decision
problems is \emph{empirical Bayes}, which models $\mu_i$ as random effects
\citep{efron2012large}. This additional distributional structure implies that optimal
decision rules are posterior quantities involving the estimable distribution of $\mu_i$.
Formally, EB methods assume that $\mu_{1:n} \mid z_{1:n} \sim P$. If the distribution
$P$ were known, then the optimal decision with respect to expected welfare compares the
posterior mean of $\mu_i$ to $k_i$: $\E_P[\mu_i
\mid Y_i, z_i] > k_i$. Thus, EB methods propose estimating $P$ and plug $\hat P$ into posterior decision rules. 

However, this recipe  hinges on modeling and estimating the \emph{entire joint distribution} of
$\mu_{1:n} \mid z_{1:n}$ well, which can make subsequent procedures and guarantees less
robust. In full generality, $\mu_{1:n}
\mid z_ {1:n}$ can be highly complex: $\mu_i,\mu_j$ may be correlated
\citep{bonhomme2024estimating,muller2024spatial}; they may be flexibly predicted by
rich covariates \citep{chen2022empirical}; or $\mu_i$'s randomness may simply be
questionable if we observed the entire population. Even nonparametric empirical
Bayes methods
necessarily make some compromises,\footnote{For instance, \citet{chen2022empirical}
assumes $\mu_ {1:n}$ are mutually independent and restricts dependence on covariates to a
location-scale family.} and we may be concerned that these compromises harm resulting
decisions. Moreover, it is unclear whether estimating \emph{all} features of $P$ well is
necessary for a given decision problem; perhaps different decision problems call for
different ways of estimating $P$.

Our procedure is thus complementary to empirical Bayes. EB-style modeling for $\mu_i$ is
apt for generating reasonable decision rule classes $\mathcal D$---indeed optimal ones if the model on $\mu_i$ happens to be correct. (A leading family of such
EB-derived decision rules is the \closegauss{} class in \Cref{sec:examples}.) Choosing a
member of $\mathcal D$ with \assure{} attains guarantees for the
compound selection problem directly---thus not requiring practitioners to take the random
effects model fully seriously. When the random-effects model is correct, tuning via
\assure{} does little harm. These advantages are attractive features of SURE
\citep{ghosh2025stein,cheng2025optimal} that \assure{} inherits.

In terms of theoretical analysis, the EB literature studies a different regret criterion that averages
over draws of $\mu_{1:n}$ and benchmarks against the oracle Bayes rule
\citep{jiang2020general,soloff2024multivariate,chen2022empirical}:\[
  \mathsf{EBRegret}_n = \E_{\substack{\mu_{1:n} \sim P \\ Y_i \mid \mu_i, z_i \sim \Norm
      (\mu_i,
  \sigma_i^2)}} \bk{\frac{1} {n}\sum_ {i=1}^n \max(0, \E_P[\mu_i \mid Y_i, z_i] - k_i)
    - u(\hat P)
  },
\]
where $u(\hat P) = \frac{1}{n} \sum_{i=1}^n \one(\E_{\hat P}[\mu_i \mid Y_i, z_i] > k_i)
(\mu_i - k_i)$. Compared to the regret \eqref{eq:regret}, $ \mathsf{EBRegret}_n$ 
additionally integrates over $\mu_{1:n} \mid z_{1:n}
\sim P$ and chooses the benchmark as
$\one (\E_P[\mu_i \mid Y_i, z_i] > k_i)$. When the class of decision rules $\mathcal D$ is
exactly the class of Bayes decisions, controlling
\eqref{eq:regret} automatically controls $\mathsf{EBRegret}_n$: \[
  \mathsf{EBRegret}_n = \E_{\substack{\mu_{1:n} \sim P \\ Y_i \mid \mu_i, z_i \sim \Norm
      (\mu_i,
  \sigma_i^2)}}\bk{W(P) - u(\hat P)} \le \E_{\mu_{1:n}\sim P}\bk{\regret_n}.
\]

When the prior is misspecified, however, $\mathsf{EBRegret}$ ceases to be relevant. This is exactly where \assure{*} provides value. To
illustrate, in \cref {ex:robbins,fig:Robbins_skew_example}, we consider using a
misspecified Gaussian prior under which $\mu_i \iid
\Norm(m_0, s_0^2)$. The implied selection rule, given by $\one\pr{\E[\mu_i \mid Y_i] =
\frac{s_0^2} {s_0^2+1}Y_i
+ \frac{m_0} {s_0^2+1} > 0}$, is therefore also misspecified. An empirical Bayesian may estimate $
(m_0,s_0^2)$ by maximum likelihood.\footnote{Interestingly, in \cref{ex:robbins}, this
yields
the
same selection decisions as choosing $(m_0, s_0^2)$ by SURE.} \Cref{fig:Robbins_skew_example} shows that this EB-MLE/SURE approach underperforms \assure*. 
This is because the misspecification
of the prior \emph{misaligns} optimal selection and estimation decisions \citep
{manski2021econometrics}.\footnote{To be sure, some of the misspecification  can be remedied by nonparametric modeling
\citep{jiang2009general}. Even with a nonparametric empirical Bayes model, one could
continue to use \assure{} to robustly evaluate and further fine-tune decisions. We discuss an implementation in \cref{sec:extensions}.}

\section{Theoretical Assurances}
\label{sec:theory}

This section presents statistical guarantees for \assure{*}. Our main results are upper and
matching lower bounds for {regret} \eqref{eq:regret}. We first present upper and
lower bounds achieving $\tilde O(1/\sqrt{n})$ regret, which hold under weak assumptions.
Additional assumptions on $\mu_{1:n}$ allow for a faster rate by exploiting an analogue
of the margin condition
\citep{audibert2005fast}.

\subsection{Main regret bound}
Our main result is the following $(1/\sqrt{n})$-regret bound. To state the theorem,    for $q \ge 1$, define  $m_q:= \norm{m}_{q,n} := \pr{\frac{1}{n} \sum_{i=1}^n |\mu_i - k_i|^q}^{1/q}$.



{\begin{restatable}[$\tilde O(1/\sqrt{n})$-regret for \assure{*}]{theorem}{thmmain}
  \label{thm:main_regret_bound}

Suppose that the class of decision rules $\br{\delta(z;\beta): \beta \in \mathcal B
\subset \R^d}$ is VC-subgraph with
index $V(\cl{D})$ and envelope function $D$. Let $\hat{\beta} \in
\argmax_{\beta\in\mathcal B} \hat{W}(\beta)$ be the optimal decision chosen by \assure{*}. 
If $ \sigma_{1:n}, \sigma_{1:n}^{-1}$ and the envelope $D$
are all uniformly bounded by an absolute constant, then for some further constant $C$,  \[
 \regret_n \lesssim
  C V(\cl{D})(m_2 + m_4^4)\frac{(\log n)^{1/4}\sqrt{\log\log n}}{\sqrt n}.
\]

\end{restatable}
}

Thus, further assuming that $|\mu_i - k_i|, V(\mathcal D)$ are uniformly bounded, \assure* achieves $\frac{(\log n)^{1/4} \sqrt{\log \log n}}{\sqrt{n}}$ regret. We check in \cref{asub:vc_checks} that the leading decision rule classes
enumerated in \cref{sec:examples} are VC-subgraph.

Broadly speaking, \cref{thm:main_regret_bound} follows a standard empirical risk
minimization argument. Regret can be controlled by controlling the empirical process
$\sup_{\beta\in\mathcal B} |\hat W(\beta) - W(\beta)|$, which is manageable as $\mathcal D$ is assumed to
be VC-subgraph. There are two complications in our setup that require care. First,
standard empirical risk minimization~ (e.g., Theorem 8.4.4 in \citet{vershynin2009high})
controls---in our notation---the gap $W(\beta^*) -
\E[W(\hat\beta)]$. This is a different quantity than \eqref{eq:regret} since $\E[W
(\hat\beta)]\neq \E[u(\hat\beta)]$. Second, since $(z_i, \mu_i)$ are treated as
constants, the random variables $Y_i$ are independent but not identically distributed,
complicating standard i.i.d. empirical process theory.

\subsection{Lower bound}
\label{sec:lower_bound}

\Cref{thm:main_regret_bound} is unimprovable  in the worst case, up to log factors.

\begin{restatable}[Lower bounds for regret]{theorem}{thmmatchinglb}
  \label{thm:matching_lb}

  Fix $\sigma_i = 1$ and $k_i = 0$ for all $i$. Without loss of generality, assume
  $\delta(z_i;\beta) = \beta
  \in \bR$. Let $M > 0$, then
  over \emph{all} decision rules $a_i(Y_1,\ldots, Y_n) \in [0,1]$, for a constant $C = C_M$,
  \[
    \inf_{a(\cdot) : a_i \in [0,1]} \sup_{\mu_{1:n} \in [-1,1]^n} \regret_n(a_{1:n}) \ge \frac{C}{\sqrt{n}},
  \]
  where
  \begin{equation}
    \label{eq:gen_regret}
    \regret_n(a_{1:n}) := \sup_{|\beta| \le
    M}
    W (\beta; \mu_
    {1:n}) -
    \frac{1}{n}\sum_{i=1}^n \mu_i \E_{\mu_{1:n}}[a_i(Y_1,\ldots,
    Y_n)].
  \end{equation}
\end{restatable}

Here, the regret notion in \cref{thm:matching_lb} compares the welfare of an arbitrary
randomized selection decisions $a_i(Y_1,\dots,Y_n)$ to the welfare of the best \emph{threshold
decision} restricted to $|\beta| \leq M$. Since \eqref{eq:gen_regret} removes constraints
on the decision maker's actions but imposes constraints on the optimal benchmark, lower
bounds of \eqref{eq:gen_regret} lower bound regret.\footnote{Since the setting in
\cref{thm:matching_lb} does not have contextual information, the optimal separable
decision rule does take the threshold form.  An analogous regret notion was considered by
\cite{polyanskiy2021sharp} in their (12) for the mean square case.} 

\Cref{thm:matching_lb}
is derived by an application of Le Cam's two-point method, considering cases in which, for
some $h > 0$, either all $\mu_ {i} = h/ \sqrt{n}$ or all $\mu_i = -h/\sqrt{n}$.\footnote{An analogous construction is applicable to an empirical Bayes setup, which
partially resolves a conjecture regarding the tightness of Theorem 4 in
\citet{chen2022empirical}.} In this case, the optimal $\beta^*$ is either $+M$ or $-M$,
meaning that $W(\cdot)$ is maximized on the boundary of the parameter space. In
fact, when $\beta^*$ is a local maximum, the upper and lower bounds can be improved to
$\tilde O(1/n)$. We turn to this faster rate upper bound next.

\subsection{Fast rates for \assure{*}}
\label{sec:fast_rates}

As the lower bound hints, the performance for \assure{*} can be better than the $\tilde{O}
(n^ {-1/2})$ rate predicts, when the population optimal decision is interior. 
The additional assumptions needed  are reminiscent of the \textit{margin condition} for
fast rates in classification and treatment choice
\citep{audibert2005fast,audibert2007fast,kitagawa2018should,ponomarev2024lower,crippa2025regret}.\footnote{There, the margin condition bounds the density of difficult examples near the optimal
classification boundary. Intuitively, with many units near the classification
boundary, it
is difficult to distinguish between different candidate parameter values, meaning that the
objective function lacks a cleanly separated maximizer. This is exactly the failure mode
that \cref{thm:matching_lb} exploits for the lower bound, which we rule out in
\cref{as:fast_rate_assn}.}

For simplicity, we will restrict to a case where the decision rules are parametrized by a
scalar parameter ($\mathcal{B} \subset \R$). We leave the multidimensional extension to future work.
\begin{as}
  \label{as:fast_rate_assn}
  \begin{asenum}
  \item \label{assmp:boundedness} (Boundedness). Fix an absolute constant $B > 1$.
    Assume that $\mu_i,k_i \in [-B, B]$ and $\sigma_i \in [B^{-1}, B]$. Furthermore,
    assume that $\delta(\cdot, z_i)$ are thrice-differentiable for each $i$ and $\delta
    (\cdot, z_i)$ and its three derivatives are uniformly bounded by $B$.

  \item \label{assmp:unimodal} (Unimodality). There exists a unique local and global
    maximizer $\beta^*$ in the interior of $\mathcal{B}$ at which $W(\beta^*) > 0$. 

  \item \label{assmp:curvature} (Curvature). Fix some absolute $\kappa>0$. $W''(\beta^*) < -\kappa$.

  \item \label{assmp:well_sep_max} (Well-separated maximum). Fix some absolute $\xi > 0$. Then
    $
    \sup_{|\beta - \beta^*| > \frac{\kappa}{B}} W(\beta) < W(\beta^*) - \xi.
    $
  \end{asenum}
\end{as}

\begin{restatable}[Fast Rates for \assure{*}]{theorem}{thmfastrates}
  \label{thm:fast_rates}

  Suppose \cref{as:fast_rate_assn} holds. Lastly, suppose that $\delta$ and its derivatives
  $\delta(\beta,z), \delta'(\beta,z), \delta''(\beta,z)$ form VC-subgraph function classes $\mathcal{D}, \mathcal D',
  \mathcal D''$ with indices $V(\cl{D})$,
  $V(\cl{D}')$, and $V(\cl{D}'')$, respectively. Then the regret of \assure{*}
  satisfies
  \[
    \E\left[W(\beta^*) - u(\hat{\beta})\right] \le C \frac{(\log n)^6}{n}.
  \]
  for constant $C > 0$ that depends only on $B, \kappa, \xi$ in \cref{as:fast_rate_assn}
  and $V(\cl{D}'), V(\cl{D}''),
  \mathcal B$.
\end{restatable}

\Cref{thm:fast_rates} exploits the fact that $\hat\beta$ is a local maximum, which
enables the approximation $W (\hat\beta)\approx W (\beta^*) + \frac{1}{2}W'' (\beta^*) (\hat\beta -
\beta^*)^2$. Analogously Taylor approximating $0 = \hat W'(\hat\beta)$ shows 
that $\hat\beta - \beta^* = O_P(|\hat W' - W'|)$; therefore, regret is approximately
$O_P(|\hat W' - W'|^2)$. Since it can be shown that \assure{*} also estimates $W'$ at
$\tilde O(n^{-1/2})$ rates, we obtain the rate in
\cref{thm:fast_rates}. \Cref{sec:example_fast_rate_class} illustrates a class of decision rules and lower-level
assumptions on $\mu_{1:n}$ for which \cref {thm:fast_rates} applies. The log factors in \cref{thm:fast_rates} are crude.

Finally, we present a lower bound for the fast rate in \cref{thm:fast_rates}.

\begin{restatable}[Fast Regret Lower Bounds]{theorem}{propfastratelb}
  \label{prop:gaussian_matching_lb_fast_rate}
  Fix $\sigma_i = 1$ and $k_i = 0$.  Consider decision rules $\delta(z_i;\beta) = \beta$ for  $\beta \in \bR$. Let $\Theta = \Theta_{n,\kappa,\xi} \subseteq
  \bR^n$ be the space of $(\mu_1,\dots,\mu_n)$ satisfying the assumptions
  \crefrange{assmp:boundedness}{assmp:well_sep_max} for a fixed $\xi, \kappa$. Then there
  exists choices of $\xi$ and $\kappa$ such that
  \[
    \inf_{a(\cdot) : a_i \in \set{0,1}} \sup_{\mu_{1:n} \in \Theta } \regret_n(a_{1:n}) =
    \Omega(n^{-1}),
  \]
  with $\regret_n(a_{1:n})$ defined as in \eqref{eq:gen_regret}.
\end{restatable}

The proof of \cref{prop:gaussian_matching_lb_fast_rate} lower bounds the
compound regret by the Bayes regret for a well-chosen prior and then further lower bounds
the Bayes regret by Le Cam's two point argument.\footnote{Section 3.1
  of
  \citet{liang2004optimal} states a $\Omega(n^{-1})$ lower bound for a related but distinct
  regret quantity, where the decision rule is learned from data points $Y_1,\dots,Y_n$ and
  evaluated on a newly drawn decision problem $(\mu_{n+1},Y_{n+1})$. We are concerned with
  an ``in-sample'' version of \citet{liang2004optimal}'s regret, motivating a distinct proof
strategy.}

\section{Extensions}
\label{sec:extensions}

This section presents two extensions. We first consider robustness to Gaussian
approximation. We then outline how \assure{*} can be applied to more complex decision
rules via leave-one-out cross-validation.

\subsection{Robustness to Gaussianity}
\label{sec:robustness_to_gaussianity}

So far, we have operated under the Gaussian sequence model \eqref
{eq:normal_model}. Gaussianity of $Y_i$ is motivated by applications where the noisy
estimates are averages of micro-data, for which the central limit theorem justifies a
Gaussian approximation. However, since \assure{} methods rely on the Gaussianity of
$Y_i$,\footnote{The same would be true of using SURE to tune decision rules that are
  nonlinear in $Y_i$. If decision rules are restricted to be linear in $Y_i$, the quadratic
nature of SURE implies that consistent estimation of certain moments is sufficient.} a
reasonable worry is whether our methods are sensitive to the quality of the Gaussian
approximation.

This section thus considers $Y_i$ explicitly as averages of i.i.d. random variables, and
analyzes a feasible version of \assure{*} as an estimator of a rescaled version of
welfare \`a la \citet{hirano2009asymptotics}. Our focus is on welfare estimation here,
but we expect these results translate to regret control like \cref
{thm:main_regret_bound,thm:fast_rates}. We find that the behavior of \assure{*} degrades
smoothly in the quality of the Gaussian approximation. Moreover, because \assure{*}
exploits approximate Gaussianity of $Y_i$, it outperforms simple estimators of welfare
that do not exploit this information.

We consider a stylized data-generating process where the micro-data $Y_{ij} \iid G_i^{
(m)}$, for heterogeneous distributions $G_1^{(m)},\ldots, G_n^{(m)}$ and $j=1,\ldots,
m$.\footnote{Extending this setup to heterogeneous cell-sizes $m_i$ is notationally
cumbersome but straightforward.} The decision-maker  observes the sample mean of cell
$i$, $\bar Y_i := \frac{1}{m} \sum_{j=1}^m Y_{ij}$.  For simplicity, assume that costs
are zero $k_i =
0$. We will consider a local asymptotic regime similar to \citet
{hirano2009asymptotics} and \citet{chen2026normal}, where the micro-data distribution $G^{(m)}_i$ has mean $\mu_i^{
(m)} := \mu_i/\sqrt{m}$ and standard deviation $\sigma_i$. 

Let $F_i^{(m)}(t):= P_{G_i^{(m)}}({(\bar Y_i - \mu_i^{(m)})/\sigma_i \le t })$ be the CDF of
the null $t$-statistic and let $\delta_i =
\sqrt{m} \delta_i^{(m)}$ be a rescaled
version of the threshold. This rescaling allows us to rescale the welfare and write it as: 
\begin{align*}
     W(\delta) &:=
  \frac{1}{n} \sum_{i=1}^n \underbrace{(\sqrt{m} \cdot \mu_i^{(m)})}_{\text{scaling by
  $\sqrt m$}} \E_{G_i^{(m)}}[\one(\bar Y_i > \delta_i^{(m)})] =
  \frac{1}{n} \sum_{i=1}^n \mu_i  \bar F_i^{(m)} \pr{
    \frac{\delta_i - \mu_i}{\sigma_i}
  },
   \numberthis \label{eq:finite-sample-welfare}\\ \text{ for } & \bar F_i^{(m)}(t)  := 1 - F_i^{(m)}(t)
\end{align*}
The welfare \eqref{eq:finite-sample-welfare} is the analogue of \eqref{eq:welfare};
here,
the selection probability is calculated via the actual
finite-sample distribution $G_i^{(m)}$ rather than its Gaussian approximation. Relative
to the parameters $\mu_i^{(m)}$ in their native scale, \eqref{eq:finite-sample-welfare}
scales welfare by $\sqrt{m}$, accentuating differences at the $1/\sqrt{m}$ scale.

For fixed $\delta$, we study the estimation performance of a feasible version of 
\assure{*} for \eqref{eq:finite-sample-welfare} and compare it to reasonable alternatives.
Let \[
  \hat\sigma_i^2 = \frac{1}{m-1} \sum_{j=1}^m (Y_{ij} - \bar Y_i)^2
\]
be the estimated sample variance and let $Y_i := \sqrt{m} \bar Y_i$. The feasible
version of \assure{*} \eqref{eq:assure_term} is then, for $h = 1/\sqrt{2\log n}$,
\begin{equation}\label{eq:hatw_hi}
  \hat W(\delta) = \frac{1}{n} \sum_{i=1}^n \hat w_{h,i}(Y_i; \delta) :=
  \frac{1} {n} \sum_{i=1}^n Y_i \Csinc\pr{\frac{Y_i-\delta_i}{\hat\sigma_i h}} - \frac{\hat\sigma_i}{h} \sinc \pr{\frac{Y_i-\delta_i}{\hat\sigma_i h}}.
\end{equation}

\begin{theorem}[\assure{*} bias with non-Gaussian micro-data]
  \label{thm:assure_nongaussian_microdata} Under the preceding setup, fix
  $\delta_1,\ldots, \delta_n$. Suppose $G_i^{
  (m)}$ have mean $\mu_i/\sqrt{m}$, variance $\sigma_i^2$, and third and fourth absolute
  central moments $\mu_{3,i}, \mu_{4,i}.$ The bias of \assure{*} satisfies
  \[
    \left| \frac{1}{n}\sum_{i=1}^n \E_{G_i^{(m)}}\left[\hat{w}_{h,i}(\sqrt{m}\bar{Y}_i;\delta)
    \right] - W(\delta) \right| \leq C\left(\frac{\log(n)^2}{\sqrt{m}} + \frac{1}{n\log
    n}\right),
  \]
  where the constant $C$ depends on $\min_i \sigma_i$ and averages of polynomials in $|\mu_i|,\sigma_i, \mu_{3,i},\mu_{4,i}.$
\end{theorem}

By independence and the bound $\Var(\hat w_{h,i}(Y_i; \delta)) = O((\log n)^2)$, we
have that $\var(\hat W(\delta)) = O\pr{
\frac{(\log n)^2}{n}}.$ The mean-squared error of \assure{*} as an estimator of $W
(\delta)$, for fixed decisions, is thus
\[
  \E[(\hat W - W)^2] = O\left( \frac{\log(n)^4}{m} + \frac{\log(n)^2}{n}
  \right).
\]

The $(\log n)^4/m$ term reflects the fact that the Gaussianity of $Y_i$ is only
approximate, which disappears smoothly as $m \to \infty$ and $Y_i$ becomes Gaussian. At
this scaling, this term also
reflects the benefit of exploiting the approximate Gaussianity. To see this, we can
contrast with the mean-squared error of a simple estimator that does not use this
information: Consider an
estimator that simply replaces $\mu_i$ with $\bar Y_i$:\footnote{Maximizing $\hat W_{\mathrm{CES}}$ over fixed thresholds $\delta_i = \beta$  results in
the empirical success rule $\one(\bar Y_i > 0)$ \citep{manski2004statistical}.}
\[\hat{W}_{\mathrm{CES}}(\delta) =\frac{1}{n} \sum_{i=1}^n \sqrt{m} \bar Y_i \mathbf{1}(
    \bar{Y}_i >
\delta_i^{(m)}).\]
  Under mild
assumptions, it is not difficult to show that the bias at this scale does not
vanish,\footnote{Similar bias and mean-squared error properties apply to $
  \hat{W}_{\text{plug-in}}(\delta) := \frac{1}{n}\sum_{i=1}^n Y_i \Phi\left(
  \frac{Y_i - \delta}{\hat \sigma_i} \right).
$}
leaving the mean-squared error no better than $1+1/n:$
\[
  |\E[\hat W_{\mathrm{CES}}] - W| = \Omega(1) \quad \E[(\hat W_{\mathrm{CES}} - W)^2] =
  \Omega
  \pr{1 + \frac{1}{n}}.
\]
Therefore, as an estimator for the welfare \eqref
{eq:finite-sample-welfare}, \assure{*} strictly dominates these welfare estimators
whenever $m \gg
\log(n)^4$ and $(m, n) \to \infty.$

\subsection{Ensembling complex decision rules}
\label{sec:loo_fitted_decisions}

In our second extension, we discuss how \assure{} can enable practitioners  to
\emph{ensemble} multiple decision rules. For instance, one decision rule  may come from a
nonparametric EB model without modeling $z$-dependence---in which $\mu_i \mid z_i \iid
G_1$ for some $G_1$ modeled nonparametrically, perhaps estimated through nonparametric
maximum likelihood. Another may come from a parametric EB model that considers contextual
information---$\mu_i \mid z_i \sim G_2(\cdot \mid z_i)$, where $G_2(\cdot)$ is parametric
\textsc{close-gauss} and estimated via plug-in methods. Yet another method may simply try
to estimate $\E[Y_i \mid z_i] = \E[\mu_i \mid z_i]$ flexibly.

This is made tractable with leave-one-out cross-fitting. For the example above, for each
$i=1,\ldots,n,$ let $m_{\hat G_1}^{(-i)}(y, z)$ be a posterior mean function fitted on
the data $ (Y_ {-i},z_{-i})$. Let $m_{\hat G_2}^{(-i)}(y, z)$ denote the analogous
leave-one-out posterior mean for a second empirical Bayes method. Finally, let $\hat{f}^{
(-i)}$ denote a regression model trained using the data $(Y_{-i},z_{-i})$. We may
consider the class of decision rules which selects a unit whenever the ensembled
prediction is greater than the implementation cost:
\begin{equation}
  \label{eq:ensemble_class}
  b_1 m_{\hat G_1}^{(-i)}(Y_i,z_i) + b_2 m_{\hat G_2}^{(-i)}(Y_i,z_i) + b_3 \hat{f}^{
  (-i)}(z_i) \geq k_i,
\end{equation}
where the parameters $\beta = (b_1,b_2,b_3)$ are restricted to be in the unit simplex.
By the monotonicity of the posterior mean function for Gaussian observations, the decision
\eqref{eq:ensemble_class} is equivalent to $Y_i \geq
\delta_{\ensemble}(z_{1:n},Y_{-i};\beta)$ for some function $\delta_{\ensemble}$. \assure
{} can be used to tune the parameter $\beta$.

To that end, let $\delta(z_{1:n},Y_{-i};\beta)$ be some parametrized threshold that
depends on all $z_{1:n}$ and other units $Y_{-i}$. Following \eqref{eq:assure_term}, let
\[
  w_h(Y_i; \beta,z_i,Y_{-i}) :=
  (Y_i - k_i) \Csinc\pr{
    \frac{Y_i - \delta(z_{1:n},Y_{-i};\beta)}{\sigma_i h}
  } - \frac{\sigma_i}{h} \sinc\pr{
    \frac{Y_i - \delta(z_{1:n},Y_{-i};\beta)}{\sigma_i h}
  }.
\]
\assure{*} then estimates welfare by averaging $w_h(Y_i; \beta,z_i,Y_{-i})$ across $i$.
Then, by the law of iterated expectations conditioning on $Y_{-i}$, we have that
\assure{*} remains a low-bias estimator for the welfare of this ensemble decision:
\begin{align*}
  \E\hat W(\beta) := \frac{1}{n} \sum_{i=1}^n  \E w_h(Y_i; \beta,z_i,Y_{-i}) & =\frac{1}{n}
  \sum_ {i=1}^n \E\left[\E\left[w_h(Y_i; \beta,z_i,Y_{-i}) \ | \ Y_{-i} \right] \right] \\
  & \approx \frac{1}{n} \sum_{i=1}^n (\mu_i - k_i) \E\left[ \Phi\left(\frac{\mu_i -
  \delta(z_{1:n},Y_{-i}; \beta)}{\sigma_i} \right) \right]. \tag{\cref{thm:assure_bias}}
\end{align*}
The approximation in the last line is suggested by \cref{thm:assure_bias}.

Under mild stability assumptions, \assure{} is expected to perform no worse asymptotically
than the constituent models since the decision class nests each model individually. This
ensemble class thus enables \assure{} to improve on or robustify any set of methods. We
defer a detailed theoretical analysis of the ensembling method for future work.

\section{Simulation}
\label{sec:simulations}

To illustrate \assure{} methods as robustifying empirical Bayes, we consider a setup where
$\mu_{1:n}$ do come from some random effect distribution, but we misspecify this
distribution. In particular, we consider\footnote{The marginal
  distribution of $\log \sigma_i$ is drawn from $\Norm(0.5, 0.35^2)$, truncated to $
[-1.15, 1.15]$.} \[
  \mu_i \mid \sigma_i \sim \Norm(m_0(\log \sigma_i), 0.2) \quad m_0(x) = -\frac{x}{2} + 3 e^
  {-\frac{1}{2} \pr{\frac{x-1}{0.12}}^2} \quad \log\sigma_i \in [-1.15, 1.15]. \numberthis
  \label{eq:nonlinear_dgp}
\] Assume costs $k_i$ are zero. The class of decision rules we consider is \[
  \mathcal D = \br{\delta(\sigma_i; \beta) = -\sigma_i^2 (a + b \log \sigma_i): (a,b) \in
  \R^2}.
\] $\mathcal D$ is motivated by the misspecified model where we model $m_0$ as linear:
\[\mu_i
  \mid \sigma_i \sim \Norm
  (\tilde a + \tilde b
\log \sigma_i, s_0^2), \numberthis \label{eq:linear_mean_eb}\] for parameters $(\tilde a,
\tilde b, s_0^2)$. In particular, under this empirical Bayes model, since costs are zero,
the posterior decision rules are $Y_i > -\sigma_i^2
\frac{\tilde a + \tilde b \log \sigma_i} {s_0^2}$, and only $a=\tilde a/s_0^2$ and
$b=\tilde b/s_0^2$ are relevant for parametrizing thresholds.

\begin{figure}[htb]
  \includegraphics[width=\textwidth]
  {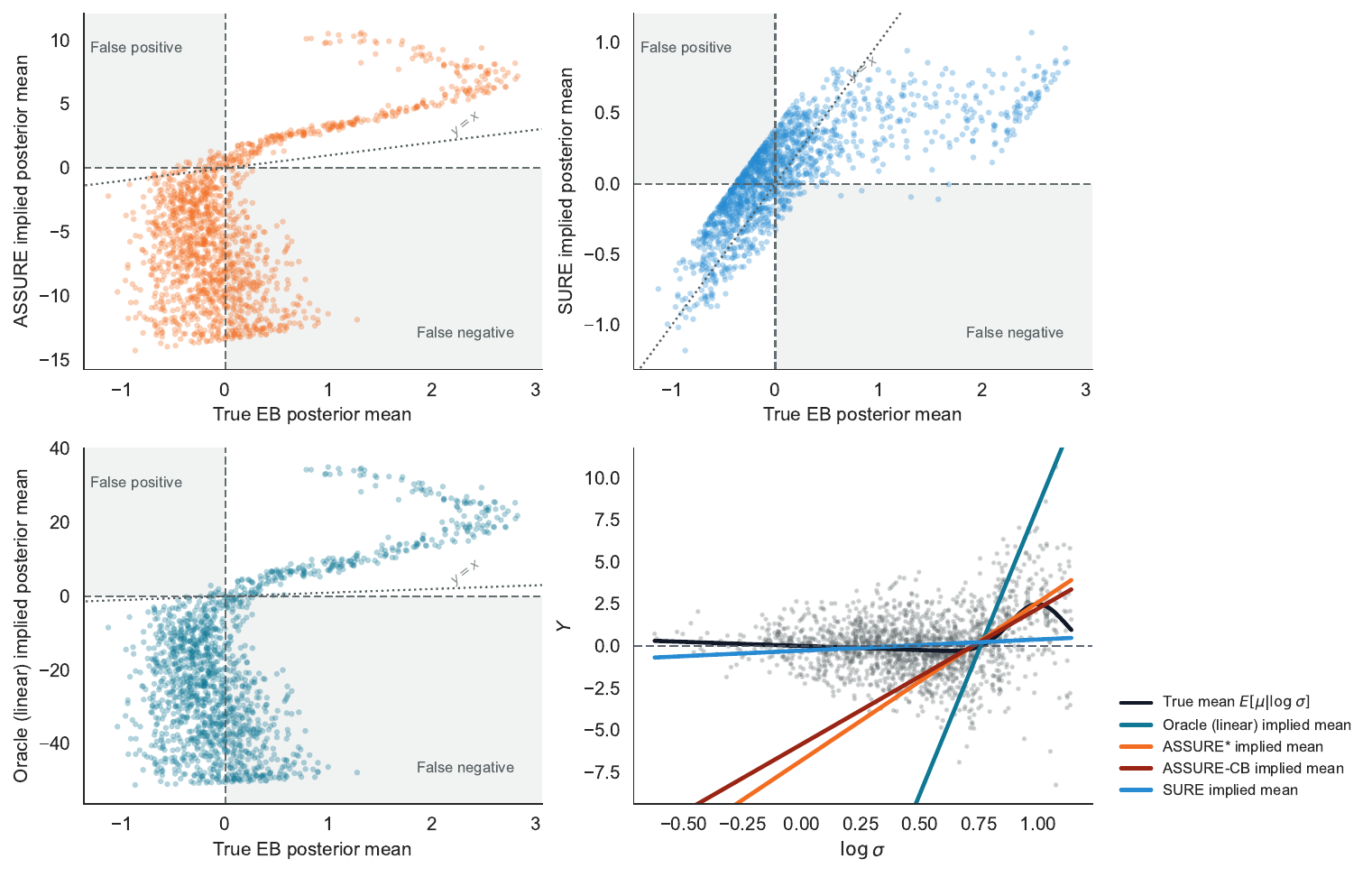}

  \caption{Results from the DGP \eqref{eq:nonlinear_dgp} with a misspecified prior.}

  \begin{proof}[Notes] We show a single draw from \eqref{eq:nonlinear_dgp} with $n = 1500$.
    The lower-right panel plots \(Y_i\) against \(\log \sigma_i\) for one sample with
    \(n=1500\), together with the true conditional mean \(m_0(\log \sigma_i)\) and the linear empirical-Bayes mean functions implied by SURE, \assure
    {*}, \assurecb, and the oracle welfare \eqref{eq:general_welfare} maximizer over
    \(\mathcal D\). The other three panels compare each method's implied posterior mean
    with the true posterior mean under the data-generating process; the dashed line
    indicates the 45 degree line. \assurecb{} produces very similar posterior means as
    \assure{*} and is thus omitted.
  \end{proof}
  \label{fig:nonlinear_every}
\end{figure}

\begin{table}[htb]
  \centering
  \small
  \setlength{\tabcolsep}{4pt}
  \resizebox{\textwidth}{!}{%
    \begin{tabular}{@{}lccccccc@{}}
      \toprule
      & \shortstack{True posterior\\mean}
      & \shortstack{Oracle\\over $\mathcal D$}
      & \assure{*}
      & \assurecb{}
      & SURE
      & \shortstack{Plug-in\\empirical Bayes}
      & $\one(Y_i > 0)$ \\
      \midrule
      Average utility
      & 0.295
      & 0.278
      & 0.267
      & 0.2668
      & 0.227
      & 0.228
      & 0.201 \\
      \bottomrule
    \end{tabular}
  }
  \caption{Average utility across 1000 draws from \eqref{eq:nonlinear_dgp} with $n=1500$.}
  \label{tab:nonlinear_dgp_utility}

  \begin{proof}[Notes]
    $(\mu, \sigma, Y_i)$ are redrawn. Average utility is \eqref{eq:welfare} averaged over 1000
    draws.
    Plug-in empirical Bayes fits $\tilde a, \tilde b$ with ordinary least-squares of $Y_i$
    on $\log \sigma_i$. It then fits $s_0^2$ by averaging the squared OLS residuals and
    subtracting by the average of $\sigma_i^2$.
  \end{proof}
\end{table}
For a draw of $Y_i$, we fit decision rules over $\mathcal D$ with \assure{*}, \assurecb,
and by maximizing the oracle welfare function \eqref{eq:general_welfare}. We likewise fit
$\tilde a,\tilde b, s_0^2$ in \eqref{eq:linear_mean_eb} through SURE or plug-in empirical
Bayes.\footnote{Plug-in estimation for \eqref{eq:linear_mean_eb} estimates $\tilde a,
  \tilde b$ via OLS of $Y_i$ on a constant and $\log \sigma_i$. It then estimates $s_0^2$ by
  taking the difference in the residual variance of this regression and the average
$\sigma_i^2$.} \Cref {tab:nonlinear_dgp_utility} shows the average utility across 1000
draws. We see from \cref{tab:nonlinear_dgp_utility} that \assure{} methods substantially
improve on SURE or plug-in empirical Bayes. Normalizing the performance of $\one(Y_i > 0)$
to zero and oracle over $\mathcal D$ to one, \assure* has relative performance 0.857 while
SURE and plug-in have relative performances only about 0.35.

To unpack why \assure{*} has good performance, it is helpful to compare results in terms of implied EB posterior means. We thus convert the threshold estimates $(a,b)$ into implied
$\tilde a,
\tilde b$ by fixing the prior variance at its true value $s_0^2 = 0.2$. Under this
normalization, the corresponding EB coefficients are $\tilde a = 0.2a$, $\tilde b = 0.2
b$. \Cref{fig:nonlinear_every} visualizes different methods on this empirical Bayes
scale for a single draw of the data.

The bottom-right of \cref{fig:nonlinear_every} shows a plot of $Y_i$ against $\log \sigma_i$
with the implied prior mean estimates of various methods. SURE effectively fits a
weighted squared-error approximation to the true $m_0(\log \sigma_i)$. On the other
hand,
\assure{} and oracle directly optimize for threshold rules, thereby resulting in implied
$\E [\mu
\mid \log\sigma]$ estimates that have high fidelity to where the conditional mean
function $m_0 (\log\sigma)$ crosses zero.

We can visualize the gain of \assure{} by plotting the implied posterior mean estimates
against the true posterior mean under \eqref{eq:nonlinear_dgp} in the other three panels
of \cref{fig:nonlinear_every}. Indeed, SURE results in a good squared error
approximation of the true posterior mean, aligning well with the 45-degree line.
However, because it does not tune parameters for the selection objective, it incurs many
false positives---where the SURE estimates of $\mu_i$ are positive but true posterior
mean is negative. On the other hand, tuning decision rules implied by the linear model
through the oracle objective \eqref{eq:general_welfare} trades off false positives and
false negatives more efficiently, resulting in higher welfare than SURE. \assure
{} mimics the behavior of the oracle---its implied posterior means are inaccurate
predictions of $\mu_i$, but lead to better selections.

\begin{figure}[htb]
  \centering
  \includegraphics[width=\textwidth]
  {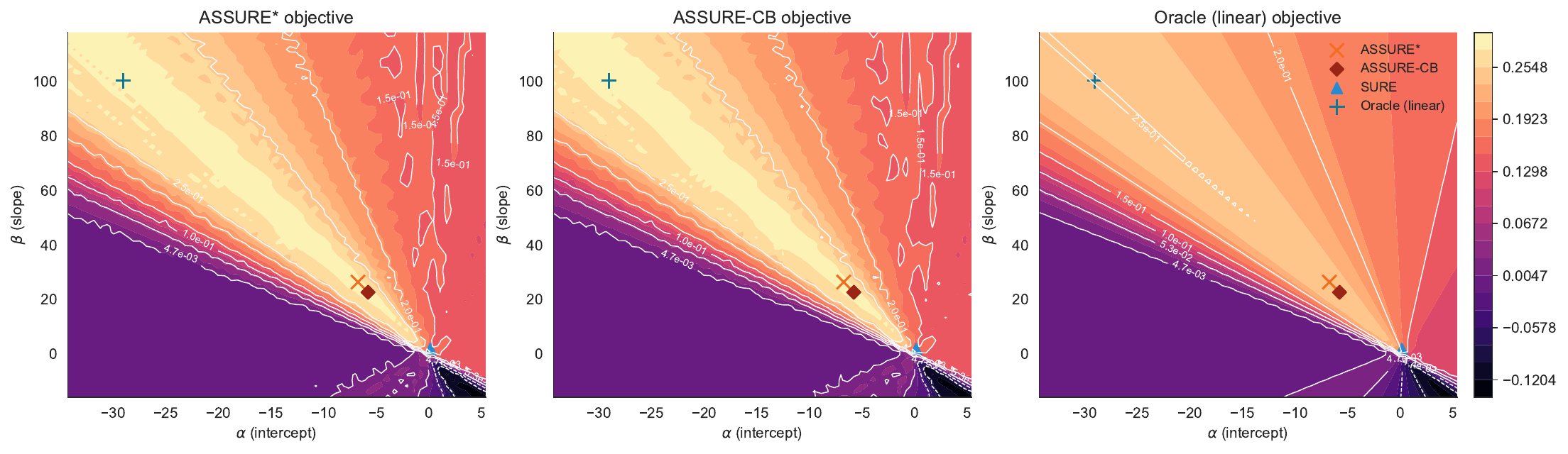}
  \caption{Objective function landscape of \eqref{eq:nonlinear_dgp}.}
  \label{fig:optimization_diagnostics}
\end{figure}

Lastly, we can examine whether the \assure{*} objective function is made ill-behaved by
the oscillations in the sinc kernel.  \Cref{fig:optimization_diagnostics} shows the
objective function  for \assure{*}, \assurecb, and the true welfare function \eqref
{eq:welfare_expected} for one draw of data. Both \assure{} methods do result in nonconvex
objective function surfaces, possibly motivating sampling-based methods \`a
la \citet{chernozhukov2003mcmc}. However, it does not appear that the higher-order $\sinc$ kernel
\emph{per se} results in worse optimization landscape than the Gaussian kernel.

\section{Empirical Applications}
\label{sec:empirical_application}

We illustrate \assure{} in three empirical applications in economics: selecting census
tracts to maximize economic mobility, identifying discrimination among large firms,
and selecting innovations in experimentation programs. The first two applications show
that plug-in linear-shrinkage empirical Bayes methods can be suboptimal for welfare
when the underlying model for $P$ is misspecified, and that \assure{} improves
robustness in those settings. The last application instead uses
\assure{} to assess hypothesis-testing rules for A/B testing.

\subsection{The Opportunity Atlas.}
\label{sec:oa_application}

\begin{figure}[!htb]
  \centering

  \includegraphics[width=\textwidth]
  {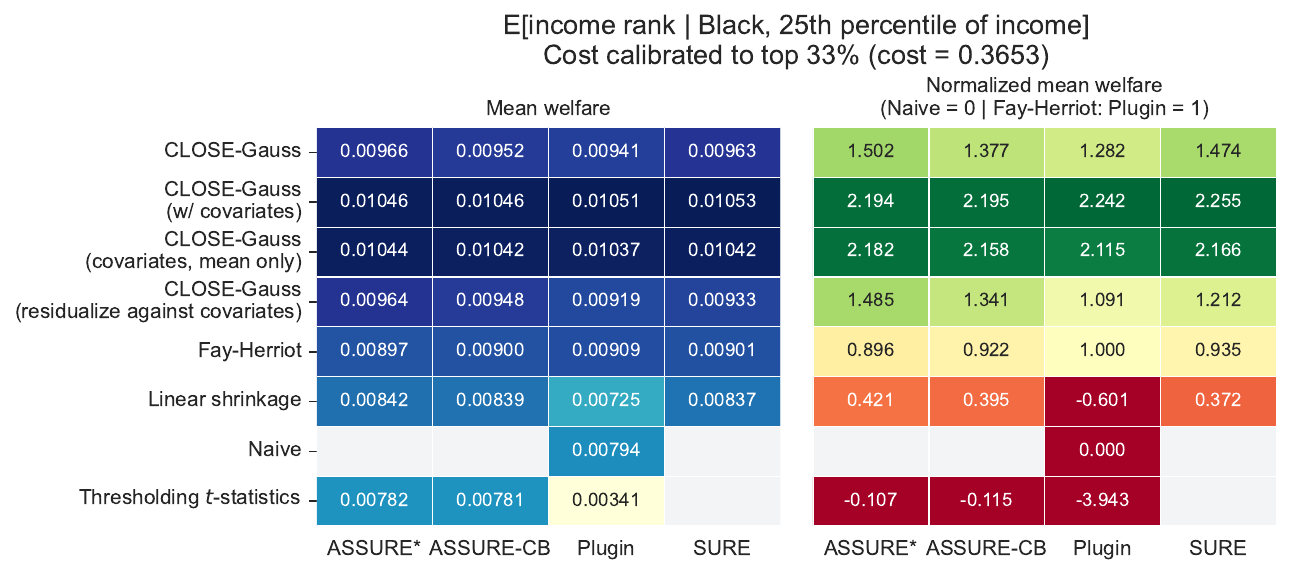}

  \includegraphics[width=\textwidth]
  {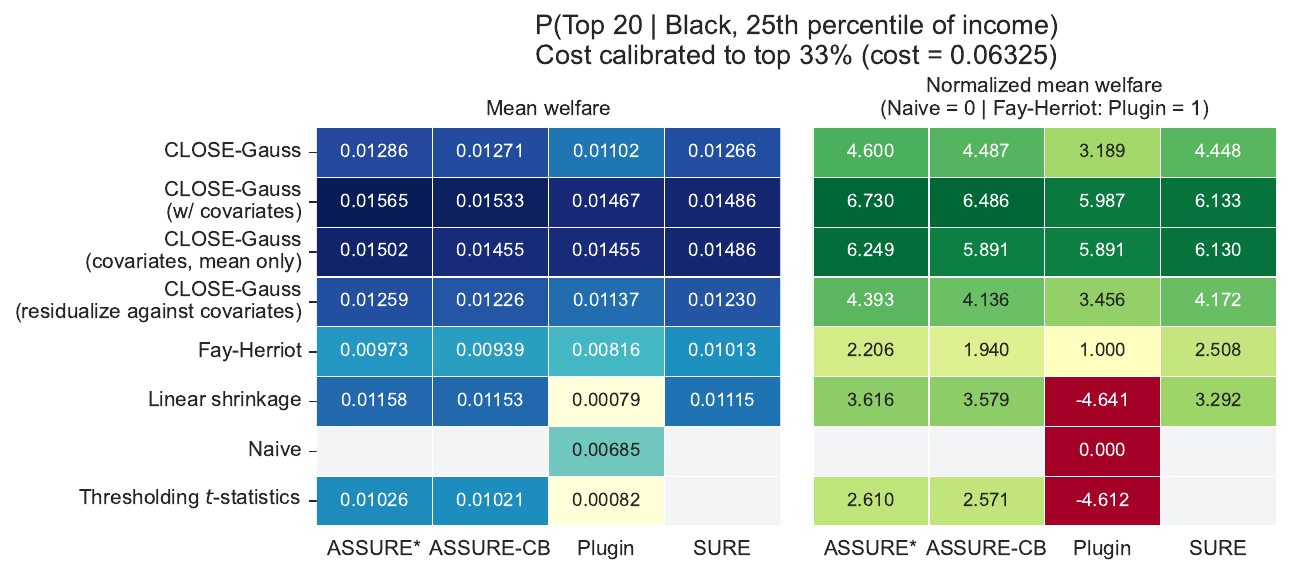}
  \caption{Absolute and relative welfare for two different economic mobility measures in
  \citet{bergman2024creating} application.}
  \label{fig:bergman_app}

  \begin{proof}[Notes] Results are averaged over 1000 coupled-bootstrap draws. ``Naive''
    is the decision rule that simply selects $\one(Y_i > k)$. \closegauss{} ignores
    covariates and models $\mu \mid \sigma_i$. In addition, we consider three variants of
    \closegauss
    {} with covariates. The first variant (\closegauss ~w/ covariates) includes covariates
    in both prior mean and prior variance symmetrically as with $\log\sigma_i$. The
    second variant (\closegauss ~w/ covariates, mean only) includes covariates in the
    prior mean, but models the prior variance as completely constant. The third variant
    (\closegauss, residualize against covariates) residualizes $Y_i$ linearly against the
    covariates, and runs \closegauss{} on the residuals. In relative comparisons, we
    normalize the performance of naive to $0$ and that of plug-in Fay-Herriot to $1$.
  \end{proof}
\end{figure}

We begin with the setting studied by \citet{bergman2024creating}, which can be cast
naturally as a compound selection problem. Their goal is to identify high-mobility
Census tracts for recommending to low-income households, using noisy estimates from
\citet{chetty2018opportunity}. Our implementation follows \citet
{chen2022empirical}, who studies a national version of this application using
nonparametric empirical Bayes methods.

To cast this decision problem in our framework, let $Y_i$ denote the estimated economic
mobility of a Census tract $i$, defined either as mean income rank or as probability of
high income for Black individuals growing up in relatively poor households in tract $i$.
Denote true economic mobility by $\mu_i$, and the estimated standard error of $Y_i$
in \citet{chetty2018opportunity} as $\sigma_i$.

\citet{bergman2024creating} are interested in selecting the top third of census tracts. We interpret this targeting rule as implicitly defining a constant
implementation cost $k_i = k$. To calibrate \(k\), we compute posterior means under
Fay--Herriot and let \(k_{\mathrm{calibrated}}\) be their 66th percentile. Under this
calibration, an empirical Bayesian is indifferent between selecting and not selecting the
marginal tract in the top third. We fix this calibrated cost throughout.

To evaluate the performance of different selection decisions, following \citet
{chen2022empirical}, we employ coupled bootstrap
\eqref{eq:cb} to mimic a $90/10$ split by synthetically constructing a holdout sample that
provides an unbiased estimate of
performance on the compound decision problem with inflated variances. We consider a few
instances of decision rules discussed in \cref{tab:close-gauss-fam,sec:examples}:
(i) Linear shrinkage, (ii) Fay--Herriot,\,(iii) Parametric
\closegauss, (iv) Parametric \closegauss{} with covariates and its
variants, and (v) Thresholding $t$-statistics. See \cref{sec:additional_numerical_results}
for implementation details.

For each class of decision rules, we then consider fitting their parameters with
\assure{*}, \assurecb, SURE, and plug-in estimation. Plug-in estimation for
\closegauss{} methods involves method-of-moment estimation of the prior parameters.
Plug-in estimation for thresholding $t$-statistics simply chooses the conventional
threshold $\beta = 1.96$.

\Cref{fig:bergman_app} summarizes results. The left panels report absolute welfare and the
 right panels report welfare normalized so that the performance of $\one(Y_i > k)$
 achieves zero relative welfare and the plug-in Fay--Herriot decision achieves relative
 welfare equal to 1.  For flexible empirical Bayes models that are arguably better
 specified, all methods perform similarly: \assure{} methods do not generate large
 improvements, though they do not appear to hurt. In those cases, \assure{} certifies
 that plug-in empirical Bayes methods, though they optimize on average across $\mu_
 {1:n} \sim P$, do well on a particular instance of the parameters. On the other hand, for
 less flexible
 decision classes, \assure{*} and
\assurecb{} noticeably improve over plug-in and SURE tuning---again illustrating the robustness gains from targeting welfare directly.

\subsection{Discrimination in large firms.} \label{sec:firms_application}

\citet{kline2022systemic} conducted a large correspondence experiment to measure the
 contact gap $\mu_i$ between white and Black job applicants across $n=108$ large U.S.
 firms. Specifically, fictitious applications were sent to approximately 125 job listings
 per firm. The difference in callback rates between races was then averaged across jobs
 within a firm to produce the estimate $Y_i$.

Using a procedure that controls the false discovery rate at the 5\% level, 
\citet{kline2022systemic}
identify 23 firms that discriminate against Black applicants. Under the ``intensive-margin'' preference in
their paper, using these estimates to select firms for further investigation maps
directly into utility \eqref{eq:welfare}, where $k$ is the cost of an investigation.
If we calibrate $k$ using the nonparametric empirical Bayes benchmark from
\citet{kline2022systemic}, adapted from \citet{efron2016empirical}, the implied value
is $k=0.025$, which rationalizes the selection of 23 firms.

To illustrate the performance of \assure{*}, consider the class of linear shrinkage
decision rules (\cref{tab:close-gauss-fam}).  With constant $k$, linear shrinkage
implied thresholds simplify to $\delta(\sigma_i, \mu_0, s_0^2) = k +\sigma_i^2 \frac{k
-m_0} {s_0^2} =: k + \sigma_i^2 \beta$, and thus we can let $\beta := (k-m_0)/s_0^2$ be a
scalar parameter. \Cref{fig:discrimination_white_assure}(a) shows that directly plugging
in the MLE estimates for $(m_0, s_0)$ would give a higher threshold and select fewer
firms than \assure{*}.

\begin{figure}[ht]
  \centering

  \begin{subfigure}[t]{0.48\textwidth}
    \centering
    \includegraphics[width=\textwidth]{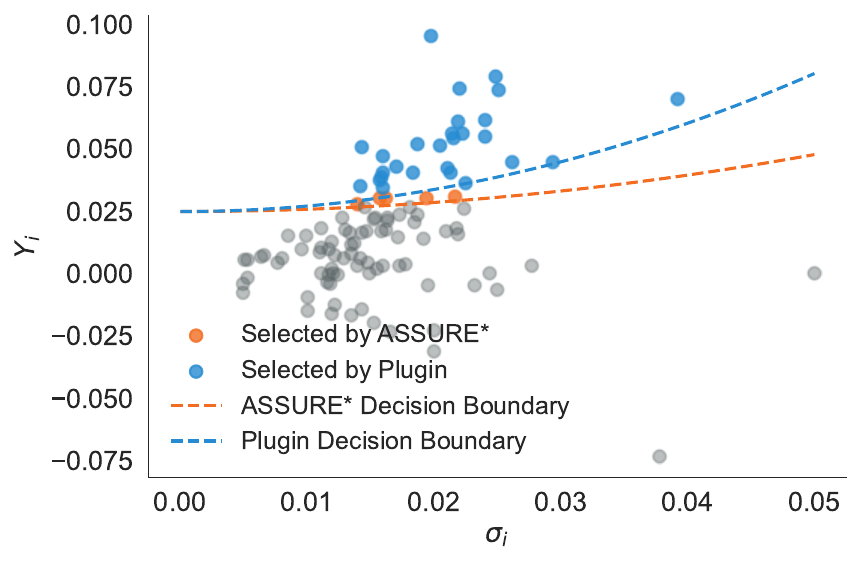}
    \phantomsubcaption
    \caption*{\normalfont(a) \assure{*} vs. plug-in linear shrinkage}
    \label{fig:firm_risk_male}
  \end{subfigure}
  \hfill
  \begin{subfigure}[t]{0.48\textwidth}
    \centering
    \includegraphics[width=\textwidth]{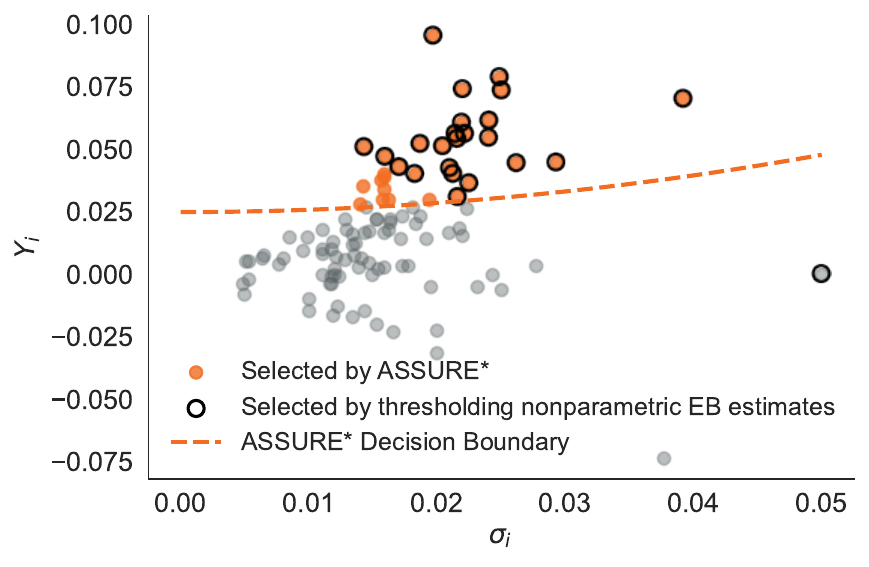}

    \phantomsubcaption
    \caption*{\normalfont(b) \assure{*} vs nonparametric EB}
    \label{fig:firm_risk_white}
  \end{subfigure}
  \hfill

  \caption{Selections from $n=108$ firms studied by \citet{kline2022systemic}}\label{fig:discrimination_white_assure}
\end{figure}

While \citet{kline2022systemic} reported linear shrinkage estimates, their preferred model
is a nonparametric empirical Bayes model as shown in
\cref{fig:discrimination_white_assure}(b). This is precisely due to concerns that the
parametric Gaussian prior corresponding to the linear shrinkage rule is misspecified. For
instance, the underlying distribution of true contact gaps may be skewed. As a result, the
posterior mean tends to underestimate large $\mu_i$, making the shrinkage rule
conservative. \assure{*} can partly correct for this and partially salvage linear
shrinkage rules. In a cross-validation exercise over ten splits (three waves to estimate
the threshold and two to estimate the resulting welfare), the plug-in linear shrinkage
rule achieves welfare 39\% below that of thresholding based on the nonparametric EB
estimates, whereas \assure{*} is only 25\% below.

\subsection{Experimentation programs.}
\label{sec:experimentation_program_application}

In the technology industry, \emph{experimentation programs} are portfolios of
A/B tests that evaluate interventions affecting a business metric of interest
\citep{azevedo2020b, sudijono2024optimizing, chou2025evaluating}. For experiment $i$, the
average treatment
effect estimate $Y_i$ is plausibly normally distributed around the
true average treatment effect $\mu_i$, measured in terms of the business metric.  We
use an  anonymized Netflix dataset from
\citep{sudijono2024optimizing}, containing treatment-effect estimates and standard
deviations
for $331$ feature experiments. \Cref{fig:experimentation_program_viz} summarizes the
data.

\begin{figure}
  \centering
  \includegraphics[width=0.8\linewidth]
  {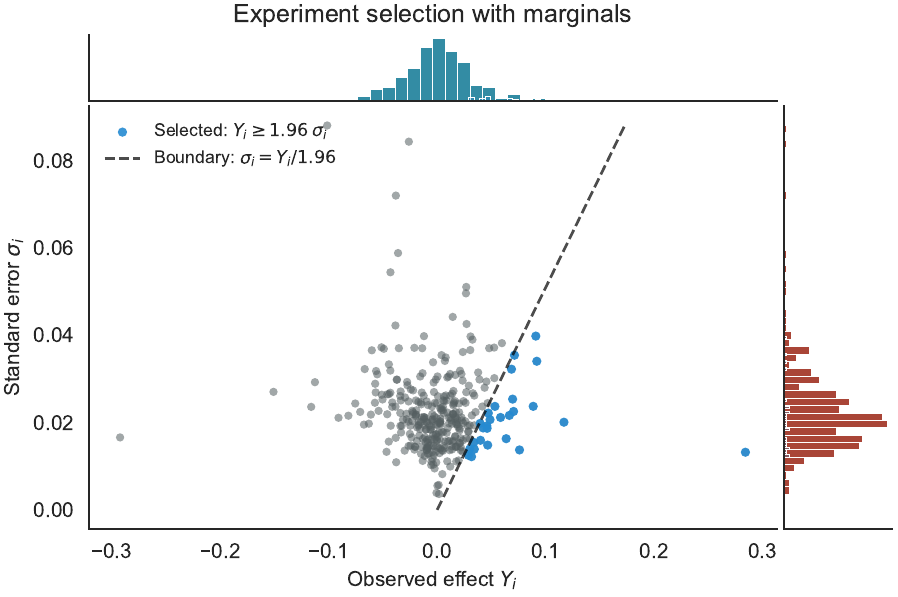}
  \caption{Experimentation-program data with the conventional significance threshold}
  \begin{proof}[Notes]
    The center panel plots observed treatment effects against their standard errors for
    roughly 330 feature experiments. Blue points satisfy
    the conventional rule $Y_i \geq 1.96 \sigma_i$, and the dashed line traces this decision
    boundary. The top and right margins show the empirical distributions of observed effects
    and standard errors.
  \end{proof}

  \label{fig:experimentation_program_viz}
\end{figure}

We consider selecting for interventions to implement in the platform as a compound
selection problem. A common ad hoc decision procedure selects via the
$t$-statistic decision class in \cref{sec:examples}, $\mathcal D = \br{\delta(z;\beta) = k
  + \sigma
\beta}$. Conventional practice chooses the cutoff to be $t = \beta = 1.96$,
equivalently a two-sided $p$-value of $0.05$. We instead use \assure{*} to choose the
threshold $\beta$ that maximizes
the expected welfare of the decision, taking $k_i = 0.$
\Cref{fig:welfare_curve_example_experimentation_programs} shows the \assure{*} estimate
$\hat{W}(\beta)$. By maximizing the curve $\hat{W}(\beta)$, \assure{*} suggests using a
decision corresponding to $\hat\beta \approx 0$.

Because the data are anonymized (both the treatment-effect estimates and their standard
deviations are obfuscated by random multiplicative factors), the optimal decision cannot
be interpreted directly. Even so, the shape of the estimated welfare curve is
informative: If implementation is nearly costless, the preferred threshold is much more
lenient than the conventional ($t=1.96$) rule, echoing
\citet{azevedo2020b,berman2022false,sudijono2024optimizing}.

\begin{figure}
  \centering
  \includegraphics[width=0.75\linewidth]{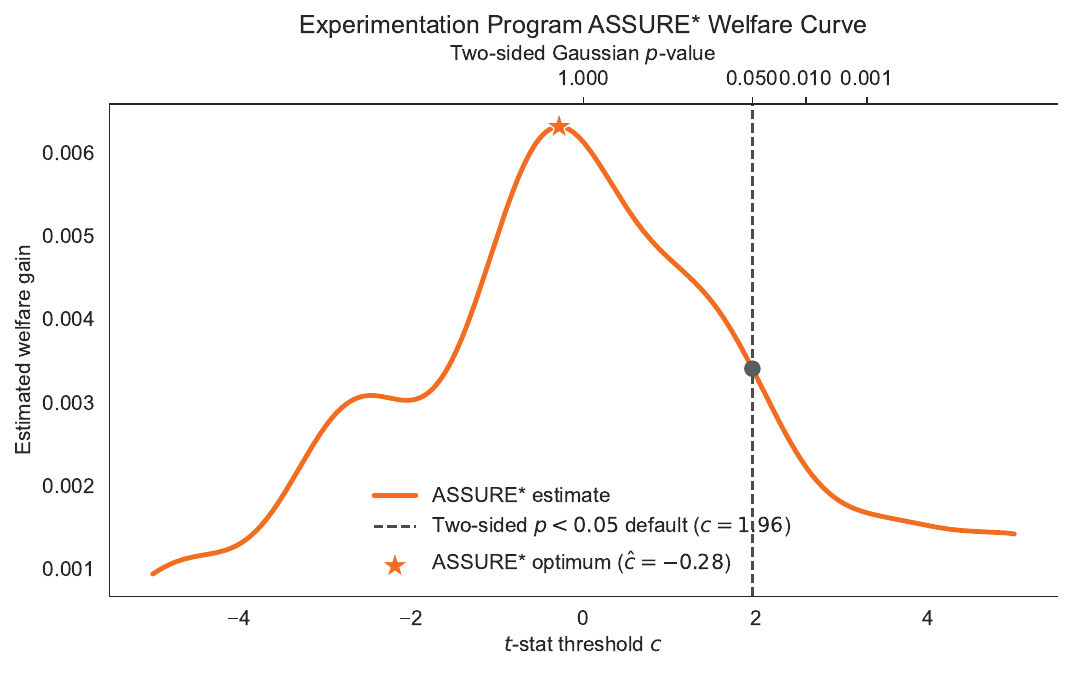}
  \caption{\assure{*} welfare curve for $t$-statistic thresholds.}
  \begin{proof}[Notes]
    The solid curve plots the estimated welfare $\hat W(c)$ for decision rules of the
    form $Y_i \geq c \sigma_i$. The star marks the \assure{*}-optimal threshold, while
    the dashed vertical line marks the conventional benchmark $c = 1.96$. The top axis
    reports the corresponding two-sided Gaussian $p$-values.
  \end{proof}
  \label{fig:welfare_curve_example_experimentation_programs}
\end{figure}
\begin{figure}
  \centering
  \includegraphics[width=0.99\linewidth]{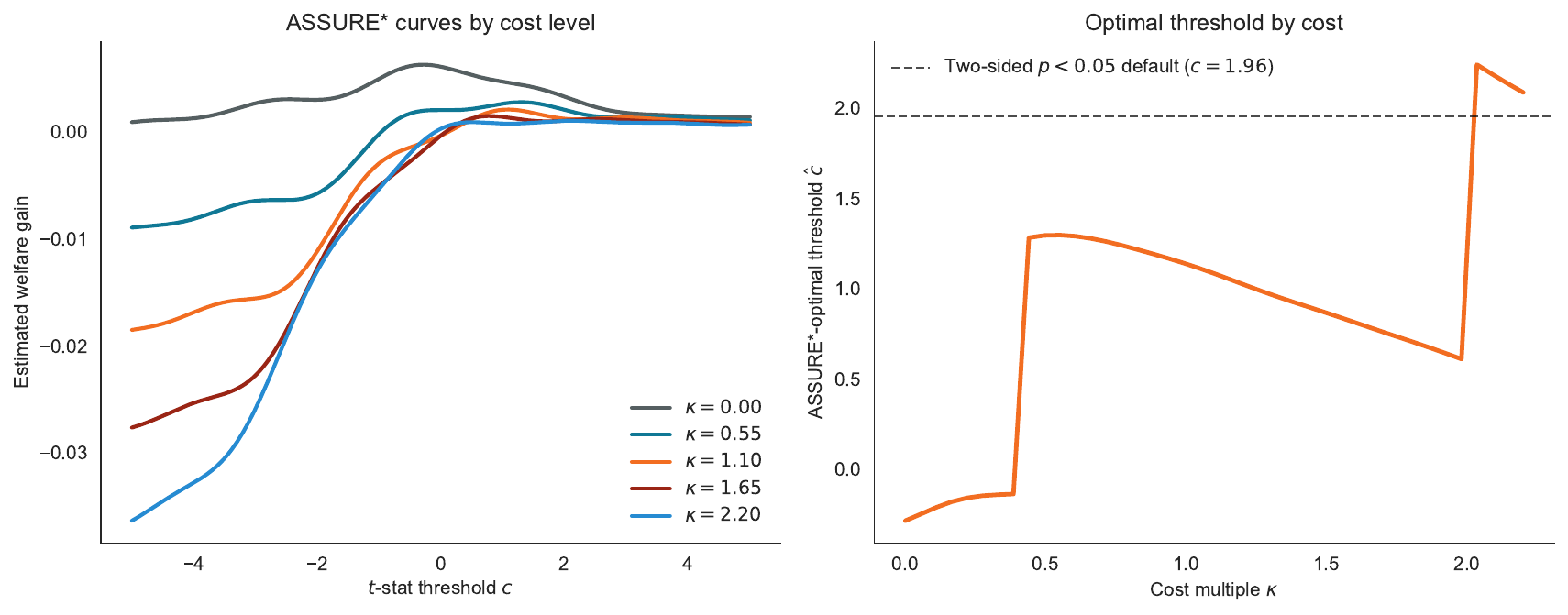}
  \caption{\assure{*} welfare and optimal thresholds by cost level.}
  \caption*{\footnotesize Notes: The left panel shows estimated welfare curves for decision rules $Y_i \geq k + c \sigma_i$ under several constant implementation costs $k$. The right panel plots the resulting \assure{*}-optimal threshold $c$ as costs vary, where costs are normalized as multiples of the median estimated treatment-effect magnitude in the dataset. The dashed horizontal line marks the conventional benchmark $c = 1.96$.}
  \label{fig:example_exp_program_cost_analysis}
\end{figure}

The \assure{} welfare estimate can also be used in reverse: It lets practitioners infer
the cost level under which the conventional $t=1.96$ rule would be approximately
optimal. To illustrate this, \cref{fig:example_exp_program_cost_analysis} scales the
implementation cost to be $\kappa$ times the median estimated treatment-effect size in
the dataset. The left panel plots $\hat W(\beta)$ for different values of $\kappa$,
and the right panel plots the corresponding \assure{}-optimal threshold. As $\kappa$
rises, the preferred rule becomes more conservative. The figure suggests that the
conventional $t=1.96$ benchmark is consistent with implementation costs roughly twice
the median treatment effect.

\section{Conclusion}
\label{sec:discussion}

This paper introduces \assure, a data-driven approach for compound selection decisions,
closely related to SURE. In particular, \assure{*} is near-unbiased for the expected
utility of selection problems for fixed $\mu_{1:n}$, under weak assumptions, producing
robust evaluations of the quality of selection decisions. Because it accurately estimates
welfare, optimizing \assure{*} also leads to attractive selection decisions. We present both
theoretical guarantees and practical applications illustrating these advantages.

\bibliographystyle{aer}
\bibliography{main}


\appendix

\numberwithin{lemma}{section}
\numberwithin{theorem}{section}
\numberwithin{corollary}{section}
\numberwithin{proposition}{section}
\numberwithin{as}{section}
\numberwithin{figure}{section}
\numberwithin{equation}{section}

\begin{appendices}

\end{appendices}

\section{Proofs}

Throughout, define
\begin{align*}
  \Psi_h(Y_i; z_i, C) &:= (Y_i - k_i) \Csinc\pr{
    \frac{Y_i - C}{\sigma_i h}
  } - \frac{\sigma_i}{h} \sinc\pr{
    \frac{Y_i - C}{\sigma_i h}
  } \\
  &= (Y_i - k_i) \br{\frac{1}{2} + \frac{1}{\pi} \Si\pr{
      \frac{Y_i - C}{\sigma_i h}
  }} - \frac{\sigma_i}{h} \sinc\pr{
    \frac{Y_i - C}{\sigma_i h}
  }
  \numberthis. \label{eq:assure_term_Psi}
\end{align*}

\begin{proof}[\textbf{Proof of \cref{thm:assure_bias}}]
  We first invoke \cref{lemma:expected_assure}, where a calculation with the well-known
  Dirichlet integral \[
    \int_0^\infty \frac{\sin ax}{x} dx = \lim_{h \downarrow 0} \int_0^{1/h}\frac{\sin ax}{x} dx
  \]
  shows that
  \begin{align*}
    (\mu-k) \Phi\left(\frac{\mu - C}{\sigma} \right) & = \frac{1}{2} (\mu - k) + \frac{1}{\pi} (\mu - k) \int_0^\infty  \frac{1}{\omega} e^{-\frac{1}{2}\omega^2} \sin(\omega (\mu - C)/\sigma)d\omega \\
    \E \Psi_h(Y,z,C) & = \frac{1}{2} (\mu - k) + \frac{1}{\pi} (\mu - k) \int_0^{1/h}\frac{1}{\omega} e^{-\frac{1}{2}\omega^2} \sin(\omega (\mu - C)/\sigma)d\omega.
  \end{align*}

  Taking the difference,
  \[
    (\mu-k) \Phi\left( \frac{\mu - C}{\sigma}\right) - \E \Psi_h(Y,z,C) = \frac{1}{\pi}\int_{1/h}^\infty \frac{1}{\omega} e^{-\frac{1}{2}\omega^2} (\mu-k) \sin(\omega(\mu - C)/\sigma) d\omega,
  \]
  which has an upper bound given by
  \begin{align*}
    \frac{|\mu - k|}{\pi}\int_{1/h}^\infty \frac{1}{\omega} e^{-\frac{1}{2}\omega^2}
    d\omega & \leq \frac{h|\mu - k|}{\pi} \int_{1/h}^\infty e^{-\frac{1}{2}\omega^2} d\omega  \le h^2|\mu - k| e^{-1/(2h^2)}.
  \end{align*}
  The last line follows from the Mills ratio bound, which holds uniformly over $C,\sigma$.
  As a result, the bias for $\hat{W}_{h}$ is bounded above by
  $
  \left(\frac{1}{n}\sum_{i=1}^n |\mu_i - k_i| \right) h^2e^{-1/(2h^2)},
  $ by letting $C_i = \delta(z_i; \beta)$.

  Next, we establish the second claim pointwise in $\beta$. Set
  $\lambda_n:=1/h=\sqrt{2\log n}$ and $C_i:=\delta(z_i;\beta)$. Using
  $\Csinc(x)=1/2+\pi^{-1}\Si(x)$, write the $i$th summand as
  \[
    w_h(Y_i;z_i,\beta)
    =\frac{1}{2}(Y_i-k_i)
    +\frac{1}{\pi}(Y_i-k_i)\Si\left(\lambda_n\frac{Y_i-C_i}{\sigma_i}\right)
    -\sigma_i\lambda_n\sinc\left(\lambda_n\frac{Y_i-C_i}{\sigma_i}\right).
  \]
  Since $\Si$ is bounded, the second moment of the first two terms is bounded by a
  constant multiple of $\E(Y_i-k_i)^2=\sigma_i^2+(\mu_i-k_i)^2$. For the last term, 
  \cref{lemma:sinc_second_moment_bound} calculates
  \[
    \E\left[\left(\sigma_i\lambda_n\sinc\left(\lambda_n\frac{Y_i-C_i}{\sigma_i}\right)\right)^2\right]
    \lesssim \lambda_n\sigma_i^2.
  \]
  Therefore,
  \[
    \Var(w_h(Y_i;z_i,\beta))
    \lesssim (\mu_i-k_i)^2+\lambda_n\sigma_i^2,
  \]
  where we used $\lambda_n\ge 1$ to absorb the remaining $\sigma_i^2$ term. By independence,
  \[
    \Var(\hat W_n(\beta))
    \lesssim
    \frac{1}{n}\left(\frac{1}{n}\sum_{i=1}^n(\mu_i-k_i)^2
    +\lambda_n\frac{1}{n}\sum_{i=1}^n\sigma_i^2\right).
  \]
  The bias bound above, evaluated at $h=1/\sqrt{2\log n}$, gives
  \[
    \abs{\E \hat W_n(\beta)-W(\beta)}
    \lesssim \frac{1}{n\log n}\left(\frac{1}{n}\sum_{i=1}^n|\mu_i-k_i|\right).
  \]
  By Jensen's inequality, the square of this bias bound is dominated by
  $n^{-1}(n^{-1}\sum_{i=1}^n(\mu_i-k_i)^2)$. Combining these displays with the
  bias-variance decomposition proves the claimed pointwise mean-squared error bound.
\end{proof}

\medskip
\begin{proof}[\textbf{Proof of \cref{thm:main_regret_bound}}]
  Let $\beta^*\in\argmax_{\beta\in\mathcal B}W(\beta)$.
  The optimality of
  $\hat\beta$ for $\hat{W}_n$ implies $\hat{W}_n(\hat{\beta})\ge \hat{W}_n(\beta^*)$. Recalling $u(\beta)$ from \eqref{eq:u_beta}, we have  $\E[u(\beta^*)] = W(\beta^*)$ and
  \begin{align}
    \regret_n
    & = W(\beta^*) - \E[u(\hat{\beta})] = \E[u(\beta^*) - u(\hat{\beta})]\nonumber\\
    & \le \E\left[u(\beta^*) - \hat{W}_n(\beta^*) + \hat{W}_n(\hat{\beta}) - u(\hat{\beta})\right]\le  2\E\bk{\sup_{\beta\in\mathcal B}
    |u(\beta) - \hat{W}_n(\beta)|}\nonumber\\
    & \le 2\E\bk{\sup_{\beta\in\mathcal B}
    |\hat W(\beta) - W(\beta)|} + 2\E\bk{\sup_{\beta\in\mathcal B}
    |u(\beta) - W(\beta)|}.
    \label{eq:regret_decompose_main_bound}
  \end{align}

  Starting from \eqref{eq:regret_decompose_main_bound}, we use the bias bound in \cref
  {thm:assure_bias}(1) to obtain:
  \[
    \eqref{eq:regret_decompose_main_bound} \lesssim \E\left[\sup_{\beta\in\mathcal B} |u(\beta) - W
    (\beta)|\right] + \E\left[\sup_{\beta\in\mathcal B} |\hat{W}(\beta) -
    \E\hat{W}(\beta)|\right] + \underbrace{\frac{m_1}{n \log n}}_{\text{ by \cref{thm:assure_bias} }}.
  \]
  Letting
  $A_1(\beta)  := u(\beta) - W(\beta)$,
  $A_2(\beta)  := \hat{W}(\beta) - \E\hat{W}(\beta),$
  we will next individually bound $\E \sup_{\beta\in\mathcal B} |A_1(\beta)|, \E \sup_{\beta\in\mathcal B} |A_2(\beta)|$.
  We defer the bounds for these two terms to
  \cref
  {prop:realized_regret_term_empirical_process_bound,prop:assure_term_empirical_process_bound},
  which prove, 
  if $ \sigma_{1:n}, \sigma_{1:n}^{-1}$ and the envelope $D$
are all uniformly bounded by an absolute constant, that
  \begin{align*}
    \E \sup_{\beta\in\mathcal B} |A_1(\beta)| & \lesssim m_2\frac{\sqrt{V(\cl{D})}}{\sqrt{n}}   \\
    \E \sup_{\beta\in\mathcal B} |A_2(\beta)| & \lesssim
    C V(\cl{D})(m_2 + m_4^4)\frac{(\log n)^{1/4} \sqrt{\log \log n}}{\sqrt{n}},
  \end{align*}
 for some constant $C$ that depends only on $\sigma_{1:n}$ and $D$. Therefore
 \[
    \regret_n \le \eqref{eq:regret_decompose_main_bound}
    \lesssim C V(\cl{D})(m_2 + m_4^4) \frac{(\log n)^{1/4}\sqrt{\log \log n}}{\sqrt{n}}.
    \qedhere
  \]

\end{proof}

\begin{proof}[\textbf{Proof of \cref{thm:matching_lb}}]
  Fix some $h > 0$ and let $\mu_+ = (h/\sqrt{n},\ldots, h/\sqrt{n})'$ and $\mu_- = -\mu_+$.
  For any selection decision $a_1,\ldots, a_n$, let $g_+ = \frac{1}{n} \sum_
  {i=1}^n \E_
  {\mu_+} [a_i(Y_1,\ldots, Y_n)]$ and define $g_-$ analogously. Then the regret at $\mu_+$ is
  \begin{align*}
    &\sup_{|\beta| < M} W(\beta; \mu_
    {+}) -
    \frac{1}{n}\sum_{i=1}^n \mu_i \E_{Y_{1:n} \sim \Norm(\mu_+, I_n)}[a_i(Y_1,\ldots,
    Y_n)] = \frac{h}{\sqrt{n}} \pr{\Phi\pr{
        \frac{h}{\sqrt{n}} +M
    } - g_+}
  \end{align*}
  Analogously, the regret at $\mu_-$ is \[
    -\frac{h}{\sqrt{n}} \pr{\Phi\pr{-\frac{h}{\sqrt{n}} - M} - g_-} = \frac{h}{\sqrt{n}}
    \pr{
      \Phi\pr{
      \frac{h}{\sqrt{n}} + M} - 1 + g_-
    }
  \]
  Thus, the maximum regret over $\mu_+$ and $\mu_-$ is \[
    \frac{h}{\sqrt{n}} \pr{\Phi\pr{
    \frac{h}{\sqrt{n}} + M} - 1 + \max(1-g_+, g_-)} \numberthis
    \label{eq:max_regret_two_point_N}
  \]
  Thus, for any choice of selection decisions $a_{1:n}$, \[
    \sup_{\mu_{1:n} \in [-1,1]^n} \br{\sup_{|\beta| < M}
      W (\beta; \mu_
      {1:n}) -
      \frac{1}{n}\sum_{i=1}^n \mu_i \E_{Y_{1:n} \sim \Norm(\mu_{1:n}, I_n)}[a_i(Y_1,\ldots,
    Y_n)]} \ge \eqref{eq:max_regret_two_point_N}.
  \]

  Let $\varphi(Y_1,\ldots, Y_n) = \frac{1}{n} \sum_i a_i(Y_1,\ldots, Y_n) \in [0,1]$. When
  viewed as a test against $H_0: \mu = \mu_-$ for $H_1 : \mu = \mu_+$, $g_-$ is its type I
  error and $1-g_+$ is its type II error. Then \[
    \max(1-g_+, g_-) \ge \frac{(1-g_+) + (g_-)}{2} \ge \frac{1}{2}\pr{1-\mathrm{TV}\pr{
        \Norm(\mu_+, I_n), \Norm(\mu_-, I_n)
    }}.
  \]
  By Pinsker's inequality, $
  \mathrm{TV}\pr{
    \Norm(\mu_+, I_n), \Norm(\mu_-, I_n)
  } \le h.
  $ Thus, \[
    \eqref{eq:max_regret_two_point_N} \ge \frac{h}{\sqrt{n}} \pr{
      \Phi\pr{
      \frac{h}{\sqrt{n}} + M} - 1 + \frac{1}{2} - \frac{h}{2}
    }\ge \frac{h}{\sqrt{n}} \pr{
      c_M - \frac{h}{2}
    } \gtrsim_M \frac{1}{\sqrt{n}}
  \]
  for $c_M = \Phi(M) - 1/2 > 0$ and the choice $h = c_M$. This completes the proof.
\end{proof}

\medskip
\begin{proof}[\textbf{Proof of \cref{thm:fast_rates}}]
  Decompose the regret as
  \begin{align*}
    \regret_n(\hat{\beta}) & = \E\left[W(\beta^*) - W(\hat{\beta})\right] + \E \left[ W(\hat{\beta}) -  u(\hat{\beta})\right].
  \end{align*}
  \Cref{thm:small_external_regret} shows that the first term is  $O((\log n)^6 n^{-1})$. To handle the second term $\E \left[ W(\hat{\beta}) -  u(\hat{\beta})\right]$, we use a
  leave-one-out stability argument. Define $Y_{-i}$ to be the vector of observations without $Y_i$ and
  \begin{equation}
    \hat{\beta}^{(-i)} = \argmax_\beta \sum_{j\neq i} g_n(Y_j,z_j,\beta),
  \end{equation}
  where \[g_n(y,z,\beta) = \frac{(y - k)}{2} + \frac{(y - k)}{\pi} \Si\left(\frac{\lambda_n(y -
    \delta (\beta,z))}{\sigma} \right) - \sigma \lambda_n \sinc\left(\frac{\lambda_n(y -
    \delta (\beta,z))}{\sigma} \right)
  \]
  are the summands in the \assure{*} estimator. We will show that for
  each $i$, $ \hat{\beta}^{(-i)}$ only differs from $\hat{\beta}$ by $\tilde{O}_P
  (1/n).$ To explain why this will suffice, first write
  \begin{align*}
    W(\hat{\beta}) - \E u(\hat{\beta}) &= W(\hat{\beta}) - \frac{1}{n}\sum_{i=1}^n (\mu_i - k_i)\Prob(Y_i \geq \delta(\hat{\beta}^{(-i)},z_i)) \\
    & \qquad + \frac{1}{n}\sum_{i=1}^n (\mu_i - k_i)\left(\Prob(Y_i \geq \delta(
    \hat{\beta}^{(-i)},z_i)) - \Prob(Y_i \geq \delta(\hat{\beta},z_i)) \right).
  \end{align*}
  Thus
  \begin{align}
    \label{eq:line_1_loo_stability}
    \E\left[W(\hat{\beta}) - u(\hat{\beta}) \right] & = \frac{1}{n} \sum_{i=1}^n (\mu_i - k_i)\E\left[\Phi\left(\frac{\mu_i - \delta(\hat{\beta},z_i)}{\sigma_i}\right) - \Phi\left(\frac{\mu_i - \delta(\hat{\beta}^{(-
    i)},z_i)}{\sigma_i}\right) \right] \\
    \label{eq:line_2_loo_stability}
    & \quad + \frac{1}{n}\sum_{i=1}^n (\mu_i - k_i)\left(\Prob(Y_i \geq \delta(\hat{\beta}^{(-i)},z_i)) - \Prob(Y_i \geq \delta(\hat{\beta},z_i)) \right).
  \end{align}

  Using Lipschitzness of $\Phi$ and $\delta(\cdot,z_i)$, we will upper bound \eqref{eq:line_1_loo_stability} by
  \[
    \frac{1}{n}\sum_{i=1}^n \frac{L_i|\mu_i - k_i|}{\sigma_i} \E|\hat{\beta} - \hat{\beta}^{(-i)}|
  \] for some constants $L_i$. To handle \eqref{eq:line_2_loo_stability}, we truncate. Fix a truncation
  threshold $E_n := C(\log n)^\gamma / n$ for $C,\gamma$ constants to be chosen later.
  \begin{align*}
    & \Prob(Y_i \geq \delta(\hat{\beta}^{(-i)},z_i)) - \Prob(Y_i \geq \delta(\hat{\beta},z_i)) \\
    & = \E\left[ \Prob(Y_i \geq \delta(\hat{\beta}^{(-i)},z_i) \ | \ Y_{-i}) - \Prob
    (Y_i \geq \delta(\hat{\beta},z_i) \ | \ Y_{-i}) \right] \\
    & \leq \E \left[\Prob\br{|Y_i - \delta(\hat{\beta}^{(-i)},z_i)| \leq |\delta(
    \hat{\beta}^{(-i)},z_i) - \delta(\hat{\beta},z_i)| \ | \ Y_{-i}} \right] \\
    & \leq \E \left[\Prob(|Y_i - \delta(\hat{\beta}^{(-i)},z_i)| \leq L_i|\hat{\beta}^{(-i)} - \hat{\beta}| \ | \ Y_{-i}) \right] \\
    & \leq \E \left[\Prob(|Y_i - \delta(\hat{\beta}^{(-i)},z_i)| \leq L_i E_n \ | \ Y_{-i} )  \right] + \Prob(|\hat{\beta}^{(-i)} - \hat{\beta}| > E_n) \\
    & \leq \E \left[ \Phi\left(\frac{\delta(\hat{\beta}^{(-
      i)},z_i) + L_i E_n - \mu_i}{\sigma_i}\right) - \Phi\left(\frac{\delta(\hat{\beta}^{(-
    i)},z_i) - L_i E_n - \mu_i}{\sigma_i}\right) \right] \\
    & \quad + \Prob(|\hat{\beta}^{(-i)} - \hat{\beta}| > E_n).
  \end{align*}
  Using Lipschitzness again, we find that \eqref{eq:line_2_loo_stability} is bounded above by
  \begin{equation}
    \abs{\eqref{eq:line_2_loo_stability}} \le \frac{1}{n}\sum_{i=1}^n \frac{L_i|\mu_i - k_i|}
    {\sigma_i} E_n + \frac{1} {n}\sum_ {i=1}^n \abs{\mu_i - k_i} \Prob(|\hat{\beta}^{(-i)} -
    \hat{\beta}| > E_n).
  \end{equation}
  By the compactness assumption on the search space for $\beta,$ we may bound \[\E[|
    \hat{\beta}^{(-i)} - \hat{\beta}|] \leq M \Prob(|\hat{\beta}^{(-i)} -
  \hat{\beta}| > E_n) + E_n\] for some constant $M$ depending on the diameter of the
  search space. Therefore we may bound $\E \left[ W(\hat{\beta}) -  u(\hat
  {\beta})\right]$ by
  \begin{equation}
    \frac{1}{n}\sum_{i=1}^n \left(1 + \frac{L_i}{\sigma_i} \right)|\mu_i - k_i| \left(E_n + \Prob(|\hat{\beta}^{(-i)} - \hat{\beta}| > E_n) \right).
  \end{equation}
  \Cref{thm:LOO_stability} shows that
  \[
    \Prob\left(|\hat{\beta}^{(-i)} - \hat{\beta}| > \frac{4MB \log n}{\kappa n}\right) \leq O\left(\frac{1}{n} \right),
  \]
  so we may take $E_n = O(\frac{\log n}{n}).$ Thus
  \begin{align*}
    \E \left[ W(\hat{\beta}) -  u(\hat{\beta})\right] & \leq  \frac{B^3}{n} \sum_{i=1}^n \left(E_n + \Prob(|\hat{\beta}^{(-i)} - \hat{\beta}| > E_n) \right) \\
    & \leq B^3 \left(O\left(\frac{\log n}{n}\right) + O\left(\frac{1}{n}\right)\right) = O\left(\frac{\log n}{n}\right).
  \end{align*}
  Combining this with the conclusion of \cref{thm:small_external_regret} gives the result.
\end{proof}

\medskip
\begin{proof}[\textbf{Proof of \cref{prop:gaussian_matching_lb_fast_rate}}] Throughout the proof
  let $a_{1:n}$ denote a vector of binary selections, where the $i$th index is $a_i
  (Y_1,\dots,Y_n) \in \set{0,1}.$ We will first exhibit a class of priors that are
  supported on $\Theta$. \Cref{prop:assumptionBsatisfiesA} verifies that a set of
  assumptions \cref{assmp:lower_level_mu_boundedness}--\cref
  {assmp:lower_level_curvature} are sufficient for \cref{as:fast_rate_assn}.

  Consider the priors $G_p^{\otimes n}$, where $G_p := \mathsf{Rad}_p$ is the
  one-dimensional Rademacher prior with $\mu_i = +1$ with probability $p$ and $\mu_i = -1$
  with probability $1-p$. Let $G_p^{\otimes n}|A$ be the conditional distribution of $\mu_
  {1:n}$ conditioned on $A := \set{\min_i \mu_i < \max_i \mu_i} \cap \set{|\overline
  {\mu}| < C'}$ for some absolute constant $C'$ small enough. Then $G_p^
  {\otimes n}|A$ satisfies \cref{assmp:lower_level_mu_boundedness} with $B = 1$ and,
  choosing $\kappa$ small enough, it satisfies \cref{assmp:lower_level_curvature} since
  $\frac{1}{n}\sum_{i=1}^n \mu_i^2 = 1$. Lastly, \cref
  {assmp:lower_level_bounded_decision} is implied if $|\overline{\mu}| < \frac{m_1}
  {\varphi(1)} \overline{\sigma}^2$. Since $\mu_i$ is either $+1$ or $-1$, this is in turn
  implied by $\overline{\mu} < C'$ for $C'$ small enough.


  Next, observe the following inequalities:
  \begin{align}
    \sup_{\mu_{1:n} \in \Theta} \regret_n(a_{1:n}) & \geq \E_{G_p^{\otimes n}|A}\left[\regret_n(a_{1:n})\right] \\
    & = \frac{\E_{G_p^{\otimes n}}\left[\regret_n(a_{1:n})\one(A)\right]}{\Prob_{
    G_p^{\otimes n}}(A)} \\
    & = \frac{\E_{G_p^{\otimes n}}\left[\regret_n(a_{1:n})\right] - \E_{
    G_p^{\otimes n}}\left[\regret_n(a_{1:n}) \one(A^c)\right]}{\Prob_{
    G_p^ {\otimes n}}(A)} \\
    & \geq \frac{\E_{G_p^{\otimes n}}\left[\regret_n(a_{1:n})\right] -\Prob_{G_p^{\otimes n}}(A^c)}{1 - \Prob_{G_p^{\otimes n}}(A^c)}.
    \label{eq:lb_reduction}
  \end{align}
  Observe that $\Prob_{G_p^{\otimes n}}(A^c)$ is exponentially small in $n$ for
  all $p$ in a small neighborhood of $\frac{1}{2}.$

  We will thus focus on analyzing the term
  $\E_{G_p^{\otimes n}}\left[\regret_n(a_{1:n})\right].$ We have
  \begin{align*}
    \E_{G_p^{\otimes n}}\left[\regret_n(a_{1:n})\right] & =  \E_{G_p^{\otimes n}} \left\lbrace \sup_\beta W(\beta) - \frac{1}{n}\sum_{i=1}^n \mu_i\E\left[a_i(Y_1,\dots,Y_n) \right] \right\rbrace \\
    & \geq \sup_\beta \E_{G_p^{\otimes n}}W(\beta) - \E_{G_p^{\otimes n}} \left[\frac{1}{n}\sum_{i=1}^n \mu_i\E\left[a_i(Y_1,\dots,Y_n) \right] \right].
  \end{align*}
  Since $W(\beta) = \frac{1}{n}\sum_{i=1}^n \mu_i \Prob(Y_i \geq \beta)$ and $\mu_i$ are independent, observe that
  \[
    \sup_\beta \E_{G_p^{\otimes n}} W(\beta)
    = \sup_\beta \E_{G_p}\left[\mu \Prob(Y \geq \beta \mid \mu)\right]
    = \E_{G_p}\E_Y[\mu \mathbf{1}\set{\E_{G_p}[\mu \mid Y] > 0}],
  \] which follows from monotonicity of the posterior expectation in $Y.$ This is the
  welfare of the Bayes-optimal decision (see e.g. \citet{azevedo2020b}). Therefore, we
  have lower bounded $\sup_{\mu_{1:n} \in \Theta} \regret_n(a_{1:n})$ by the Bayes regret
  of any product prior $G_p^{\otimes n}$ up to some additive errors and multiplicative
  factors that are exponentially close to $1$ in $n$.

  We will finish the proof using Le Cam's two point argument. By \cref{lemma:monotone} we
  may restrict to decision rules which are monotone: $a_i
  (Y) = \mathbf{1}\set{Y_i > \delta_i(Y_{(-i)})}$. For a prior $G$ on $\set{-1,+1}$, let
  $m_G(y) = \E_G[\mu \mid y]$ be the posterior mean, let $a_G^*(y) = \mathbf{1}\set{m_G(y)
  > 0}$ be the Bayes optimal decision, and let $R(G,a_{1:n})$ denote the Bayes regret
  under the product prior $G^{\otimes n}$. We can equivalently write this Bayes regret as
  \begin{align*}
    R(G,a_{1:n}) & := \frac{1}{n}\sum_{i=1}^n \E_{G^{\otimes n}}\left[ \mu_i \left( a_G^*(Y_i) - a_i(Y_{1:n})\right)\right] \\
    & = \frac{1}{n}\sum_{i=1}^n \E_{G^{\otimes n}}\left[ m_G(Y_i) \left( a_G^*(Y_i) - a_i(Y_
    {1:n})\right)\right] \\&= \frac{1}{n}\sum_{i=1}^n \E_{G^{\otimes n}}\left[\left|m_G(Y_i)\right|
    \mathbf{1}\set{a_G^*(Y_i) \neq a_i(Y_{1:n})}\right].\numberthis
    \label{eq:bayes_regret_as_classification}
  \end{align*}
  where \eqref{eq:bayes_regret_as_classification} is verified by
  \cref{lemma:bayres_regret_as_classification}.

  Take
  \[
    G_0 := \mathsf{Rad}_{1/2 - 1/(4\sqrt{n})},
    \qquad
    G_1 := \mathsf{Rad}_{1/2 + 1/(4\sqrt{n})}.
  \]
  By \cref{lemma:rad_pair_large_bayes_regret}, for every monotone rule $a_{1:n}$ there is
  an absolute constant $c > 0$ such that
  \[
    R(G_0,a_{1:n}) + R(G_1,a_{1:n}) \geq \frac{c}{n}
  \]
  for all sufficiently large $n$. Hence
  \[
    \max\left(R(G_0,a_{1:n}),R(G_1,a_{1:n})\right) \geq \frac{c}{2n}.
  \]
  Applying \eqref{eq:lb_reduction} separately under $G_0^{\otimes n}$ and $G_1^{\otimes n}$, and using
  $\Prob_{G_0^{\otimes n}}(A^c), \Prob_{G_1^{\otimes n}}(A^c) = O(e^{-n})$, we obtain
  \begin{align*}
    \sup_{\mu_{1:n} \in \Theta} \regret_n(a_{1:n})
    & \geq \max\left(\frac{R(G_0,a_{1:n}) - \Prob_{G_0^{\otimes n}}(A^c)}{1 - \Prob_{G_0^{\otimes n}}(A^c)},
    \frac{R(G_1,a_{1:n}) - \Prob_{G_1^{\otimes n}}(A^c)}{1 - \Prob_{G_1^{\otimes n}}(A^c)}\right) \\
    & \geq \frac{c}{2n} - O(e^{-n})
    = \Omega(n^{-1}),
  \end{align*}
  as desired.
\end{proof}

\medskip
\begin{proof}[\textbf{Proof of \cref{thm:assure_nongaussian_microdata}}]
  For each $i$, let $Z_i \sim \Norm(\mu_i,\sigma_i^2)$. Recall the definition of $w_h$ in \eqref{eq:assure_term} and define
  $w_{h,i}(y; \delta) = w_h(y; \sigma_i, \delta).$
  By the triangle inequality,
  \begin{align*}
    & \left| \frac{1}{n}\sum_{i=1}^n \E\left[\hat{w}_{h,i}(\sqrt{m}\bar{Y}_i;\delta_i)\right]
    - W(\delta) \right| \\
    & \leq \frac{1}{n}\sum_{i=1}^n \left|\E\left[\hat{w}_{h,i}(\sqrt{m}\bar{Y}_i;\delta_i) -
    w_{h,i}(\sqrt{m}\bar{Y}_i;\delta_i)\right]\right| \\
    & \quad + \frac{1}{n}\sum_{i=1}^n \left|\E\left[w_{h,i}(\sqrt{m}\bar{Y}_i;\delta_i) -
    w_{h,i}(Z_i;\delta_i)\right]\right| \\
    & \quad + \left| \frac{1}{n}\sum_{i=1}^n \E\left[w_{h,i}(Z_i;\delta_i)\right] -
    W(\delta) \right| \\
    & =: T_1 + T_2 + T_3.
  \end{align*}
  The terms have distinct roles: $T_1$ controls the effect of estimating $\sigma_i^2$,
  $T_2$ controls the error from replacing the non-Gaussian sample averages by Gaussian
  analogues, and $T_3$ is the residual bias from targeting the Gaussian welfare rather than
  $W(\delta)$.

  By \cref{lemma:nongaussian_microdata_term1}, $T_1 \lesssim \sqrt{\log n}/\sqrt{m}$. By
  \cref{lemma:nongaussian_microdata_term2}, $T_2 \lesssim (\log n)^2/\sqrt{m}$. The
  idea for $T_2$ is to use the Lindeberg swapping argument, defined as follows. For each $i$, relabel the micro-data as $Y_1,\dots,Y_m$,
  introduce Gaussian variables $Z_1,\dots,Z_m$ with the same mean and variance, and define
  hybrid sums $S_k$ in which the first $k$ terms come from $Z_t$ and the remaining $m-k$
  terms
  are the original $Y_t$. Then
  \[
    \E\left[w_{h,i}(Z_i;\delta_i)\right] - \E\left[w_{h,i}(\sqrt{m}\bar{Y}_i;\delta_i)\right]
    = \sum_{k=0}^{m-1} \E\left[w_{h,i}(S_{k+1};\delta_i) - w_{h,i}(S_k;\delta_i)\right].
  \]
  Each increment is expanded around the common background sum $S_k'$ obtained by removing the
  swapped variable. Because $Y_{k+1}$ and $Z_{k+1}$ have the same first two moments, the
  constant, linear, and quadratic Taylor terms cancel after taking expectations, so only a
  third-order remainder remains. \cref{lemma:nongaussian_microdata_term2} shows that
  $|w_{h,i}^{(3)}(x;\delta_i)| \lesssim \sigma_i^{-2} h^{-4} + |x|\sigma_i^{-3} h^{-3}$, so
  one swap costs at most $P_i m^{-3/2}(h^{-4}+h^{-3})$ for a polynomial $P_i$ in the moments.
  Summing over the $m$ swaps gives $P_i(h^{-4}+h^{-3})/\sqrt{m}$, which is
  $P_i(\log n)^2/\sqrt{m}$ since $h = 1/\sqrt{2\log n}$.

  Finally,
  \cref{lemma:nongaussian_microdata_term3} gives
  $T_3 \lesssim (n\log n)^{-1} + m^{-1/2}$. Combining the three bounds proves the claim.
\end{proof}

\section{Miscellany}

\begin{lemma}[Reduction to monotone decision rules]
  \label{lemma:monotone}
  Fix $\sigma_{1:n}, k_{1:n}$. Let $a_i (Y_i; Y_{-i}) \in [0,1]$ be a decision rule and
  assume $a_i(\cdot, Y_{-i})$ is not almost surely zero or almost surely one. There
  exists a threshold rule $a_i^* (\cdot; Y_ {-i}) = \one(Y_i \ge \delta_i(Y_{-i}))$, where
  $\delta_i(Y_{-i}) \in \R$, such that
  \[
    \E_{\mu_{1:n}} [(\mu_i - k_i) a_i(Y_i; Y_{-i})] \le \E_{\mu_{1:n}} [(\mu_i - k_i)
    a_i^* (Y_i; Y_{-i})]
  \]
  for all $\mu_{1:n} \in \R^n$.
\end{lemma}
\begin{proof}
  Let \[
    \gamma_0(Y_{-i}) = \E_{\mu_i = k_i}[ a_i(Y_i; Y_{-i}) \mid Y_{-i} ] \in [0,1].
  \]
  Define $a_i^*(y; Y_{-i}) = \one(y \ge c(Y_{-i}))$ such that \[
    \E_{\mu_i = k_i}\bk{\one(Y_i \ge c(Y_{-i}))  \mid Y_{-i}} =  \gamma_0(Y_{-i}).
  \]

  Note that $a_i^*(y; Y_{-i})$ is the uniformly most powerful test for $H_0: \mu_i \le
  k_i$ with \emph{conditional size} $\gamma_0(Y_{-i})$, against $H_1: \mu_{i} > k_i$.
  Likewise, $1-a_i^*(y; Y_{-i})$ is the UMP test when we swap the null and the
  alternative. $a_i(y; Y_{-i})$ is a randomized test that has the same conditional size
  $\gamma_0(Y_{-i})$.

  Let $\mu_i > k_i$, then \[
    \E_{\mu_i}[a_i(Y_i; Y_{-i}) \mid Y_{-i}] \le \E_{\mu_i}[a_i^*(Y_i; Y_{-i}) \mid Y_{-i}]
  \]
  since $a_i^*$ is weakly more powerful than $a_i$. Conversely, if $\mu_i < k_i$, then \[
    \E_{\mu_i}[a_i(Y_i; Y_{-i}) \mid Y_{-i}] \ge \E_{\mu_i}[a_i^*(Y_i; Y_{-i}) \mid Y_{-i}]
  \]
  since $1-a^*$ is more powerful than $1-a_i$. Thus \[
    (\mu_i - k_i) \E_{\mu_i}[a_i(Y_i; Y_{-i}) \mid Y_{-i}] \le (\mu_i - k_i) \E_
    {\mu_i}[a_i^*(Y_i; Y_{-i}) \mid Y_{-i}].
  \]
  The conclusion follows by integrating out $Y_{-i}$.
\end{proof}

\subsection{Lack of an unbiased estimator}

We present an argument that is referenced in \citet{stefanski1989unbiased} without proof.
By \cref{lemma:truncated_exp}, if an unbiased estimator for $W(\beta)$ were to exist, then one could unbiasedly estimate $\mu \Phi\pr{\frac{\mu-c}{\sigma}}$. This means that one could estimate $\varphi(\mu/\sigma)/\sigma$ unbiasedly. \Cref{prop:unbiased_stefanski} then shows that $\varphi
(\mu)$ cannot be unbiasedly estimated, at least not with estimators with reasonable tail
behavior.

\begin{lemma}
  \label{lemma:truncated_exp}
  Let $Y \sim \Norm(\mu,\sigma^2)$. Then
  \[
    \E[Y\mathbf{1}\set{Y \geq c}] = \frac{\sigma^2}{\sqrt{2\pi\sigma^2}}\exp\left(-\frac{(\mu - c)^2}{2\sigma^2}\right) + \mu\Phi\left(\frac{\mu - c}{\sigma}\right).
  \]
\end{lemma}
\begin{proof}[Proof of \cref{lemma:truncated_exp}]
  Rewrite the expression as
  \begin{align*}
    \E[Y\mathbf{1}\set{Y \geq c}] & = \E[(Y-\mu)\mathbf{1}\set{Y - \mu \geq c - \mu}] + \mu \Prob(Y \geq c) \\
    & = \int_{c-\mu}^\infty x\frac{1}{\sqrt{2\pi\sigma^2}}\exp\left(-\frac{x^2}{2\sigma^2}\right) dx + \mu \left(1 -\Phi\left(\frac{c - \mu}{\sigma} \right) \right) \\
    & = -\frac{\sigma}{\sqrt{2\pi}}\exp\left(-\frac{x^2}{2\sigma^2}\right)\bigg\rvert_{c-\mu}^\infty + \mu \left(1 -\Phi\left(\frac{c - \mu}{\sigma} \right) \right) \\
    & = \frac{\sigma}{\sqrt{2\pi}}\exp\left(-\frac{(c-\mu)^2}{2\sigma^2}\right) + \mu
    \left(1 -\Phi\left(\frac{c - \mu}{\sigma} \right) \right).\qedhere
  \end{align*}
\end{proof}

\begin{proposition}[\citet{stefanski1989unbiased}]
  \label{prop:unbiased_stefanski}
  Let $Y \sim \Norm(0,1)$.
  Let $f: \R \to \R$ be some real-valued function where $\E_{Y \sim \Norm(\mu,1)}[|f(Y)|]
  < \infty$ for every $\mu \in \R$.
  Define $m : \C \to \C$ as $m(z) = \int f(y) \varphi(y-z) \,dy$. If $m$ is entire, then
  there is some $\mu \in \R$ where
  $
    \E_\mu[f(Y)] \neq \varphi(\mu).
  $
\end{proposition}
\begin{proof}
  Towards contradiction, suppose $m(\mu) = \E_\mu[f(Y)]=\varphi(\mu)$ over $\mu\in\R$.
  Since $\varphi(\cdot)$ is analytic, we must have $m(z) = \varphi(z)$ over $\C$. Let
  $z = it$, we have that for all $t \in \R$, \[
    \int_\R f(y) \varphi(y) e^{ity} dy e^{t^2/2} = \frac{1}{\sqrt{2\pi}}e^{t^2/2} \implies  \int_\R f(y)
    \varphi(y) e^{ity} dy = \frac{1}{\sqrt{2\pi}}
  \]
  Note that $
    \int |f(y) \varphi(y)| dy = \int |f(y)| \varphi(y) \,dy = \E_{Y \sim \Norm(0,1)}
    [|f (Y)|] < \infty.
  $
  Thus $f \varphi$ is $L^1$. By the Riemann--Lebesgue lemma, the Fourier transform
  vanishes at infinity: \[
    0 = \lim_{|t| \to \infty} \int_\R f(y)
    \varphi(y) e^{ity} dy = \frac{1}{\sqrt{2\pi}}.
  \]
  This is a contradiction.
\end{proof}

It remains to discuss when $m(z)$ is entire. A sufficient condition is that $f(Y)$ is
integrable under a slightly noisier Gaussian: $\E_{Y \sim \Norm(0, 1+\epsilon^2)}[|f(Y)|]
< \infty$. When this is true, we can then differentiate the real and imaginary parts of
$m(z)$ under the integral sign and verify that the partials are continuous and satisfy the
Cauchy--Riemann equations.

\newpage

\begin{appendices}
  \begin{center}

    \textbf{\large Online Appendix for ``Compound Selection Decisions: An Almost SURE
    Approach''}

    Jiafeng Chen, Lihua Lei, Timothy Sudijono, Liyang Sun, and Tian Xie

    \today
  \end{center}

  \DoToC

  \newpage

  \setcounter{section}{2}

  \section{Auxiliary results for proofs}
  \label{sec:auxiliary_results_for_proofs}

  This appendix collects auxiliary lemmas and intermediate arguments organized by the
  main-text proof they support.

  \subsection{Auxiliary results for \cref{thm:assure_bias}}

  \begin{lemma}
    \label{lemma:sin_cos_gaussian_facts}
    Let $Y \sim \Norm(\mu,\sigma^2)$ and $\omega, C$ be constants. Then:
    \begin{align*}
      \E \sin \omega(Y - C) & = e^{-\frac{1}{2}\sigma^2 \omega^2} \sin(\omega(\mu - C))\\
      \E \cos \omega(Y - C) & = e^{-\frac{1}{2}\sigma^2 \omega^2} \cos(\omega(\mu - C))\\
      \E\left[Y \sin(\omega(Y - C)) \right] & = e^{-\frac{1}{2}\sigma^2 \omega^2}\left[\mu \sin(\omega(\mu - C)) + \omega \sigma^2 \cos(\omega(\mu - C))  \right] \\
      & = e^{-\frac{1}{2}\sigma^2 \omega^2}\left[\mu \sin(\omega(\mu - C))  \right] + \sigma^2 \omega \E\left[\cos(\omega(Y - C))\right].
    \end{align*}
  \end{lemma}
  \begin{proof}
    For the first claim, observe that
    \begin{align*}
      \E\left[\sin \omega(Y - C)\right] & = \frac{1}{2i}\E\left[e^{i\omega(Y - C)} - e^{-i\omega (Y-C)} \right] \\
      & = \frac{1}{2i}\left(e^{i\omega \mu - \frac{1}{2}\sigma^2 \omega^2 - i\omega C}  - e^{-i\omega \mu - \frac{1}{2}\sigma^2 \omega^2 + i\omega C} \right) \\
      & = \frac{1}{2i}e^{-\frac{1}{2}\sigma^2\omega^2}\left(e^{i\omega \mu- i\omega C}  - e^{-i\omega \mu + i\omega C} \right) \\
      & = e^{-\frac{1}{2}\sigma^2\omega^2} \sin(\omega(\mu - C)).
    \end{align*}
    An analogous argument with cosine shows the second claim. For the third claim, use
    Stein's identity $
    \E[Yf(Y)] = \mu \E f(Y) + \sigma^2\E[f'(Y)]
    $, and then apply the first two equalities of the lemma.
  \end{proof}

  \begin{lemma}
    \label{lemma:expected_assure}
    In the proof of \cref{thm:assure_bias}, we have, for $z = (k,\sigma)$,
    \begin{align*}
      (\mu-k) \Phi\left(\frac{\mu - C}{\sigma} \right) & = \frac{1}{2} (\mu - k) + \frac{1}{\pi} (\mu - k) \int_0^\infty  \frac{1}{\omega} e^{-\frac{1}{2}\omega^2} \sin(\omega (\mu - C)/\sigma)d\omega \\
      \E \Psi_h(Y,z,C) & = \frac{1}{2} (\mu - k) + \frac{1}{\pi} (\mu - k) \int_0^{1/h}\frac{1}{\omega} e^{-\frac{1}{2}\omega^2} \sin(\omega (\mu - C)/\sigma)d\omega
    \end{align*}
    interpreted as Lebesgue integrals.
  \end{lemma}

  \begin{proof}

    We consider the fact that for all $a \in \bR$,\footnote{The integral is interpreted as an improper Riemann integral, so
      \[
        \int_0^\infty \frac{\sin ax}{x} dx = \lim_{h \downarrow 0} \int_0^{1/h}\frac{\sin ax}{x} dx.
      \]
      The truncated integrals $\int_0^{1/h}\frac{\sin ax}{x} dx$ can be treated as Lebesgue integrals.
    }
    \begin{equation}
      \label{eq:sinc_integral_identity}
      \int_0^\infty \frac{\sin ax}{x} dx = \frac{\pi}{2}\text{sign}(a).
    \end{equation}
    Then
    \begin{align*}
      & (\mu - k)\Phi\left(\frac{\mu - C}{\sigma} \right) \\
      & = (\mu - k) \int_{-\infty}^\infty \mathbf{1}\set{y - C \geq 0}\frac{1}{\sigma}\varphi\left(\frac{y - \mu}{\sigma} \right) dy \\
      & = (\mu - k)\int_{-\infty}^\infty \left(\frac{1}{2} + \frac{1}{2}\text{sign}\left(\frac{y - C}{\sigma} \right) \right)\frac{1}{\sigma}\varphi\left(\frac{y - \mu}{\sigma} \right) dy \\
      & = \frac{1}{2} (\mu - k) + \frac{1}{\pi}(\mu - k) \int_{-\infty}^\infty \lim_{h \downarrow 0} \int_0^{1/h}\frac{\sin(\omega(y-C)/\sigma)}{\omega} d\omega \frac{1}{\sigma}\varphi\left(\frac{y - \mu}{\sigma} \right) dy \\
      & = \frac{1}{2} (\mu - k) + \frac{1}{\pi}(\mu - k) \lim_{h \downarrow 0} \int_{-\infty}^\infty  \int_0^{1/h}\frac{\sin(\omega(y-C)/\sigma)}{\omega} d\omega \frac{1}{\sigma}\varphi\left(\frac{y - \mu}{\sigma} \right) dy.
    \end{align*}
    The last step follows from dominated convergence theorem, noting that
    \[
      \sup_h\left|  \int_0^{1/h}\frac{\sin(\omega(y-C)/\sigma)}{\omega} d\omega \right| = \sup_h\left| \Si((y-C)/h\sigma) \right| = K <\infty
    \]
    if $x > C$, for an absolute constant $K$. A similar calculation shows when $x - C < 0$ the integral is bounded above by the same uniform constant $K.$ Thus an integrable dominating function is given by $K \varphi((y-\mu)/\sigma)/\sigma.$ Now applying Fubini's theorem, we can switch the order of integration and apply \cref{lemma:sin_cos_gaussian_facts}:
    \begin{align*}
      &\frac{1}{\pi}(\mu - k) \lim_{h \downarrow 0} \int_0^{1/h}\frac{1}{\omega} \left(\int_{-\infty}^\infty  \sin(\omega(y-C)/\sigma)  \frac{1}{\sigma}\varphi\left(\frac{y - \mu}{\sigma} \right) dy \right)d\omega \\
      & =  \frac{1}{\pi} \lim_{h \downarrow 0} \int_0^{1/h}\frac{1}{\omega} e^{-\frac{1}{2}\omega^2} (\mu - k) \sin(\omega (\mu - C)/\sigma)d\omega \\
      & = \frac{1}{\pi} \int_0^\infty \frac{1}{\omega} e^{-\frac{1}{2}\omega^2} (\mu - k)
      \sin(\omega (\mu - C)/\sigma)d\omega.
    \end{align*}
    This proves the first claim. For the second claim, note that
    \begin{align*}
      &\frac{1}{\pi} \int_0^{1/h} \frac{1}{\omega} e^{-\frac{1}{2}\omega^2} (\mu - k)
      \sin(\omega (\mu - C)/\sigma)d\omega \\
      & = \frac{1}{\pi} \int_0^{1/h}\frac{1}{\omega}\int_{-\infty}^\infty \left[ (y-k)\sin\left(\frac{\omega(y-C)}{\sigma}\right) - \omega \sigma \cos\left(\frac{\omega(y-C)}{\sigma}\right) \right] \frac{1}{\sigma}\varphi\left(\frac{y - \mu}{\sigma} \right) dy d\omega \\
      & = \frac{1}{\pi} \int_0^{1/h}\int_{-\infty}^\infty \left[ \frac{(y-k)\sin\left(\frac{\omega(y-C)}{\sigma}\right)}{\omega} - \sigma \cos\left(\frac{\omega(y-C)}{\sigma}\right) \right] \frac{1}{\sigma} \varphi\left(\frac{y - \mu}{\sigma} \right) dy d\omega.
    \end{align*}
    Again we may apply Fubini noting that the integral over $\omega$, being over a compact domain, is finite, yielding
    \begin{align*}
      & \frac{1}{\pi}   \int_{-\infty}^\infty \int_0^{1/h}\left[ (y-k)\frac{\sin(\omega(y-C)/\sigma)}{\omega} - \sigma \cos(\omega(y-C)/\sigma) \right] d\omega \frac{1}{\sigma} \varphi\left(\frac{y - \mu}{\sigma} \right) dy \\
      & =\frac{1}{\pi}   \int_{-\infty}^\infty \left[ \int_0^{1/h} (y-k)\frac{\sin(\omega(y-C)/\sigma)}{\omega} d\omega - \sigma \frac{\sin((y-C)/h\sigma)}{(y-C)/\sigma}\right]  \frac{1}{\sigma} \varphi\left(\frac{y - \mu}{\sigma} \right) dy \\
      & = \E\left[ \frac{1}{\pi}(Y-k) \Si\left(\frac{Y - C}
      {h\sigma}\right) - \frac{\sigma}{h} \sinc\left(\frac{Y - C}{h\sigma}\right) \right] \\
      &=  \E\bk{
        \Psi_h(Y; z, C)
      } - \frac{1}{2}(\mu-k).
    \end{align*}
    This proves the second claim.
  \end{proof}

  \subsection{Auxiliary results for \cref{thm:main_regret_bound}}

To state the proofs, define the norms $ \norm{D}_ {p,n}^p := \frac{1}{n}\sum_{i=1}^n |D(z_i)|^p$ and $\norm{D}_{n} := \norm{D}_
{2,n}$. Let $s(z_i) = \sigma_i$ and for $q \ge 1$, $s_q:= \norm{s}_{q,n}$. Let $m(z_i) =
\mu_i - k_i$ and $m_q:= \norm{m}_{q,n}$.
  \begin{proposition}
    \label{prop:realized_regret_term_empirical_process_bound}
    In the proof of \cref{thm:main_regret_bound}, we have
    \[
      \E \sup_{\beta\in\mathcal B} |A_1(\beta)| \lesssim m_2\frac{\sqrt{V(\cl{D})}}{\sqrt{n}}.
    \]
  \end{proposition}
  \begin{proof}
    We will handle this with i.n.i.d. empirical process theory applied to the function class $
    \cl{F} = \set{f_\beta: \beta \in \mathcal B}$ with $f_\beta(y,\mu,z) = (\mu - k)\mathbf
    {1}\set{y - \delta(z;\beta) \geq 0}$. The class $\cl{F}$ has a uniform envelope function
    given by $F = |\mu - k|$. Moreover, \cref{lemma:vc_indexes} shows that $\cl{F}$ is a
    VC subgraph class with index $O(V(\cl{D}))$: in
    particular, $y - \delta (\beta,z)$ is a
    VC subgraph function class ranging over $\beta\in\mathcal B,$ and so is $\mathbf{1}\set{y - \delta
    (\beta,z) \geq 0}$. Multiplying by the fixed function $(\mu-k)$ and applying \cref
    {lemma:vc_indexes}(2) gives the claim. Next, we apply \cref
    {thm:inid_donsker_bound} with $V_i = (Y_i,\mu_i,z_i)$ and \cref{lemma:vc_subgraph_class_covering_number_bound},
    \begin{align*}
      \E \sup_{\beta\in\mathcal B} |A_1(\beta)| & = \frac{1}{\sqrt{n}} \E \sup_{\beta\in\mathcal B} \frac{n}{\sqrt{n}}|A_1
      (\beta)|
      = \frac{1}{\sqrt{n}} \E \sup_{\beta\in\mathcal B} \frac{1}{\sqrt{n}} \sum_{i=1}^n \br{
        f_\beta(Y_i, \mu_i, z_i) - \E[f_\beta(Y_i, \mu_i, z_i)]
      }
      \\
      & \leq \frac{1}{\sqrt{n}} J(1,\cl{F} \ | \ F, \bb{L}^2) \left(\frac{1}{n}\sum_{i=1}^n |\mu_i - k_i|^2 \right)^{1/2}
      \tag{\cref{thm:inid_donsker_bound}}
      \\
      & \lesssim \frac{1}{\sqrt{n}} \left(\frac{1}{n}\sum_{i=1}^n |\mu_i - k_i|^2 \right)^{1/2} \int_0^1 \sqrt{1 + V(\cl{D}) \ln \frac{1}{\e} } d\e
      \tag{\cref{lemma:vc_subgraph_class_covering_number_bound}}
      \\
      & \lesssim \sqrt{V(\cl{D})}\left(\frac{1}{n}\sum_{i=1}^n (\mu_i - k_i)^2 \right)^{1/2}
      \frac{1}{\sqrt{n}}. \qedhere
    \end{align*}
  \end{proof}

  \begin{lemma}
    \label{lemma:sinc_second_moment_bound}
    Let $Y_i\sim \Norm(\mu_i,\sigma_i^2)$ and let $\lambda>0$. Uniformly over $C\in\bR$,
    \[
      \E\left[\left(\sigma_i\lambda\sinc\left(\lambda\frac{Y_i-C}{\sigma_i}\right)\right)^2\right]
      \lesssim \lambda\sigma_i^2.
    \]
    Consequently, for $h_\beta(y,z):=\sigma\lambda_n\sinc(\lambda_n(y-\delta(z;\beta))/\sigma)$,
    \[
      \sup_{\beta\in\mathcal B} \frac{1}{n}\sum_{i=1}^n P_i h_\beta^2 \lesssim \lambda_n s_2^2.
    \]
  \end{lemma}
  \begin{proof}
    By the change of variables $u=\lambda(Y_i-C)/\sigma_i$,
    \begin{align*}
      &\E\left[\left(\sigma_i\lambda\sinc\left(\lambda\frac{Y_i-C}{\sigma_i}\right)\right)^2\right] \\
      &\quad = \sigma_i^2\lambda \int_{\bR}
      \sinc^2(u)\varphi\left(\frac{C-\mu_i}{\sigma_i}+\frac{u}{\lambda}\right)du
      \lesssim \lambda\sigma_i^2,
    \end{align*}
    since $\sup_x\varphi(x)<\infty$ and $\int_{\bR}\sinc^2(u)du<\infty$. Averaging over
    $i$ gives the second claim.
  \end{proof}

  \begin{lemma}
  \label{lemma:si_remainder_second_moment_bound}
  Let
  \[
    \rho_\lambda(t)
    :=
    \frac{1}{\pi}\Si(\lambda t)-\frac{1}{2}\mathrm{sgn}(t),
    \qquad \lambda\ge 1.
  \]
  Let $Y_i\sim \Norm(\mu_i,\sigma_i^2)$. Then, uniformly over $C\in\bR$,
  \[
    \E\left[
      (Y_i-k_i)^2
      \rho_\lambda\left(\frac{Y_i-C}{\sigma_i}\right)^2
    \right]
    \lesssim
    \lambda^{-1}\left\{\sigma_i^2+(\mu_i-k_i)^2\right\}.
  \]
  Consequently, for
  $
    r_\beta(y,z)
    :=
    (y-k)\left[
      \frac{1}{\pi}\Si\left(\lambda_n\frac{y-\delta(z;\beta)}{\sigma}\right)
      -\frac{1}{2}\mathrm{sgn}\left(\frac{y-\delta(z;\beta)}{\sigma}\right)
    \right],
  $
  we have
  \[
    \sup_{\beta\in\mathcal B} \frac{1}{n}\sum_{i=1}^n \E r_\beta(Y_i,z_i)^2
    \lesssim
    \lambda_n^{-1}(s_2^2+m_2^2).
  \]
\end{lemma}

\begin{proof}
  First note that
  \[
    |\rho_\lambda(t)|
    \lesssim
    1\wedge \frac{1}{\lambda |t|}.
  \]
  Indeed, for $t>0$, by \eqref{eq:sinc_integral_identity},
  \[
    \rho_\lambda(t)
    =
    -\frac{1}{\pi}\int_{\lambda t}^\infty \frac{\sin u}{u}\,du,
  \]
  and integration by parts gives
  \[
    \left|\int_a^\infty \frac{\sin u}{u}\,du\right|
\le
\frac{|\cos a|}{a}
+
\int_a^\infty \frac{|\cos u|}{u^2}\,du
\le
\frac{1}{a}
+
\int_a^\infty \frac{1}{u^2}\,du
=
\frac{2}{a}, \quad a>0.
  \]
  The case $t<0$ follows by oddness, while the uniform boundedness follows from the
  boundedness of $\Si$.

  Set
  \[
    X_i:=\frac{Y_i-\mu_i}{\sigma_i}\sim \Norm(0,1),
    \qquad
    a:=\frac{\mu_i-C}{\sigma_i}.
  \]
  Then $(Y_i-C)/\sigma_i=X_i+a$ and $Y_i-k_i=(\mu_i-k_i)+\sigma_i X_i$. Thus it is enough
  to show, uniformly over $a\in\bR$,
  \[
    \E\rho_\lambda(X_i+a)^2\lesssim \lambda^{-1},
    \qquad
    \E X_i^2\rho_\lambda(X_i+a)^2\lesssim \lambda^{-1}.
  \]
  For the first display, using the boundedness of the normal density,
  \[
    \E\rho_\lambda(X_i+a)^2
    \lesssim
    \int_{\bR}\left(1\wedge \frac{1}{\lambda^2 t^2}\right)\varphi(t-a)\,dt
    \lesssim
    \int_{\bR}\left(1\wedge \frac{1}{\lambda^2 t^2}\right)\,dt
    \lesssim
    \lambda^{-1}.
  \]
  For the second display,
  \begin{align*}
    \E X_i^2\rho_\lambda(X_i+a)^2
    &=
    \int_{\bR}(t-a)^2\rho_\lambda(t)^2\varphi(t-a)\,dt \\ 
    &\lesssim \int_{\bR} (t^2 + a^2) \rho_\lambda(t)^2\varphi(t-a)\,dt \\
    & \le \underbrace{\int_{\bR} t^2 \cdot \frac{1}{\lambda^2 t^2}\varphi(t-a)\,dt}_{\lesssim \lambda^{-1}} + a^2\int_{\bR}\rho_\lambda(t)^2\varphi(t-a)\,dt
    \end{align*}
  It remains to bound
  \[
    a^2\int_{\bR}\rho_\lambda(t)^2\varphi(t-a)\,dt,
  \]
  uniformly over $a$. 
  If $|a|\le 2$, this is $\lesssim \lambda^{-1}$ by the previous bound. If $|a|>2$,
  split the integral over $|t|\le |a|/2$ and $|t|>|a|/2$. On the first region,
  $|t-a|\ge |a|/2$, so
  \[
    a^2\int_{|t|\le |a|/2}\rho_\lambda(t)^2\varphi(t-a)\,dt
    \lesssim
    a^2 e^{-a^2/8}\int_{\bR}\left(1\wedge \frac{1}{\lambda^2t^2}\right)dt
    \lesssim
    \lambda^{-1},
  \]
  where the last inequality is due to that $a^2 e^{-a^2/8}$ is bounded over $a\in \bR$.
  On the second region,
  $
    \rho_\lambda(t)^2
    \lesssim
    \frac{1}{\lambda^2 a^2},
  $
  and hence
  $
    a^2\int_{|t|>|a|/2}\rho_\lambda(t)^2\varphi(t-a)\,dt
    \lesssim
    \lambda^{-2}
    \lesssim
    \lambda^{-1}.
  $
  Therefore
  \[
    \E X_i^2\rho_\lambda(X_i+a)^2\lesssim \lambda^{-1}
  \]
  uniformly in $a$.

  Combining the preceding bounds,
  \[
    \E\left[
      (Y_i-k_i)^2
      \rho_\lambda\left(\frac{Y_i-C}{\sigma_i}\right)^2
    \right]
    \lesssim
    \lambda^{-1}\left\{\sigma_i^2+(\mu_i-k_i)^2\right\}.
  \]
  Averaging over $i$ and taking the supremum over $\beta$ gives the stated bound.
\end{proof}

\begin{lemma}
  \label{lemma:si_remainder_localized_entropy}
  Let
  \[
    \rho_{\lambda}(t)
    :=
    \frac{1}{\pi}\Si(\lambda t)-\frac{1}{2}\mathrm{sgn}(t),
  \]
  fix $M\geq1$, and define
  \[
    \cl{R}_M
    :=
    \set{
    (y-k)\rho_{\lambda_n}\left(\frac{y-\delta(z;\beta)}{\sigma}\right)
    \one\{|y-k|\leq M\}
    :\beta\in\mathcal B}.
  \]
  Suppose $L_\sigma\geq1$ satisfies $D(z)/\sigma\leq L_\sigma$, where $D$ is the envelope function of the decision class $\cl{D}$. 
  Let $0<b\leq1$ and $F_M(y,z)\ge C|y - k|\mathbf{1}\set{|y - k| \leq M}$ for some constant $C>0$ and is an envelope for $\cl{R}_M$. Then
  \[
    J(b,\cl{R}_M\mid F_M,L_2)
    \lesssim
    b\sqrt{V(\cl{D})\log\frac{e\lambda_n L_\sigma}{b}}.
  \]
\end{lemma}
\begin{proof}
  Define 
  \begin{align*}
    \cl{G}_M
    &:=
    \set{
    \frac{1}{\pi}(y-k)
    \Si\left(\lambda_n\frac{y-\delta(z;\beta)}{\sigma}\right)
    \one\{|y-k|\leq M\}:\beta\in\mathcal B},\\
    \cl{S}_M
    &:=
    \set{
    \frac{1}{2}(y-k)
    \mathrm{sgn}\left(\frac{y-\delta(z;\beta)}{\sigma}\right)
    \one\{|y-k|\leq M\}:\beta\in\mathcal B}.
  \end{align*}
  For an arbitrary discrete probability measure $Q$, define the reweighted measure
  \[
  dQ_{F_M}(y,z) := \frac{F_{M}(y,z)^2}{\norm{F_M}^2_{Q,2}} dQ(y,z).
  \]
  Let $Q$ be any such measure and $0<\eta\leq1$.
  Since $\cl{R}_M\subset \cl{G}_M-\cl{S}_M$, the product cover gives
  \[
    N(\eta \norm{F_M}_{Q,2},\cl{R}_M,L_2(Q))
    \leq
    N(\eta \norm{F_M}_{Q,2}/2,\cl{G}_M,L_2(Q))\cdot 
    N(\eta \norm{F_M}_{Q,2}/2,\cl{S}_M,L_2(Q)).
  \]

  We next construct an entropy bound for $\mathcal G_M$. Write
  \[
    g_{\beta,M}(y,z)
    :=
    \frac{1}{\pi}(y-k)
    \Si\left(\lambda_n\frac{y-\delta(z;\beta)}{\sigma}\right)
    \one\{|y-k|\leq M\}.
  \]
  The $\mathcal G_M$ has the deterministic envelope
  \[
    G_M(y,z)
    :=
    c |y-k|\one\{|y-k|\leq M\}
    \leq c F_M,
  \]
  for some absolute constant $c>0$ because $\Si(\cdot)$ is bounded. Define
  \[
    \cl{D}/\sigma
    :=
    \{(y,z)\mapsto \delta(z;\beta)/\sigma:\beta\in\mathcal B\}.
  \]
  We claim that, for some absolute constants $c>0$ and every $0<\eta<1$,
  \begin{equation}
  \label{eq:si_covering_number_pushforward}
    N(\eta \norm{F_M}_{Q,2},\cl{G}_M,L_2(Q))
    \leq
    N\left(c\frac{\eta}{\lambda_n},\cl{D}/\sigma,L_2(Q_{F_M})\right).
  \end{equation}
  Indeed, take any $c\eta/\lambda_n$-cover of $\cl{D}/\sigma$ in the $L_2(Q_{F_M})$
  metric. For any element $\delta(\cdot;\beta)/\sigma$ in this class, let
  $\tilde d/\sigma$ be a covering element satisfying
  $\norm{(\delta(\cdot;\beta)-\tilde d)/\sigma}_{Q_{F_M},2}\leq c\eta/\lambda_n$.
  Since $\Si$ is Lipschitz and the truncation is independent of $\beta$,  we have for some universal constant $C> 0$,
  \begin{align*}
    &\norm{g_{\beta,M}(y,z)
    -\frac{1}{\pi}(y-k)
    \Si\left(\lambda_n\frac{y-\tilde d(z)}{\sigma}\right)
    \one\{|y-k|\leq M\}}_{Q,2}\\
    &\qquad\leq
    C\lambda_n
    \left(\int
    (y-k)^2\one\{|y-k|\leq M\}
    \left(\frac{\delta(z;\beta)-\tilde d(z)}{\sigma}\right)^2
    dQ(y,z)\right)^{1/2}\\
    &\qquad\leq
    C\lambda_n
    \left(\int
    F_M(y,z)^2
    \left(\frac{\delta(z;\beta)-\tilde d(z)}{\sigma}\right)^2
    dQ(y,z)\right)^{1/2}\\
    &\qquad=
    C\lambda_n \norm{F_M}_{Q,2}
    \norm{(\delta(\cdot;\beta)-\tilde d)/\sigma}_{Q_{F_M},2} \\
    &\qquad\leq \eta \norm{F_M}_{Q,2}.
  \end{align*}
  Thus the pushforward of the $L_2(Q_{F_M})$ cover of $\cl{D}/\sigma$ through the map
  \[
    d/\sigma
    \mapsto
    \frac{1}{\pi}(y-k)
    \Si\left(\lambda_n\frac{y-d(z)}{\sigma}\right)
    \one\{|y-k|\leq M\}
  \]
  is an $\eta \norm{F_M}_{Q,2}$-cover of $\cl{G}_M$ in $L_2(Q)$.

  The class $\cl{D}/\sigma$ is VC-subgraph with index of order $V(\cl{D})$ and
  envelope $D/\sigma$. Since $D/\sigma\leq L_\sigma$ on the fixed design support,
  \[
    c\frac{\eta}{\lambda_n}
    \geq
    c\frac{\eta}{\lambda_n L_\sigma}\norm{D/\sigma}_{Q_{F_M},2}.
  \]
  Applying \cref{lemma:vc_subgraph_class_covering_number_bound} to
  $\cl{D}/\sigma$ gives
  \[
    \log N(\eta \norm{F_M}_{Q,2},\cl{G}_M,L_2(Q))
    \lesssim
    V(\cl{D})\log\frac{e\lambda_n L_\sigma}{\eta}.
  \]

  For $\mathcal S_M$, note that
  $\mathrm{sgn}((y-\delta(z;\beta))/\sigma)=2\one\{y\geq\delta(z;\beta)\}-1$.
  Therefore $\cl{S}_M$ is a VC-subgraph class with envelope $F_M/2$ and index
  $O(V(\cl{D}))$, so
  \[
    \log N(\eta \norm{F_M}_{Q,2},\cl{S}_M,L_2(Q))
    \lesssim
    V(\cl{D})\log\frac{e}{\eta}.
  \]
  The preceding product cover implies, after changing constants, that for
  every $0<\eta\leq1$,
  \[
    \log N(\eta\norm{F_M}_{Q,2},\cl{R}_M,L_2(Q))
    \lesssim
    V(\cl{D})\log\frac{e\lambda_n L_\sigma}{\eta}
    +
    V(\cl{D})\log\frac{e}{\eta}.
  \]
  Since $\lambda_n\geq1$ and $L_\sigma\geq1$, the second term is absorbed into
  the first one, so
  \[
    \sup_Q
    \log N(\eta\norm{F_M}_{Q,2},\cl{R}_M,L_2(Q))
    \lesssim
    V(\cl{D})\log\frac{e\lambda_n L_\sigma}{\eta},
  \]
  where the supremum is over discrete probability measures $Q$. Therefore
  \begin{align*}
    J(b,\cl{R}_M\mid F_M,L_2)
    &=
    \int_0^b
    \sup_Q
    \sqrt{1+\log N(\eta\norm{F_M}_{Q,2},\cl{R}_M,L_2(Q))}
    \,d\eta\\
    &\lesssim
    \sqrt{V(\cl{D})}
    \int_0^b
    \sqrt{\log\frac{e\lambda_n L_\sigma}{\eta}}\,d\eta.
  \end{align*}
  Finally, 
  \[
    \int_0^b\sqrt{\log\frac{e\lambda_n L_\sigma}{\eta}}\,d\eta
    =
    b\int_0^1\sqrt{\log\frac{e\lambda_n L_\sigma}{b}+\log\frac{1}{u}}\,du
    \lesssim
    b\sqrt{\log\frac{e\lambda_n L_\sigma}{b}},
  \]
  which gives
  \[
    J(b,\cl{R}_M\mid F_M,L_2)
    \lesssim
    b\sqrt{V(\cl{D})\log\frac{e\lambda_n L_\sigma}{b}},
  \]
  after enlarging the absolute constant.
\end{proof}

\begin{proposition}
    \label{prop:assure_term_empirical_process_bound}
    Let $\lambda_n=\sqrt{2\log n}$
    \[
      L_{\sigma,n}:=1\vee\max_{1\le i\le n}\frac{D(z_i)}{\sigma_i},
      \qquad
      \omega_{4,n}^4:=m_4^4+s_4^4,
      \qquad
      U_{H,n}:=\max_{1\le i\le n}\{\sigma_i+D(z_i)\}.
    \]
    Further define 
    \[ M:=\log(eL_{\sigma,n})\left\{
  V(\cl{D})(1+s_2+m_2+\norm{D}_n+U_{H,n})
  +\omega_{4,n}^4
  \right\}.\] 
  Then, in the proof of \cref{thm:main_regret_bound},
  \begin{align*}
      \E \sup_{\beta\in\mathcal B} |A_2(\beta)|
	      & \lesssim M \frac{(\log n )^{1/4}\sqrt{\log \log n}}{\sqrt{n}}.
  \end{align*}
    
  \end{proposition}

\begin{proof}
  Write
  \[
    S_n:=s_2+m_2,
 \quad 
    \rho_{\lambda_n}(t)
    :=
    \frac{1}{\pi}\Si(\lambda_n t)-\frac{1}{2}\mathrm{sgn}(t).
  \]
 Let
  \[
    P_i f:=\E[f(Y_i,z_i)],
    \qquad
    \bar P f:=\frac{1}{n}\sum_{i=1}^n P_i f,
    \qquad
    \bb{G}_n f_\beta:=\frac{1}{\sqrt n}\sum_{i=1}^n
    \{f_\beta(Y_i,z_i)-P_i f_\beta\}.
  \]

Define
  \begin{align*}
    q_\beta(y,z)
    &:=
    \frac{1}{2}(y-k)\mathrm{sgn}\left(\frac{y-\delta(z;\beta)}{\sigma}\right), \quad 
    \cl{Q}:=\set{q_\beta:\beta\in\mathcal B},\\
    r_\beta(y,z)
    &:=
    (y-k)\rho_{\lambda_n}\left(\frac{y-\delta(z;\beta)}{\sigma}\right), \quad \cl{R}:=\set{r_\beta:\beta\in\mathcal B},\\
    h_\beta(y,z)
    &:=
    \sigma\lambda_n\sinc\left(\lambda_n\frac{y-\delta(z;\beta)}{\sigma}\right), \quad \cl{H}:=\set{h_\beta:\beta\in\mathcal B}.
  \end{align*}
  We write $\norm{\mathbb G_n}_{\mathcal F} = \sup_{f\in \mathcal F} |\mathbb G_n f|$.
  We then have
  \begin{equation}
    \label{eq:propc2_decomposition}
    \E\sup_{\beta\in\mathcal B} |A_2(\beta)|
    \leq
   \underbrace{ \frac{1}{2}\E\left|\frac{1}{n}\sum_{i=1}^n(Y_i-\mu_i)\right|}_{\lesssim
   s_2/\sqrt{n}}
    +\frac{1}{\sqrt n}
    \left\{
      \E\norm{\bb{G}_n}_{\cl{Q}}
      +\E\norm{\bb{G}_n}_{\cl{R}}
      +\E\norm{\bb{G}_n}_{\cl{H}}
    \right\},
  \end{equation}

We begin with $\norm{\mathbb{G}_n}_{\mathcal Q}$.
  Since
  $\mathrm{sgn}((y-\delta(z;\beta))/\sigma)=2\one\{y\ge\delta(z;\beta)\}-1$,
  $\cl{Q}$ is VC-subgraph with index $O(V(\cl{D}))$ and envelope $|y-k|$. Therefore
  \cref{thm:inid_donsker_bound,lemma:vc_subgraph_class_covering_number_bound} give
  \begin{equation}
    \label{eq:propc2_q_bound}
    \E\norm{\bb{G}_n}_{\cl{Q}}
    \lesssim
    \sqrt{V(\cl{D})}(s_2+m_2).
  \end{equation}

Moving on to $\norm{\mathbb{G}_n}_{\mathcal R}$, 
  set $M_n:=n^{1/4}$ and truncate $r_\beta$:
  \[
    r_{\beta,M_n}(y,z):=r_\beta(y,z)\one\{|y-k|\le M_n\},
    \qquad
    \cl{R}_{M_n}:=\set{r_{\beta,M_n}:\beta\in\mathcal B}.
  \]
  For a sufficiently large absolute constant $C$,
  \[
    F_{R,M_n}(y,z)
    :=
    C\left(|y-k|\one\{|y-k|\le M_n\}+S_n\right)
  \]
  is an envelope for $\cl{R}_{M_n}$, satisfies
  \[
    \norm{F_{R,M_n}}_{\bar P,2}\asymp S_n,
    \qquad
    F_{R,M_n}\le C(M_n+S_n),
  \]
  and, by \cref{lemma:si_remainder_second_moment_bound},
  \begin{equation}
    \label{eq:propc2_r_second_moment}
    \sup_{\beta\in\mathcal B} \bar P r_{\beta,M_n}^2
    \le
    \sup_{\beta\in\mathcal B} \bar P r_\beta^2
    \lesssim
    \lambda_n^{-1}S_n^2
    \lesssim
    \lambda_n^{-1}\norm{F_{R,M_n}}_{\bar P,2}^2.
  \end{equation}
  Apply \cref{lemma:si_remainder_localized_entropy} with $F_M$ therein as $F_{R,M_n}$. We obtain
  \begin{equation}
    \label{eq:propc2_r_entropy}
    J\left(\lambda_n^{-1/2},\cl{R}_{M_n}\mid F_{R,M_n},L_2\right)
    \lesssim
    \lambda_n^{-1/2}
    \sqrt{V(\cl{D})\log(e\lambda_n^{3/2}L_{\sigma,n})}.
  \end{equation}
  Applying \cref{cor:inid_localized_donsker_bound}
  with the localization condition~\eqref{eq:propc2_r_second_moment} and~\eqref{eq:propc2_r_entropy} yields
  \begin{equation}
    \label{eq:propc2_r_truncated_bound}
    \E\norm{\bb{G}_n}_{\cl{R}_{M_n}}
    \lesssim
    \lambda_n^{-1/2}S_n
    \sqrt{V(\cl{D})\log(e\lambda_n^{3/2}L_{\sigma,n})}
    +
    \frac{(M_n+S_n)V(\cl{D})\log(e\lambda_n^{3/2}L_{\sigma,n})}{\sqrt n}.
  \end{equation}
  Moreover, since $\rho_{\lambda_n}$ is bounded,
  \begin{align}
    &\E\sup_{\beta\in\mathcal B}
    \left|
      \frac{1}{\sqrt n}\sum_{i=1}^n
      \left\{(r_\beta-r_{\beta,M_n})(Y_i,z_i)-P_i(r_\beta-r_{\beta,M_n})\right\}
    \right|
    \\
    & \le \E
    \left|
      \frac{1}{\sqrt n}\sum_{i=1}^n
      \left\{\sup_{\beta\in\mathcal B}|r_\beta-r_{\beta,M_n}|(Y_i,z_i)+P_i\sup_{\beta\in\mathcal B}|r_\beta-r_{\beta,M_n}|\right\}
    \right|\\
    & \lesssim
    \sqrt n\,\frac{\omega_{4,n}^4}{M_n^3}.\label{eq:propc2_r_tail}
  \end{align}
  Indeed,~\eqref{eq:propc2_r_tail} follows from
  $|r_\beta-r_{\beta,M_n}|\lesssim |y-k|\one\{|y-k|>M_n\}$,
  $\E |X|\one\{|X|>M_n\}\le \E |X|^4/M_n^3$, and
  $
    \frac{1}{n}\sum_{i=1}^n P_i|Y_i-k_i|^4
    \lesssim
    s_4^4+m_4^4
    =
    \omega_{4,n}^4.
  $

  Combining~\eqref{eq:propc2_r_truncated_bound} and~\eqref{eq:propc2_r_tail}, and
  substituting $M_n=n^{1/4}$, gives
  \begin{equation}
    \label{eq:propc2_r_bound}
    \E\norm{\bb{G}_n}_{\cl{R}}
    \lesssim
    \lambda_n^{-1/2}S_n
    \sqrt{V(\cl{D})\log(e\lambda_n^{3/2}L_{\sigma,n})}
    +
    \frac{(n^{1/4}+S_n)V(\cl{D})\log(e\lambda_n^{3/2}L_{\sigma,n})}{\sqrt n}
    +
    \frac{\omega_{4,n}^4}{n^{1/4}}.
  \end{equation}

 Finally, we bound $\norm{\bb{G}_n}_{\mathcal H}$. Let
  $
    F_H(y,z):=C\lambda_n(\sigma+D(z)).
  $
  For $C$ large enough this is an envelope for $\cl{H}$, and
  \[
    \norm{F_H}_{\bar P,2}\lesssim \lambda_n(s_2+\norm{D}_n),
    \qquad
    F_H\le C\lambda_n U_{H,n}.
  \]
  By \cref{lemma:sinc_second_moment_bound},
  \begin{equation}
    \label{eq:propc2_h_second_moment}
    \sup_{\beta\in\mathcal B} \bar P h_\beta^2
    \lesssim
    \lambda_n s_2^2
    \lesssim
    \lambda_n^{-1}\norm{F_H}_{\bar P,2}^2.
  \end{equation}
  Since $\sinc'$ is bounded,
  \[
    |h_\beta(y,z)-h_{\tilde\beta}(y,z)|
    \lesssim
    \lambda_n^2|\delta(z;\beta)-\delta(z;\tilde\beta)|.
  \]
  Thus, for every discrete probability measure $Q$ and $0<\eta<1$,
  \[
    N\left(\eta\norm{F_H}_{Q,2},\cl{H},L_2(Q)\right)
    \le
    N\left(c\frac{\eta}{\lambda_n}\norm{D}_{Q_z,2},\cl{D},L_2(Q_z)\right),
  \]
  where $Q_z$ is the marginal measure on $z$. By
  \cref{lemma:vc_subgraph_class_covering_number_bound},
  \begin{equation}
    \label{eq:propc2_h_entropy}
    J\left(\lambda_n^{-1/2},\cl{H}\mid F_H,L_2\right)
    \lesssim
    \lambda_n^{-1/2}\sqrt{V(\cl{D})\log(e\lambda_n)}.
  \end{equation}
  Applying \cref{cor:inid_localized_donsker_bound}
  with~\eqref{eq:propc2_h_second_moment} and~\eqref{eq:propc2_h_entropy} gives
  \begin{equation}
    \label{eq:propc2_h_bound}
    \E\norm{\bb{G}_n}_{\cl{H}}
    \lesssim
    \lambda_n^{1/2}\sqrt{V(\cl{D})\log(e\lambda_n)}(s_2+\norm{D}_n)
    +
    \frac{\lambda_n U_{H,n}V(\cl{D})\log(e\lambda_n)}{\sqrt n}.
  \end{equation}

  Combining~\eqref{eq:propc2_decomposition},~\eqref{eq:propc2_q_bound},~\eqref{eq:propc2_r_bound}, and~\eqref{eq:propc2_h_bound}, and using $V(\cl{D})\ge 1$, gives
  \begin{align*}
    \E \sup_{\beta\in\mathcal B} |A_2(\beta)|
    & \lesssim
    \frac{\sqrt{V(\cl{D})}}{\sqrt n}(s_2+m_2) \\
    & \quad
    +\frac{1}{\sqrt n}\lambda_n^{-1/2}(s_2+m_2)
    \sqrt{V(\cl{D})\log(e\lambda_n^{3/2}L_{\sigma,n})} \\
    & \quad
    +\frac{1}{\sqrt n}\lambda_n^{1/2}
    \sqrt{V(\cl{D})\log(e\lambda_n)}(s_2+\norm{D}_n) \\
    & \quad
    +\frac{(n^{1/4}+s_2+m_2)V(\cl{D})
    \log(e\lambda_n^{3/2}L_{\sigma,n})}{n}
    +\frac{\omega_{4,n}^4}{n^{3/4}} \\
    & \quad
    +\frac{\lambda_n U_{H,n}V(\cl{D})\log(e\lambda_n)}{n}.
  \end{align*}
  Finally, since $\lambda_n=\sqrt{2\log n}$,
  $\log(eL_{\sigma,n})\ge1$, and $V(\cl{D})\ge1$, the last display implies
  \[
    \E \sup_{\beta\in\mathcal B} |A_2(\beta)|
    \lesssim
    M\frac{{(\log n)}^{1/4}\sqrt{\log\log n}}{\sqrt n},
  \]
  after absorbing the terms of order $n^{-3/4}$ and $n^{-1}$ into the displayed rate
  for all sufficiently large $n$, and then enlarging the constant.
\end{proof}

  \subsection{Tail bound for regret}

  In addition to controlling the expected welfare, empirical process techniques can also get
  high probability guarantees on the \textit{realized} in-sample regret, defined as the
  quantity
  \[
    u(\beta^*) - u(\hat{\beta}).
  \] We give a basic tail bound result for \assure{*} using the i.n.i.d. extension of \citet
  {vaart2023empirical}, Theorem 2.14.23  from \cref{sec:inid_theory}.

  \begin{theorem}
    \label{thm:high_probability_regret_bound}
    Fix $\delta > 0$. Under the assumptions of \cref{thm:main_regret_bound}, with probability at least $1 - \delta$
    \[
      u(\beta^*) - u(\hat{\beta}) \lesssim \frac{m_1}{n\log n} + \frac{m_2 + (M+s_2)\sqrt{\log n}}{\sqrt{n}} \sqrt{\log \frac{2}{\delta}}
    \]
    where $M$ is the constant in the statement of \cref{prop:assure_term_empirical_process_bound}. The same bound is true for $W(\beta^*) - W(\hat{\beta}).$
  \end{theorem}


  \begin{proof}
    The optimality of $\hat\beta$ for $\hat{W}_n(\cdot)$ shows that
    \[u(\beta^*) - u(
      \hat{\beta}) \leq u(\beta^*) - \hat{W}_n(\beta^*) + \hat{W}_n(\hat{\beta}) - u(
    \hat{\beta}) \le 2\sup_{\beta\in\mathcal B} |u(\beta) - \hat{W}_n(\beta)|\]
    Combined with the
    uniform bias bound on $\hat{W}_n$ from \cref{thm:assure_bias}, we have
    \[
      u(\beta^*) - u(\hat{\beta}) \leq 2\sup_{\beta\in\mathcal B} \left|u(\beta) - \hat{W}_n(\beta) -
      \E\left[u(\beta) - \hat{W}_n(\beta) \right] \right| + \frac{2m_1}{n \log n}.
    \]
    Let \[\norm{\bb{G}_n}_{\cl{F}} := \sqrt{n}\sup_{\beta\in\mathcal B} \abs{u(\beta) -
    \hat{W}_n (\beta) - \E\left[u(\beta) - \hat{W}_n(\beta) \right]}.\]

    Here, $\cl{F}$ is the function class with members, indexed by $\beta\in\mathcal B$,
    \[
      f_\beta(y,\mu,z) = (\mu - k)\mathbf{1}\set{y > \delta(z;\beta)} - w_h(y;z,\beta)
    \]
    with $z = (k,\sigma,x)$ and $w_h$ defined in \eqref{eq:assure_term}. We will control the sub-Gaussian norm of $\norm{\bb{G}_n}_{\cl{F}}$. Notice that an envelope function
    for this class is given by $|\mu - k| + C|y-k| + \lambda_n \sigma$, which is in turn
    bounded above by
    \[
      F(y,\mu,k,\sigma) := C(|\mu - k| + |y - \mu| + \lambda_n\sigma)
    \]
    for some constant $C$. Applying \cref{prop:orlicz_inid_bounds} with $V_i = (Y_i,\mu_i,z_i)$, the sub-Gaussian norm is bounded above by a constant multiple times
    \[
      \Estar \norm{\bb{G}_n}_{\cl{F}} + \left(\frac{1}{n} \sum_{i=1}^n (\mu_i - k_i)^2 + \lambda_n^2 \sigma_i^2 \right)^{1/2} \lesssim \Estar \norm{\bb{G}_n}_{\cl{F}} + m_2 + s_2 \lambda_n
    \]

    By definition of the Orlicz norm $\psi_2$, any random variable satisfies $\Prob(|Y| > t) \leq 2\exp\left(-t^2/\norm{Y}^2_{\psi_2}\right)$. Equivalently, with probability $\geq 1 - \delta$,
    \[
      |Y| < \norm{Y}_{\psi_2} \sqrt{\log \frac{2}{\delta}}.
    \]
    Combining this with the bound on \cref{eq:regret_decompose_main_bound} in \cref{thm:main_regret_bound}
    yields
    \[
      u(\beta^*) - u(\hat{\beta}) \lesssim \frac{m_1}{n\log n} + \frac{M (\log n)^{1/4} \sqrt{\log \log n} + m_2 + s_2\sqrt{\log n}}{\sqrt{n}} \sqrt{\log \frac{2}{\delta}},
    \]
    from which the statement follows.

    The analogous result follows for $W(\beta^*) - W(\hat
    {\beta})$ by first noting that
    \[W(\beta^*) - W(\hat{\beta}) \leq W(\beta^*) - \hat{W}_n(\beta^*) + \hat{W}_n(\hat{\beta}) - W(\hat{\beta})\le 2\sup_{\beta\in\mathcal B}
    |\hat{W}_n(\beta) - W(\beta)|,\]
    and then repeating the same steps of the proof.
  \end{proof}

  \subsection{Auxiliary results for \cref{thm:fast_rates}}
  \label{sec:proof_of_fast_rate}

  We start with a lemma calculating the derivatives of the average welfare. Since the computation is straightforward, we omit the proof.

  \begin{lemma}[Derivatives of the welfare]
    \label{lemma:third_derivative_of_welfare}
    Assume the decision thresholds $\delta(\cdot,Z)$ are smooth for each $Z$, and set $T_i(\beta) := (\mu_i - \delta(z_i;\beta))/\sigma_i$. Then the derivatives of the averaged welfare are given by
    \begin{align*}
      W'(\beta) & := -\frac{1}{n}\sum_{i=1}^n (\mu_i - k_i)\varphi\left( T_i(\beta)\right)\frac{\delta'(\beta,z_i)}{\sigma_i} \\
      W''(\beta) & := -\frac{1}{n}\sum_{i=1}^n (\mu_i - k_i)\left[\varphi\left( T_i(\beta)\right)\frac{\delta''(\beta,z_i)}{\sigma_i} - \varphi'\left( T_i(\beta)\right)\frac{\delta'(\beta,z_i)^2}{\sigma_i^2}  \right]
    \end{align*}
    and
    \begin{align*}
      W'''(\beta) & := -\frac{1}{n}\sum_{i=1}^n (\mu_i - k_i) A_i \\
      A_i & := \varphi\left( T_i(\beta)\right)\frac{\delta'''(\beta,z_i)}{\sigma_i}  - 3\varphi'\left( T_i(\beta)\right) \frac{\delta''(\beta,z_i)\delta'(\beta,z_i)}{\sigma_i^2} \\
      & + \varphi''\left( T_i(\beta)\right)\frac{\delta'(\beta,z_i)^3}{\sigma_i^3}.
    \end{align*}
    Under Assumption \ref{assmp:boundedness}, there exists a constant $C_B$ that only depends on $B$ such that
    \[\sup_{\beta} |W'(\beta)| + |W''(\beta)| + |W'''(\beta)| \le C_B.\]
  \end{lemma}

  \begin{theorem}
    \label{thm:small_external_regret}
    Under the assumptions of \cref{thm:fast_rates}, for $n$ large enough, we have
    \[
      \abs{\hat{\beta} - \beta^*} \le \frac{2C}{\kappa}\frac{(\log n)^3}{\sqrt{n}}
    \]
    with probability at least $1 - \frac{2}{n},$ for a constant $C$ depending only on $B$ and the VC indices of $\cl{D},\cl{D}'$. As a result,
    \[
      \E\left[W(\beta^*) - W(\hat{\beta})\right] \leq C\frac{(\log n)^6}{n}
    \]
    for an absolute constant $C$ depending only on $B,\kappa,$ and the VC indices of $\cl{D},\cl{D}'.$
  \end{theorem}
  \begin{proof}
    Firstly, inspect the derivatives of the welfare from \cref{lemma:third_derivative_of_welfare}.
    By the results of \cref{thm:high_probability_regret_bound} we have with probability at least $1 - \delta$
    \[
      W(\beta^*)  - W(\hat{\beta}) \leq c\left\{\frac{m_1}{n\log n} + \frac{ m_2 + (M+s_2)\sqrt{\log n}}{
      \sqrt{n}} \sqrt{\log \frac{2}{\delta}}\right\},
    \]
    where $c$ is a universal constant hidden in Theorem \ref{thm:high_probability_regret_bound} and $M$ is the constant in the statement of \cref{prop:assure_term_empirical_process_bound}. Substituting $\delta = 1/n$ and using \cref{assmp:boundedness}, we have
    \begin{equation}
      \label{eq:high_prob_bound_welfare_fast_rate_externel_regret_helper}
      W(\beta^*)  - W(\hat{\beta})\leq CB \frac{\log n}{\sqrt{n}}
    \end{equation}
    with probability at least $1 - 1/n$ and some absolute constant $C.$ Let $A_n$ be this event.

    Take $n$ large enough such that the right hand side of Eq. \eqref{eq:high_prob_bound_welfare_fast_rate_externel_regret_helper} is less than $\xi$. Using \cref{assmp:well_sep_max}, we conclude that on $A_n$, $|\beta^* - \hat{\beta}| \leq \kappa/B$.
    Now, Taylor expand $W'$ around $\beta^*$ to obtain
    \begin{equation}
      \label{eq:welfare_dv_taylor_expansion_second_term}
      W'(\beta^*) - W'(\hat{\beta}) = -W''(\beta^*)(\beta^* - \hat{\beta}) - \frac{W'''(\tilde{\beta})}{2}(\beta^* - \hat{\beta})^2
    \end{equation}
    for some $\tilde{\beta}$ between $\hat{\beta}$ and $\beta^*$. By Lemma \ref{lemma:third_derivative_of_welfare} and the reverse triangle inequality $|a-b| \geq |a| - |b|$ we have
    \begin{align*}
      \left| W'(\beta^*) - W'(\hat{\beta}) \right| & \geq \left| W''(\beta^*)\right| \left| \beta^* - \hat{\beta}\right| -  \frac{|W'''(\tilde{\beta})|}{2}(\beta^* - \hat{\beta})^2 \\
      & \geq \kappa |\beta^* - \hat{\beta}| - \frac{C_B}{2} (\beta^* - \hat{\beta})^2,
    \end{align*}
    where the last line follows from \cref{assmp:curvature}. Combining this with the Taylor expansion in Eq. \eqref{eq:welfare_dv_taylor_expansion_second_term}, for $n$ large enough, on the event $A_n$,
    \begin{equation}
      \label{eq:beta_quadratic_bound}
      \frac{\kappa}{2}\left|\beta^* - \hat{\beta} \right| \leq \left| W'(\beta^*) - W'(\hat{\beta}) \right|.
    \end{equation}

    Next, by \cref{thm:derivative_external_regret},  with probability at least $1 - 1/n$
    \begin{equation}
      \label{eq:application_of _welfare_derivative_bound}
      \left|W'(\hat{\beta})  - W'(\beta^*)\right| \leq \frac{C}{n\sqrt{\log n}} + \frac{ C(\log n)^{5/2}}{\sqrt{n}} \sqrt{\log(2n)}
    \end{equation}
    for a constant $C$ depending only on $B$ and the VC indices of $\cl{D},\cl{D}'$. Let $\tilde{A}_n$ denote this event. Combining equations \eqref{eq:beta_quadratic_bound} and \eqref{eq:application_of _welfare_derivative_bound}, we conclude that on $A_n \cap \tilde{A}_n$, which has probability at least $1 - \frac{2}{n},$ that
    \[
      \left|\beta^* - \hat{\beta} \right| \leq \frac{2C}{\kappa}\frac{(\log n)^3}{\sqrt{n}}
    \]
    for a constant $C$ depending only on $B$ and the VC indices of $\cl{D},\cl{D}'$. This proves the first claim. \\

    Next, notice that $W'(\beta^*) = 0$. By Taylor expansion of $W$ around $\beta^*$, we have
    \[
      W(\beta^*) - W(\hat{\beta}) \lesssim C_B(\hat{\beta} - \beta^*)^2 + C_B(\hat{\beta} - \beta^*)^3.
    \]
    Then for $n$ large enough,
    \begin{align*}
      \E\left[W(\beta^*) - W(\hat{\beta})\right] & \leq \E\left[W(\beta^*) - W(\hat{\beta}); A_n \cap \tilde{A}_n\right] + \E\left[W(\beta^*) - W(\hat{\beta}); (A_n \cap \tilde{A}_n)^c \right] \\
      & \lesssim C_B\E\left[(\beta^* - \hat{\beta})^2; A_n \cap \tilde{A}_n\right] + B\Prob\left((A_n \cap \tilde{A}_n)^c\right) \\
      & \lesssim \frac{C_B C^2}{\kappa^2} \frac{(\log n)^6}{n} + \frac{B}{n}.
    \end{align*}
    This concludes the proof.
  \end{proof}

  \begin{theorem}[Uniform bounds on the derivative bias]
    \label{thm:derivative_bias}
    We have
    \begin{align*}
      & |W'(\beta) - \E \hat{W}_n'(\beta)| \lesssim \left(\frac{1}{n}\sum_{i=1}^n \frac{|\mu_i - k_i|}{\sigma_i}|\delta'(\beta,z_i)|\right) h e^{-1/(2h^2)} \\
      & |W''(\beta) - \E \hat{W}_n''(\beta)| \lesssim \\
      & \quad \quad  \left[h\left(\frac{1}{n}\sum_{i=1}^n \frac{|\mu_i - k_i|}{\sigma_i}|\delta''(\beta,z_i)| \right) + \left(\frac{1}{n}\sum_{i=1}^n \frac{|\mu_i - k_i|}{\sigma_i^2}\delta'(\beta,z_i)^2 \right) \right] e^{-1/(2h^2)}
    \end{align*}

    With the choice of $h = 1/\sqrt{2\log n}$ and uniform upper bounds on $\delta'$, the bias for the first derivative is on the order of $\frac{1}{n\sqrt{\log n}} \left(\frac{1}{n}\sum_{i=1}^n \frac{|\mu_i - k_i|}{\sigma_i}\right)$.
  \end{theorem}

  \begin{proof}[Proof of \cref{thm:derivative_bias}]
    \Cref{lemma:expected_assure} shows
    \begin{align*}
      (\mu-k) \Phi\left(\frac{\mu - C}{\sigma} \right) & = \frac{1}{2} (\mu - k) + \frac{1}{\pi} (\mu - k) \int_0^\infty  \frac{1}{\omega} e^{-\frac{1}{2}\omega^2} \sin(\omega (\mu - C)/\sigma)d\omega \\
      \E \Psi_h(Y,z,C) & = \frac{1}{2} (\mu - k) + \frac{1}{\pi} (\mu - k) \int_0^{1/h}\frac{1}{\omega} e^{-\frac{1}{2}\omega^2} \sin(\omega (\mu - C)/\sigma)d\omega
    \end{align*}
    We may differentiate under the Lebesgue integral in $C$, which is justified by dominated convergence theorem. This yields
    \begin{align*}
      -\frac{\mu-k}{\sigma} \varphi\left(\frac{\mu - C}{\sigma} \right) & =  -\frac{1}{\pi} \frac{\mu - k}{\sigma} \int_0^\infty  e^{-\frac{1}{2}\omega^2} \cos(\omega (\mu - C)/\sigma)d\omega \\
      \frac{d}{dC} \E \Psi_h(Y,z,C) & = -\frac{1}{\pi} \frac{\mu - k}{\sigma} \int_0^{1/h}  e^{-\frac{1}{2}\omega^2} \cos(\omega (\mu - C)/\sigma)d\omega.
    \end{align*}
    By the chain rule, we have
    \begin{align*}
      \hat{W}_n'(\beta) & = \frac{1}{n}\sum_{i=1}^n \frac{d\Psi_{h}}{dC}(Y_i,z_i,\delta(z_i;\beta))\delta'(\beta,z_i)
    \end{align*}
    By the formulas for derivatives of $\Psi$ in \cref{lemma:derivatives of Psi} and
    \cref{lemma:sinc_properties} shows, $|\frac{d\Psi_{h}}{dC}| \leq K_h |Y-k|/\sigma$ for some constant $K_h$ depending on $h$. This provides a dominating integrable function which justifies DCT, yielding $\frac{d}{dC} \E \Psi_h = \E \frac{d\Psi_{h}}{dC}$. Therefore,
    \begin{align*}
      |W'(\beta) - \E \hat{W}'_n(\beta)| & \leq \frac{1}{\pi}\left( \frac{1}{n}\sum_{i=1}^n\frac{|\mu_i - k_i|}{\sigma_i}\delta'(\beta,z_i)\right) \int_{1/h}^\infty e^{-\frac{1}{2}\omega^2} \cos(\omega (\mu - C)/\sigma)d\omega  \\
      & \lesssim \left( \frac{1}{n}\sum_{i=1}^n\frac{|\mu_i - k_i|}{\sigma_i}\delta'(\beta,z_i)\right)he^{-1/2h^2}
    \end{align*}
    using the constant upper bound for $\cos$ and the Mills ratio bound.

    For the second derivative, the argument is similar. For notational clarity we will write $\Psi_C, \Psi_{CC}$ to be the first and second derivatives in $C$ of $\Psi_h$. Applying DCT again to justify the interchange of derivative and integral, we obtain
    \begin{align*}
      \frac{\mu-k}{\sigma^2} \varphi'\left(\frac{\mu - C}{\sigma} \right) & =  -\frac{1}{\pi} \frac{\mu - k}{\sigma^2} \int_0^\infty  \omega e^{-\frac{1}{2}\omega^2} \sin(\omega (\mu - C)/\sigma)d\omega \\
      \E \Psi_{CC}(Y,z,C) & = -\frac{1}{\pi} \frac{\mu - k}{\sigma^2} \int_0^{1/h}  \omega e^{-\frac{1}{2}\omega^2} \sin(\omega (\mu - C)/\sigma)d\omega.
    \end{align*}
    Next, notice that the second derivative of \assure{*} is given by
    \begin{equation}
      \hat{W}''_n(\beta) = \frac{1}{n}\sum_{i=1}^n \Psi_{CC}(Y_i,z_i,\delta(z_i;\beta))\delta'(\beta,z_i)^2 + \Psi_{C}(Y_i,z_i,\delta(z_i;\beta))\delta''(\beta,z_i).
    \end{equation}
    Thus,
    \begin{align*}
      |\E\hat{W}''_n(\beta) -  W''(\beta)| & \lesssim \left(\frac{1}{n}\sum_{i=1}^n \frac{|\mu_i - k_i|}{\sigma_i^2}\delta'(\beta,z_i)^2 \right) \sup_\mu \int_{1/h}^\infty \omega e^{-\frac{1}{2}\omega^2} |\sin(\omega (\mu - C)/\sigma)|d\omega \\
      & + \left(\frac{1}{n}\sum_{i=1}^n \frac{|\mu_i - k_i|}{\sigma_i}|\delta''(\beta,z_i)| \right) \sup_\mu \int_{1/h}^\infty e^{-\frac{1}{2}\omega^2} |\cos(\omega (\mu - C)/\sigma)|d\omega.
    \end{align*}
    The conclusion follows from bounding $\cos,\sin$ above by $1$ and bounding the resulting integrals.
  \end{proof}

  \begin{theorem}
    \label{thm:LOO_stability}
    In the setting and notation of \cref{thm:fast_rates}, we have
    \[
      \Prob\left(\exists i: |\hat{\beta}^{(-i)} - \hat{\beta}| > \frac{4MB \log n}{\kappa n} \right) = O\left(\frac{1}{n} \right),
    \]
    for some universal constant $M>0$.
  \end{theorem}
  \begin{proof}[Proof of \cref{thm:LOO_stability}]
    Recall the function $\Psi$ in the \assure{*} estimator of \eqref{eq:assure_term_Psi}. Let
    $\Psi_C, \Psi_{CC}$ denote the first and second partial derivatives in
    $C$, given in
    \cref{lemma:derivatives of Psi}. To control $\hat{\beta} - \hat{\beta}^{
    (-i)}$, we will apply a standard Taylor expansion using the first-order conditions
    \begin{align}
      & \sum_{i=1}^n \Psi_C(Y_i,z_i, \delta(\hat{\beta},z_i))\delta'(\hat{\beta},z_i) = 0 \\
      & \sum_{j \neq i} \Psi_C(Y_j,z_j, \delta(\hat{\beta}^{(-i)},z_j))\delta'(\hat{\beta}^{(-i)},z_j) = 0, \quad \forall i=1,\dots,n.
    \end{align}
    By Taylor's theorem applied to $
    \sum_{i=1}^n \Psi_C(Y_i,z_i,\delta(\bullet,z_i))\delta'(\bullet,z_i),
    $
    we have
    \begin{align*}
      0 & = \frac{1}{n}\sum_{i=1}^n \Psi_C(Y_i,z_i,\delta(\hat{\beta},z_i))\delta'(\hat{\beta},z_i) \\
      & = \frac{1}{n}\sum_{j=1}^n \Psi_C(Y_j,z_j,\delta(\hat{\beta}^{(-i)},z_j))\delta'(\hat{\beta}^{(-i)},z_j) + \Delta_{n,i} (\hat{\beta} - \hat{\beta}^{(-i)}) \\
      & = \frac{1}{n}\Psi_C(Y_i,z_i,\delta(\hat{\beta}^{(-i)},z_i))\delta'(\hat{\beta}^{(-i)},z_i) + \Delta_{n,i}(\hat{\beta} - \hat{\beta}^{(-i)}).
    \end{align*}
    where
    \begin{align}
      \Delta_{n,i} & := \left(\frac{1}{n}\sum_{j=1}^n \Psi_{CC}(Y_j,z_j,\delta(\tilde{\beta}_i,z_j))\delta'(\tilde{\beta}_i,z_j)^2 + \Psi_C(Y_j,\tilde{\beta}_i)\delta''(\tilde{\beta}_i,z_j) \right) \\
      & = \hat{W}_n''(\tilde{\beta}_i).
    \end{align}
    for some $\tilde{\beta}_i$ between $\hat{\beta}$ and $\hat{\beta}^{(-i)}$. By \cref{assmp:boundedness},  $|\delta'(\hat{\beta}^{(-i)},z_i)| \leq B$. By Lemma \ref{lemma:derivatives of Psi}, it is not difficult to show that
    \[
      |\Psi_C(Y_i,z_i,\delta(\hat{\beta}^{(-i)},z_i))| \lesssim \lambda_n^2 + \lambda_n |\delta(\hat{\beta}^{(-i)},z_i)| \le M B \lambda_n^2,
    \]
    for some universal constant $M>0$. Thus,
    \begin{equation}
      \label{eq:LOO_diff_bound}
      \left| \hat{\beta} - \hat{\beta}^{(-i)}\right| \le\frac{MB\lambda_n^2}{n |
      \Delta_{n,i}|}.
    \end{equation}
    By \cref{lemma:loo_curvature_bound} and a union bound, with probability at least
    $1 - O(n^{-1})$ we have $|\Delta_{n,i}| \geq \kappa/2$ for all $i$. On this event,
    \eqref{eq:LOO_diff_bound} yields
    \[
      \left| \hat{\beta} - \hat{\beta}^{(-i)}\right| \le
      \frac{2MB\lambda_n^2}{\kappa n}, \quad
    \]
    Since $\lambda_n^2 = 2 \log n$, enlarging the absolute constant $M$ if
    necessary gives
    \[
      \left| \hat{\beta} - \hat{\beta}^{(-i)}\right| \leq
      \frac{4MB \log n}{\kappa n}.
    \]
    This proves the theorem.
  \end{proof}

  \begin{lemma}
    \label{lemma:loo_curvature_bound}
    In the setting and notation of \cref{thm:LOO_stability}, let $\tilde{\beta}_i$
    be any point between $\hat{\beta}$ and $\hat{\beta}^{(-i)}$, and define
    \[
      \Delta_{n,i} := \hat{W}_{n}''(\tilde{\beta}_i).
    \]
    Then, for $n$ large enough,
    \[
      \Prob\left(|\Delta_{n,i}| \geq \kappa/2\right) \geq 1 - O\left(\frac{1}{n^2}\right).
    \]
  \end{lemma}
  \begin{proof}
    Write
    \[
      I_n := \sup_\beta \left|\hat{W}_n''(\beta) - \E\left[\hat{W}_n''(\beta)\right]\right|,
      \quad
      II_n := \sup_\beta \left|\E\left[\hat{W}_n''(\beta)\right] - W''(\beta)\right|,
    \]
    \[
      III_{n} := C_B \max_{i}\max\left(|\beta^* - \hat{\beta}^{(-i)}|, |\beta^* - \hat{\beta}|\right),
    \]
    where $C_B$ is a constant depending only on $B$. Since $W'''$ is uniformly bounded by $C_B$ and $\tilde{\beta}_i$ lies between
    $\hat{\beta}$ and $\hat{\beta}^{(-i)}$, we have
    \[
      \left|W''(\tilde{\beta}_i) - W''(\beta^*)\right| \leq III_n.
    \]
    Hence, by the reverse triangle inequality and \cref{assmp:curvature},
    \[
      |\Delta_{n,i}| \geq \kappa - I_n - II_n - III_n.
    \]

    The three error terms are all $o(1)$ with probability $1 - O(n^{-1})$.
    First, by \cref{assmp:boundedness} and \cref{thm:derivative_bias},
    \[
      II_n \leq \left(\frac{B^3}{\sqrt{2\log n}} + B^5\right)\frac{1}{n} = O(n^{-1}).
    \]
    Second, \cref{thm:second_derivative_supremum} with $\delta=n^{-1}$ gives
    \[
      I_n \leq C\frac{(\log n)^4}{\sqrt{n}}
    \]
    with probability at least $1 - 1/n$. Third,
    \cref{thm:small_external_regret} gives
    \[
      |\hat{\beta} - \beta^*| \leq C\frac{(\log n)^3}{\sqrt{n}}
    \]
    with probability at least $1 - 2/n$. Let $\beta^{(-i)*}$ maximize the
    leave-one-out welfare $W^{(-i)}$. The proof of
    \cref{thm:small_external_regret} applies verbatim to $W^{(-i)}$, except that $\delta$ is set to be $n^{-2}$, so
    \[
      |\hat{\beta}^{(-i)} - \beta^{(-i)*}| \leq C\frac{(\log n)^3}{\sqrt{n}}
    \]
    with probability at least $1 - O(n^{-2})$. Applying the union bound, we obtain that, with probability $1 - O(n^{-1})$
    \[|\hat{\beta}^{(-i)} - \beta^{(-i)*}| \leq C\frac{(\log n)^3}{\sqrt{n}}, \quad \forall i = 1,\ldots, n.\]

    Moreover, $W$ and $W^{(-i)}$
    differ uniformly by $O(n^{-1})$ under \cref{assmp:boundedness}, hence
    \[
      W(\beta^*) - W(\beta^{(-i)*}) = O(n^{-1}).
    \]
    Since $W''(\beta^*) < -\kappa$ and $W'''$ is uniformly bounded,
      there exists $r > 0$ such that $W''(\beta) \leq -\kappa/2$ for all $|\beta -
      \beta^*| \leq r$. Because $\Theta$ is compact and $\beta^*$ is the unique
    global maximizer of $W$,
    \[
      \xi_r := W(\beta^*) - \sup_{|\beta - \beta^*| \geq r} W(\beta) >
      0.
    \]
    For $n$ large enough, the bound above is smaller than $\xi_r$, so
    $|\beta^{(-i)*} - \beta^*| < r$. Therefore Taylor'xs theorem gives
    \[
      W(\beta^*) - W(\beta^{(-i)*}) = -\frac{1}{2}W''(\bar{\beta}_i)
        \left(\beta^{(-i)*} - \beta^*\right)^2 \geq \frac{\kappa}{4}\left|\beta^{(-i)*}
      - \beta^*\right|^2
    \]
    for some $\bar{\beta}_i$ between $\beta^*$ and
    $\beta^{(-i)*}$. Consequently,
    \[
      |\beta^{(-i)*} - \beta^*| = O(n^{-1/2}),
    \]
    and, by a union bound,
    \[
      III_n = O\left(\frac{(\log n)^3}{\sqrt{n}}\right)
    \]
    with probability $1 - O(n^{-1})$.

    For $n$ large enough, the sum of the bounds for $I_n$, $II_n$, and $III_n$ is
    at most $\kappa/2$. The conclusion follows by a union bound.
  \end{proof}

  \begin{theorem}
    \label{thm:derivative_external_regret}
    Under the same assumptions of \cref{thm:fast_rates},
    \[
      \E|W'(\beta^*) - W'(\hat{\beta})| \leq \frac{C}{n\sqrt{\log n}}  + \frac{C(\log n)^{5/2}}{\sqrt{n}} \left(\sqrt{\cl{V}(\cl{D}')} + \sqrt{\cl{V}(\cl{D})} \right).
    \]
    for an absolute constant $C$ depending on $B.$ Moreover, fix any $\delta \in (0,1)$. Then with probability at least $1-\delta$, we have
    \[
      |W'(\beta^*) - W'(\hat{\beta})| \leq \frac{C}{n\sqrt{\log n}} + \frac{C(\log n)^{5/2}}{\sqrt{n}}\left(\sqrt{\cl{V}(\cl{D}')} + \sqrt{\cl{V}(\cl{D})} \right)\sqrt{\ln \frac{2}{\delta}}.
    \]
  \end{theorem}

  \begin{proof}[Proof of \cref{thm:derivative_external_regret}]
    By \cref{assmp:unimodal}, $W'(\beta^*) = 0$ and by definition we have $\hat{W}'(\hat{\beta}) = 0$. Therefore,
    \begin{align*}
      |W'(\beta^*) - W'(\hat{\beta})| & = |W'(\hat{\beta})| \\
      & = |W'(\hat{\beta}) - \hat{W}_n'(\hat{\beta})| \\
      & \leq \sup_\beta |W'(\beta) - \hat{W}_n'(\beta)|,
    \end{align*}
    which can in turn be bounded above by
    \begin{equation}
      \sup_\beta |W'(\beta) - \E \hat{W}_n'(\beta)| + \sup_\beta | \hat{W}_n'(\beta) - \E\hat{W}_n'(\beta)|.
    \end{equation}
    The first term represents the bias of the derivative of the \assure{*} estimator, which we bounded in \cref{thm:derivative_bias} by, up to universal constants,
    \[
      \sup_\beta \left(\frac{1}{n}\sum_{i=1}^n \frac{|\mu_i - k_i|}{\sigma_i}|\delta'(\beta,z_i)|\right) \frac{1}{n\sqrt{2 \log n}}.
    \]
    We will focus on bounding the second term $\sup_\beta | \hat{W}_n'(\beta) - \E\hat{W}_n'(\beta)|$ using i.n.i.d. empirical process theory. The proof structure follows exactly the proof of \cref{thm:main_regret_bound}. Recalling the notation in Eq. \eqref{eq:assure_summand_first_derivative}, we need to bound the quantity
    \[
      n^{-1/2}\sup_\beta \left| \frac{1}{\sqrt{n}}\sum_{i=1}^n \left(\Psi_C(Y_i,z_i,\delta(z_i;\beta))\delta'(\beta,z_i) - \E\left[\Psi_C(Y_i,z_i,\delta(z_i;\beta))\delta'(\beta,z_i)\right] \right) \right|
    \]
    where recall that $\Psi_C(Y_i,z_i,C)$ is given by
    \[
      -\lambda_n\frac{(Y_i - k_i)}{\sigma_i}\sinc\left(\lambda_n\left(\frac{Y_i - C}{\sigma_i} \right) \right) + \lambda_n^2 \sinc'\left(\lambda_n\left(\frac{Y_i - C}{\sigma_i} \right) \right).
    \]
    Throughout this proof, let $\cl{F}^{(1)}$ denote the function class consisting of the functions
    \begin{equation}\label{eq:derivative_F1}
      f^{(1)}_\beta(y,z) = -\lambda_n\left(\frac{y-k}{\sigma} \right)\sinc\left(\lambda_n\left(y - \delta(z;\beta)\right)/\sigma \right) \delta'(\beta,z)
    \end{equation}
    indexed by $\beta$, and let $\cl{F}^{(2)}$ denote the function class consisting of the functions
    \begin{equation}\label{eq:derivative_F2}
      f^{(2)}_\beta(y,z) = \lambda_n^2 \sinc'\left(\lambda_n\left(y - \delta(z;\beta)\right)/\sigma  \right)\delta'(\beta,z).
    \end{equation}
    Let $\norm{\bb{G}}_{\cl{F}^{(1)}} = \sup_\beta |\frac{1}{\sqrt{n}} \sum_{i=1}^n f^{(1)}_\beta(Y_i,z_i)- \E f^{(1)}_\beta(Y_i,z_i)|$ and define the analogous quantity $\norm{\bb{G}}_{\cl{F}^{(2)}}$. Finally, let $\cl{F}$ denote the sum function class $\cl{F}^{(1)} + \cl{F}^{(2)}$, comprised of the functions $\Psi_C(y,z,\delta(z;\beta))\delta'(\beta,z)$. \\

    Using a uniform bound for $\sinc$ and $\sinc'$ alongside \cref{assmp:boundedness}, envelope functions for $\cl{F}^{(1)},\cl{F}^{(2)}$ are given by
    \begin{align*}
      F^{(1)}(y) & := BL \lambda_n \cdot\frac{|y-k|}{\sigma}    \\
      F^{(2)}(y) & := BL \lambda_n^2.
    \end{align*}
    for some absolute constant $L$. \cref{prop:assure_derivative_sinc_term_upper_bound,prop:assure_derivative_sinc_dv_term_upper_bound} bound these terms, yielding
    \[
      \E\norm{\bb{G}}_{\cl{F}} \leq C(\log n)^{5/2}\left(\sqrt{\cl{V}(\cl{D}')} + \sqrt{\cl{V}(\cl{D})} \right).
    \]
    Combining these propositions and the bias bound, we obtain
    \[
      \E \sup_\beta |W'(\beta) - \hat{W}_n'(\beta)| \leq \frac{B}{n\sqrt{\log n}}  + \frac{C(\log n)^{5/2}}{\sqrt{n}} \left(\sqrt{\cl{V}(\cl{D}')} + \sqrt{\cl{V}(\cl{D})} \right)
    \]
    for an absolute constant $C$. This is the first claim.

    To get the second claim, we will apply \cref{prop:orlicz_inid_bounds} with $V_i = (Y_i,z_i)$ using the envelope function
    \[
      C\left( \lambda_n\frac{|y - k|}{\sigma} + \lambda_n^2\right).
    \]
    The random variable $\lambda_n|Y_i - k_i|/\sigma_i + \lambda_n^2$ is easily seen to be sub-Gaussian with parameter  $\lesssim \lambda_n^2 + \lambda_n|\mu_i - k_i|/\sigma_i$. This shows
    \begin{align*}
      \norm{\norm{\bb{G}}_{\cl{F}}}_{\psi_2} & \lesssim \E\norm{\bb{G}}_{\cl{F}}  + \left(\frac{1}{n} \sum_{i=1}^n \left(\lambda_n^4 + \lambda_n^2(\mu_i - k_i)^2/\sigma_i^2\right) \right)^{1/2} \\
      & \lesssim \E\norm{\bb{G}}_{\cl{F}}  + \lambda_n^2 + \lambda_n \nu_2. \\
      & \lesssim (\log n)^{5/2}\left(\sqrt{\cl{V}(\cl{D}')} + \sqrt{\cl{V}(\cl{D})} \right).
    \end{align*}
    We conclude by the definition of the Orlicz norm.
  \end{proof}

  \begin{theorem}
    \label{thm:second_derivative_supremum}

    Under the same assumptions of \cref{thm:fast_rates},
    \[
      \E\sup_\beta|\hat{W}_n''(\beta) - \E \hat{W}_n''(\beta)| \leq C\frac{(\log n)^{7/2}}{\sqrt{n}}
    \]
    for a constant $C$ depending only on $B$ and the VC indices of $\cl{D},\cl{D}',\cl{D}''.$ Moreover, fix any $\delta \in (0,1)$. Then with probability greater than $1-\delta$, we have
    \[
      \sup_\beta|\hat{W}_n''(\beta) - \E \hat{W}_n''(\beta)| \leq C\frac{(\log n)^{7/2}}{\sqrt{n}} \sqrt{\log \frac{2}{\delta}}.
    \]
  \end{theorem}

  \begin{proof}[Proof of \cref{thm:second_derivative_supremum}]
    Recall that
    \[
      \hat{W}_n''(\beta) := \frac{1}{n}\sum_{j=1}^n \Psi_{CC}(Y_j,z_j,\delta(\beta,z_j))\delta'(\beta,z_j)^2 + \Psi_C(Y_j,\beta)\delta''(\beta,z_j).
    \]
    Let
    \begin{align*}
      A(\beta) & := \frac{1}{\sqrt{n}} \sum_{i=1}^n \Psi_{CC}(Y_i,z_i,\delta(z_i;\beta))\delta'(\beta,z_i)^2 \\
      B(\beta) & := \frac{1}{\sqrt{n}} \sum_{i=1}^n \Psi_{C}(Y_i,z_i,\delta(z_i;\beta))\delta''(\beta,z_i).
    \end{align*}
    By the triangle inequality, it suffices to bound
    \[
      n^{-1/2}\left(\sup_\beta |A(\beta) - \E A(\beta)| + \sup_\beta |B(\beta) - \E B(\beta)|\right)
    \]
    in expectation and with high probability. The argument of \cref{thm:derivative_external_regret}, replacing $\delta'$ with $\delta''$, establishes the bound
    \begin{equation}
      \label{eq:second_derivative_supremum_Psi_C_term}
      \E\left[\sup_\beta |B(\beta) - \E B(\beta)|\right] \leq C(\log n)^{5/2}
    \end{equation}
    for a constant $C$ depending on the VC indices of $\cl{D}, \cl{D}''.$
    Therefore, we will focus on bounding $\sup_\beta |A(\beta) - \E A(\beta)|$.\\

    From \cref{lemma:derivatives of Psi}, $\Psi_{CC}$ is given by
    \[
      \Psi_{CC}(Y_i,z_i,C) = \lambda_n^2\frac{(Y_i - k_i)}{\sigma_i^2}\sinc' \left(\lambda_n\left(\frac{Y_i - C}{\sigma_i} \right) \right) - \frac{\lambda_n^3}{\sigma_i} \sinc'' \left(\lambda_n\left(\frac{Y_i - C}{\sigma_i} \right) \right)
    \]
    We emulate the proofs of \cref{thm:main_regret_bound} and \cref{thm:derivative_external_regret}.
    Throughout this proof, define the following functions classes indexed by $\beta$:

    \begin{align}
      & \cl{F}^{(0)} := \set{f^{(0)}_\beta(y,z) =  \lambda_n^2\frac{y - k}{\sigma^2}\sinc' \left(\lambda_n\left(\frac{y - \delta(z;\beta)}{\sigma} \right) \right): \beta \in \bR} \label{eq:second_derivative_F0}\\
      & \cl{F}^{(1)} := \set{f^{(1)}_\beta(y,z) =  \lambda_n^2\frac{y - k}{\sigma^2}\sinc' \left(\lambda_n\left(\frac{y - \delta(z;\beta)}{\sigma} \right) \right)\delta'(\beta,z)^2: \beta \in \bR} \label{eq:second_derivative_F1} \\
      & \cl{F}^{(2)} := \set{f^{(2)}_\beta(y,z) = \frac{\lambda_n^3}{\sigma} \sinc'' \left(\lambda_n\left(\frac{y - \delta(z;\beta)}{\sigma} \right) \right)\delta'(\beta,z)^2: \beta \in \bR} \label{eq:second_derivative_F2}\\
      & \cl{D}^{'2} := \set{\delta'(\beta,z)^2: \beta \in \bR}, \text{ which is also VC-subgraph}\label{eq:second_derivative_clD'2}
    \end{align}

    Let $\norm{\bb{G}}_{\cl{F}^{(1)}} = \sup_\beta |\frac{1}{\sqrt{n}} \sum_{i=1}^n f^{(1)}_\beta(Y_i,z_i) - \E f^{(1)}_\beta(Y_i,z_i)|$ and define the analogous quantity $\norm{\bb{G}}_{\cl{F}^{(2)}}$. Finally, let $\cl{F}^{(1)} - \cl{F}^{(2)}$ denote the difference function class, comprised of the functions $\Psi_{CC}(y,z,\delta(z;\beta))\delta'(\beta,z)^2$. Using a uniform bound for $\sinc'$ and $\sinc''$ alongside \cref{assmp:boundedness}, envelope functions for $\cl{F}^{(1)},\cl{F}^{(2)}$ are given by
    \begin{align}
      \label{eq:assure_second_derivative_F1_second_envelope_function}
      F^{(1)} & := CB^3 \lambda_n^2 \cdot \frac{|y-k|}{\sigma} \\
      \label{eq:assure_second_derivative_F2_second_envelope_function}
      F^{(2)} & := CB^3 \lambda_n^3.
    \end{align}
    for some absolute constant $C$. Furthermore notice that $\cl{D}^{'2}$ is a VC subgraph class, by composing \cref{lemma:vc_indexes}(5) and Lemma 2.6.20.viii of \cite{vaart2023empirical}. To proceed, write
    \[
      \sup_\beta |A(\beta) - \E A(\beta)| \leq \norm{\bb{G}}_{\cl{F}^{(1)}} + \norm{\bb{G}}_{\cl{F}^{(2)}}.
    \]
    \cref{prop:assure_second_derivative_first_term_upper_bound,prop:assure_second derivative_sinc_dv_term_upper_bound}
    show that
    \begin{align*}
      & \E \norm{\bb{G}}_{\cl{F}^{(1)}} \leq C(\log n)^{5/2} \\
      & \E \norm{\bb{G}}_{\cl{F}^{(2)}} \leq C(\log n)^{7/2}
    \end{align*}
    for a constant $C$ that depend on $B$ and the VC indexes of $\cl{D},\cl{D}^{'2}$. Combining this with Eq. \eqref{eq:second_derivative_supremum_Psi_C_term} we obtain the first claim:
    \[
      \E\sup_\beta |\hat{W}_n''(\beta) - \E\hat{W}_n''(\beta)| \leq C\frac{(\log n)^{7/2}}{\sqrt{n}}.
    \]
    \sloppy To obtain the second claim, apply \cref{prop:orlicz_inid_bounds} with $V_i = (Y_i,z_i)$ and envelope function $ F^{(1)} + F^{(2)} = CB^3 \left(\lambda_n^3 + \lambda_n^2 \frac{|y-k|}{\sigma}\right)$. It is not hard to see that $F^{(1)}(Y_i,z_i) + F^{(2)}(Y_i,z_i)$ is sub-Gaussian with parameter at most a constant multiple times $B^3\left(\lambda_n^3 + \lambda_n^2\frac{|\mu_i - k_i|}{\sigma_i} \right)$. Plugging this into \cref{prop:orlicz_inid_bounds} plus the definition of the Orlicz norm yields
    \begin{align*}
      \sup_\beta |\hat{W}_n''(\beta) - \E\hat{W}_n''(\beta)| & \leq \left(\E \sup_\beta |\hat{W}_n''(\beta) - \E\hat{W}_n''(\beta)| + C\frac{\left(\lambda_n^3 + \lambda_n^2\nu_2 \right)}{\sqrt{n}} \right)\sqrt{\log \frac{2}{\delta}} \\
      & \leq C\frac{(\log n)^{7/2}}{\sqrt{n}} \sqrt{\log \frac{2}{\delta}}
    \end{align*}
    with probability at least $1 - \delta$.
  \end{proof}

  \subsubsection{Proofs for uniform bounds for \assure{*} derivatives}

  The structure of the following results are all analogous to \cref{thm:main_regret_bound}.

  \begin{proposition}
    \label{prop:assure_derivative_sinc_term_upper_bound}
    Under the assumptions of \cref{thm:derivative_external_regret}, we have
    \[
      \E\norm{\bb{G}}_{\cl{F}^{(1)}} \leq C(\log n)^{3/2}\left(\sqrt{\cl{V}(\cl{D}')} + \sqrt{\cl{V}(\cl{D})} \right),
    \]
    for a constant $C$ depending on $B$, where $\cl{F}^{(1)}$ is the class of functions defined in \eqref{eq:derivative_F1}.
  \end{proposition}
  \begin{proof}
    Following the proof of \cref{thm:inid_donsker_bound} as before with $V_i = (Y_i,z_i)$, we may apply symmetrization and a maximal inequality to obtain
    \begin{equation}
      \label{eq:derivative_si_term_entropy_integral}
      \E\left[\norm{\bb{G}}_{\cl{F}^{(1)}}\right] \leq \E \left[  \int_0^{\eta_n} \sqrt{1 + \log N(\e ,\cl{F}^{(1)},\mathbb{L}_2(\Prob_n))} d\e \right],
    \end{equation}
    with $\eta_n := \sup_\beta \norm{f^{(1)}_\beta}_n$. We will bound the covering number in two steps. First, consider an $\e$-cover $M_\e$ in the $\bb{L}^2(\Prob_n)$ metric for the function class
    \begin{equation}
      \cl{F}^{(0)} := \set{f^{(0)}_\beta(y,z) := -\lambda_n\left(\frac{y-k}{\sigma} \right)\sinc\left(\lambda_n\left(y - \delta(z;\beta)\right)/\sigma \right):\beta \in \bR}
    \end{equation}
    and a $\e$-cover in the $\bL^4(\Prob_n)$ metric $N_\e$ for the function class $\cl{D}' := \set{\delta'(\beta,z): \beta \in \bR}$. Moreover, take the cover such that all elements of $M_\e$ are bounded above by the envelope $C\lambda_n\left|\frac{y - k}{\sigma}\right|$ for an absolute constant $C$. Take $d(z) \in N_\e$ which is $\e$-close to $\delta'(\beta,z)$ and similarly, take $g(y,z) \in M_\e$ which is $\e$-close to $f^{(0)}_\beta(y,z)$. By assumption, we can take the net $N_\e$ such that $d(z_i) \leq B$ for all $i$.  Then
    \begin{align*}
      & \frac{1}{n}\sum_{i=1}^n \left(f^{(0)}_\beta(Y_i,z_i)\delta'(\beta,z_i)  - g(Y_i,z_i)d(z_i)\right)^2 \\
      \lesssim & \ \frac{1}{n}\sum_{i=1}^n f^{(0)}_\beta(Y_i,z_i)^2\left(\delta'(\beta,z_i)  - d(z_i)\right)^2 +  \frac{1}{n}\sum_{i=1}^n d(z_i)^2\left(f^{(0)}_\beta(Y_i,z_i) - g(Y_i,z_i)\right)^2 \\
      \leq & \ \lambda_n^2\left(\frac{1}{n}\sum_{i=1}^n \frac{(Y_i - k_i)^4}{\sigma_i^4}\right)^{1/2} \e^2 + B^2 \e^2
    \end{align*}
    where we used the envelope bound and Cauchy Schwarz on the last line. Thus, we conclude the following relation between covering numbers:
    \begin{equation}
      \label{eq:derivative_sinc_covering_number_pushforward}
      N(C\left(\lambda_n \hat{\nu}_4 + B \right) \e ,\cl{F}^{(1)},\mathbb{L}_2(\Prob_n)) \leq  N(\e ,\cl{D}',\mathbb{L}_4(\Prob_n)) \cdot N(\e  ,\cl{F}^{(0)},\mathbb{L}_2(\Prob_n)),
    \end{equation}
    where 
    \[
    \hat{\nu}_4 := \norm{\frac{y - k}{\sigma}}_{\bL^4(\Prob_n)}.
    \]
    We will use this notation in other proofs. Let $\hat{\zeta}_4 := \norm{\frac{y-k}{\sigma^2}}_{\bL^4(\Prob_n)}$. By a Lipschitzness argument analogous to the one used to obtain Eq. \eqref{eq:si_covering_number_pushforward}, we also have
    \begin{equation}
      N(C\lambda_n^2 \hat{\zeta}_4 \e ,\cl{F}^{(0)},\mathbb{L}_2(\Prob_n)) \leq N(\e   ,\cl{D},\mathbb{L}_4(\Prob_n)).
    \end{equation}
    Putting these bounds into Eq. \eqref{eq:derivative_si_term_entropy_integral}, we have
    \begin{align*}
      \E \norm{\bb{G}}_{\cl{F}^{(1)}} & \lesssim \E\left[\int_0^{\eta_n} \sqrt{1 + \log N\left(\frac{\e}{C\lambda_n^2\hat{\zeta}_4(\lambda_n \hat{\nu}_4 + B)},\cl{D},\bL^4(\Prob_n) \right) } d\e \right] \\
      & +  \E\left[\int_0^{\eta_n} \sqrt{1 + \log N\left(\frac{\e}{C(\lambda_n \hat{\nu}_4 + B)},\cl{D}',\bL^4(\Prob_n) \right) } d\e \right]
    \end{align*}
    By change of variables and the uniform envelope function $B$ for $\cl{D},\cl{D}'$ we then have
    \[
      \E \norm{\bb{G}}_{\cl{F}^{(1)}} \lesssim \E\norm{F^{(1)}}_{\bL^2(\Prob_n)} + \lambda_n^3\sqrt{\cl{V}(\cl{D})}\E\left[\hat{\zeta}_4(\hat{\nu}_4 + B) \right] + \lambda_n\sqrt{\cl{V}(\cl{D}')}\E[\hat{\nu}_4 + B].
    \]
    Using Cauchy Schwarz repeatedly and \cref{assmp:boundedness}, it is not difficult to observe the bound
    \[
      \E \norm{\bb{G}}_{\cl{F}^{(1)}} \leq C\lambda_n^3\left(\sqrt{\cl{V}(\cl{D}')} + \sqrt{\cl{V}(\cl{D})} \right)
    \]
    for a constant $C$ depending on $B$, as desired.
  \end{proof}

  \begin{proposition}
    \label{prop:assure_derivative_sinc_dv_term_upper_bound}
    Under the assumptions of \cref{thm:derivative_external_regret}, we have
    \[
      \E\norm{\bb{G}}_{\cl{F}^{(2)}} \leq C(\log n)^{5/2} \left(\sqrt{\cl{V}(\cl{D}')} + \sqrt{\cl{V}(\cl{D})} \right).
    \]
    for a constant $C$ depending on $B$, where $\cl{F}^{(2)}$ is the class of functions defined in \eqref{eq:derivative_F2}.
  \end{proposition}
  \begin{proof}

    Following the proof of \cref{thm:inid_donsker_bound} as before with $V_i = (Y_i,z_i)$, we may apply symmetrization and a maximal inequality to obtain
    \begin{equation}
      \label{eq:derivative_sinc_dv_term_entropy_integral}
      \E\left[\norm{\bb{G}}_{\cl{F}^{(2)}}\right] \leq \E \left[  \int_0^{\eta_n} \sqrt{1 + \log N(\e ,\cl{F}^{(2)},\mathbb{L}_2(\Prob_n))} d\e \right],
    \end{equation}
    with $\eta_n := \sup_\beta \norm{f^{(2)}_\beta}_n$. We will bound the covering number in two steps. First, consider an $\e$-cover $M_\e$ in the $\bb{L}^2(\Prob_n)$ metric for the function class
    \begin{equation}
      \cl{S} := \set{s_\beta(y,z) := \lambda_n^2 \sinc'\left(\lambda_n\left(y - \delta(z;\beta)\right)/\sigma \right):\beta \in \bR}
    \end{equation}
    and a $\e$-cover in the $\bL^2(\Prob_n)$ metric $N_\e$ for the function class $\cl{D}' := \set{\delta'(\beta,z): \beta \in \bR}$. Moreover, take the cover such that all elements of $M_\e$ are bounded above by the envelope $C\lambda_n^2$ for an absolute constant $C$. Take $d(z) \in N_\e$ which is $\e$-close to $\delta'(\beta,z)$ and similarly, take $g(y,z) \in M_\e$ which is $\e$-close to $s_\beta(y,z)$. By assumption, we can take the net $N_\e$ such that $d(z_i) \leq B$ for all $i$.  Then
    \begin{align*}
      & \frac{1}{n}\sum_{i=1}^n \left(s_\beta(Y_i,z_i)\delta'(\beta,z_i)  - g(Y_i,z_i)d(z_i)\right)^2 \\
      \lesssim & \ \frac{1}{n}\sum_{i=1}^n s_\beta(Y_i,z_i)^2\left(\delta'(\beta,z_i)  - d(z_i)\right)^2 +  \frac{1}{n}\sum_{i=1}^n d(z_i)^2\left(s_\beta(Y_i,z_i) - g(Y_i,z_i)\right)^2 \\
      \leq &C\lambda_n^4 \e^2 + B^2 \e^2,
    \end{align*}
    for a constant $C$ that only depends on $B$, where we used the envelope bound on the last line.
    Thus,
    \begin{equation}
      N(C\left(\lambda_n^2 + B \right) \e ,\cl{F}^{(2)},\mathbb{L}_2(\Prob_n)) \leq  N(\e ,\cl{D}',\mathbb{L}_2(\Prob_n)) \cdot N(\e  ,\cl{S},\mathbb{L}_2(\Prob_n))
    \end{equation}
    for some absolute constant $C.$ Since $\sinc'$ is Lipschitz and $\sigma_i^{-1} \leq B$, we also have
    \[
      N(B\lambda_n^3\e ,\cl{S},\mathbb{L}_2(\Prob_n)) \leq N(\e ,\cl{D},\mathbb{L}_2(\Prob_n)).
    \]
    Putting these bounds into Eq. \eqref{eq:derivative_sinc_dv_term_entropy_integral}, we have
    \begin{align*}
      \E \norm{\bb{G}}_{\cl{F}^{(2)}} & \lesssim \E\left[\int_0^{\eta_n} \sqrt{1 + \log N\left(\frac{\e}{C\lambda_n^3(\lambda_n^2 + B)},\cl{D},\bL^2(\Prob_n) \right) } d\e \right] \\
      & +  \E\left[\int_0^{\eta_n} \sqrt{1 + \log N\left(\frac{\e}{C(\lambda_n^2 + B)},\cl{D}',\bL^2(\Prob_n) \right) } d\e \right]
    \end{align*}
    By change of variables on the entropy integral, we have
    \begin{equation}
      \E \norm{\bb{G}}_{\cl{F}^{(2)}} \lesssim \E \norm{F^{(2)}}_n + \lambda_n^3(\lambda_n^2 + B)\sqrt{V(\cl{D})} + (\lambda_n^2 + B)\sqrt{V(\cl{D}')},
    \end{equation}
    Using the definition of the envelope $F^{(2)}$ and the fact that $\lambda_n \geq 1$ for $n$ large enough,
    \[
      \E \norm{\bb{G}}_{\cl{F}^{(2)}} \leq C\lambda_n^5 \left(\sqrt{\cl{V}(\cl{D}')} + \sqrt{\cl{V}(\cl{D})} \right),
    \]
    which is the stated bound since $\lambda_n^5 \asymp (\log n)^{5/2}$.
  \end{proof}

  \begin{proposition}
    \label{prop:assure_second_derivative_first_term_upper_bound}
    Under the setting and assumptions of \cref{thm:second_derivative_supremum}, we have
    \[
      \E \norm{\bb{G}}_{\cl{F}^{(1)}} \leq C(\log n)^{5/2}\left(\sqrt{\cl{V}(\cl{D}^{'2})} + \sqrt{\cl{V}(\cl{D})} \right),
    \]
    for a constant $C$ that depends on $B$, where $\cl{F}^{(1)}$ is defined in \eqref{eq:second_derivative_F1}.
  \end{proposition}
  \begin{proof}
    Throughout this proof, $C$ will refer to an absolute constant which may depend on $B$ and will change from line to line. Following the proof of \cref{thm:inid_donsker_bound} as before with $V_i = (Y_i,z_i)$, we may apply symmetrization and a maximal inequality to obtain
    \begin{equation}
      \label{eq:second_derivative_si_term_entropy_integral}
      \E\left[\norm{\bb{G}}_{\cl{F}^{(1)}}\right] \leq \E \left[  \int_0^{\eta_n} \sqrt{1 + \log N(\e ,\cl{F}^{(1)},\mathbb{L}_2(\Prob_n))} d\e \right],
    \end{equation}
    with $\eta_n := \sup_\beta \norm{f^{(1)}_\beta}_n$. We will bound the covering number in two steps. First, consider an $\e$-cover $M_\e$ in the $\bb{L}^2(\Prob_n)$ metric for the function class
    $\cl{F}^{(0)}$ in \eqref{eq:second_derivative_F0} and a $\e$-cover in the $\bL^4(\Prob_n)$ metric $N_\e$ for the function class $\cl{D}^{'2}$. Moreover, take the cover such that all elements of $M_\e$ are bounded above by the envelope $C\lambda_n^2\left|\frac{y - k}{\sigma^2}\right|$ for an absolute constant $C$. Take $d(z) \in N_\e$ which is $\e$-close to $\delta'(\beta,z)^2$ and similarly, take $g(y,z) \in M_\e$ which is $\e$-close to $f^{(0)}_\beta(y,z)$. By assumption, we can take the net $N_\e$ such that $d(z_i) \leq B^2$ for all $i$.  Then
    \begin{align*}
      & \frac{1}{n}\sum_{i=1}^n \left(f^{(0)}_\beta(Y_i,z_i)\delta'(\beta,z_i)^2  - g(Y_i,z_i)d(z_i)\right)^2 \\
      \lesssim & \ \frac{1}{n}\sum_{i=1}^n f^{(0)}_\beta(Y_i,z_i)^2\left(\delta'(\beta,z_i)^2  - d(z_i)\right)^2 +  \frac{1}{n}\sum_{i=1}^n d(z_i)^2\left(f^{(0)}_\beta(Y_i,z_i) - g(Y_i,z_i)\right)^2 \\
      \leq & \ \lambda_n^4\left(\frac{1}{n}\sum_{i=1}^n \frac{(Y_i - k_i)^4}{\sigma_i^8}\right)^{1/2} \e^2 + B^4 \e^2 \\
      \leq & \ B^2\lambda_n^4\left(\frac{1}{n}\sum_{i=1}^n \frac{(Y_i - k_i)^4}{\sigma_i^4}\right)^{1/2} \e^2 + B^4 \e^2
    \end{align*}
    where we used the envelope bound and Cauchy--Schwarz in the second to last line. Thus, we conclude the following relation between covering numbers:
    \begin{equation}
      \label{eq:derivative_sinc_covering_number_pushforward}
      N(CB\left(\lambda_n^2 \hat{\nu}_4 + B \right) \e ,\cl{F}^{(1)},\mathbb{L}_2(\Prob_n)) \leq  N(\e ,\cl{D}^{'2},\mathbb{L}_4(\Prob_n)) \cdot N(\e  ,\cl{F}^{(0)},\mathbb{L}_2(\Prob_n)).
    \end{equation} By an argument analogous to the one used to obtain Eq. \eqref{eq:si_covering_number_pushforward}, we also have
    \begin{equation}
      N(CB^2\lambda_n^3 \hat{\nu}_4 \e  ,\cl{F}^{(0)},\mathbb{L}_2(\Prob_n)) \leq N(\e   ,\cl{D},\mathbb{L}_4(\Prob_n)).
    \end{equation}

    To show this, take any $\e$-cover in the $\bb{L}^4(\Prob_n)$ metric of the decision function space $\cl{D}$. Take any $\delta(\cdot;\beta) \in \cl{D}$, and suppose the function $\tilde{d}$ is $\e$-close to this function. Then because $\sinc'$ is Lipschitz for some constant $C$,
    \begin{align*}
      &  \norm{f_\beta^{(0)}(y,z) - \lambda_n^2\frac{(y - k)}{\sigma_i^2}\sinc'(\lambda_n(y - \tilde{d})/\sigma)}_{\bb{L}^2(\Prob_n)} \\
      \leq \ & C\lambda_n^3\left(\frac{1}{n} \sum_{i=1}^n \frac{(Y_i - k_i)^2}
      {\sigma_i^6}  \left(\delta(z_i;\beta) - \tilde{d} \right)^2 \right)^{1/2} \\
      \leq \ & CB\lambda_n^3\left( \frac{1}{n} \sum_{i=1}^n \frac{(Y_i - k_i)^2}
      {\sigma_i^4}  \left(\delta(z_i;\beta) - \tilde{d} \right)^2 \right)^{1/2} \\
      \leq \ & CB^2\lambda_n^3 \left( \frac{1}{n} \sum_{i=1}^n \frac{(Y_i - k_i)^4}
        {\sigma_i^4}\right)^{1/4} \left(\frac{1}{n} \sum_{i=1}^n \left(\delta(z_i;\beta) -
      \tilde{d} \right)^4\right)^{1/4} 
      \tag{Cauchy--Schwarz, \cref{assmp:boundedness}} \\
      \leq \ & CB^2 \lambda_n^3 \hat{\nu}_4 \e.
    \end{align*}

    Putting these bounds into Eq. \eqref{eq:second_derivative_si_term_entropy_integral}, we have
    \begin{align*}
      \E \norm{\bb{G}}_{\cl{F}^{(1)}} & \lesssim \E\left[\int_0^{\eta_n} \sqrt{1 + \log N\left(\frac{\e}{CB^3\lambda_n^3\hat{\nu}_4(\lambda_n^2 \hat{\nu}_4 + B)},\cl{D},\bL^4(\Prob_n) \right) } d\e \right] \\
      & +  \E\left[\int_0^{\eta_n} \sqrt{1 + \log N\left(\frac{\e}{CB(\lambda_n^2 \hat{\nu}_4 + B)},\cl{D}^{'2},\bL^4(\Prob_n) \right) } d\e \right].
    \end{align*}
    By a change of variables argument and using the uniform envelope functions $B,B^2$ for $\cl{D},\cl{D}^{'2}$ we then have
    \[
      \E \norm{\bb{G}}_{\cl{F}^{(1)}} \lesssim \E\norm{F^{(1)}}_{\bL^2(\Prob_n)} + \lambda_n^5\sqrt{\cl{V}(\cl{D})}\E\left[\hat{\nu}_4(\hat{\nu}_4 + B) \right] + \lambda_n^2\sqrt{\cl{V}(\cl{D}^{'2})}\E[\hat{\nu}_4 + B],
    \]
    where $F^{(1)}$ is the envelope function defined in \eqref{eq:assure_second_derivative_F1_second_envelope_function}. Using Cauchy Schwarz repeatedly and \cref{assmp:boundedness}, it is not difficult to observe the bound
    \[
      \E \norm{\bb{G}}_{\cl{F}^{(1)}} \leq C\lambda_n^5\left(\sqrt{\cl{V}(\cl{D}^{'2})} + \sqrt{\cl{V}(\cl{D})} \right)
    \]
    for a constant $C$ depending on $B$, as desired.
  \end{proof}

  \begin{proposition}
    \label{prop:assure_second derivative_sinc_dv_term_upper_bound}
    Under the setting and assumptions of \cref{thm:second_derivative_supremum}, we have
    \[
      \E\norm{\bb{G}}_{\cl{F}^{(2)}} \leq C(\log n)^{7/2} \left(\sqrt{\cl{V}(\cl{D}^{'2})} + \sqrt{\cl{V}(\cl{D})} \right).
    \]
    for a constant $C$ depending on the bound $B$, where $\cl{F}^{(2)}$ is defined in \eqref{eq:second_derivative_F2}.
  \end{proposition}
  \begin{proof}Throughout this proof, $C$ will refer to an absolute constant which may depend on $B$ and will change from line to line.
    As before,
    \begin{equation}
      \label{eq:second_derivative_sinc_dv_term_entropy_integral}
      \E\left[\norm{\bb{G}}_{\cl{F}^{(2)}}\right] \leq \E \left[  \int_0^{\eta_n} \sqrt{1 + \log N(\e ,\cl{F}^{(2)},\mathbb{L}_2(\Prob_n))} d\e \right],
    \end{equation}
    with $\eta_n := \sup_\beta \norm{f^{(2)}_\beta}_n$. We will bound the covering number in two steps. First, consider an $\e$-cover $M_\e$ in the $\bb{L}^2(\Prob_n)$ metric for the function class
    \begin{equation}
      \cl{S} := \set{s_\beta(y,z) := \frac{\lambda_n^3}{\sigma} \sinc''\left(\lambda_n\left(y - \delta(z;\beta)\right)/\sigma \right):\beta \in \bR}
    \end{equation}
    and a $\e$-cover in the $\bL^2(\Prob_n)$ metric $N_\e$ for the function class $\cl{D}^{'2}$. Moreover, take the cover such that all elements of $M_\e$ are bounded above by the envelope $CB\lambda_n^3$ for an absolute constant $C$. Take $d(z) \in N_\e$ which is $\e$-close to $\delta'(\beta,z)^2$ and similarly, take $g(y,z) \in M_\e$ which is $\e$-close to $s_\beta(y,z)$. By \cref{assmp:boundedness}, we can take the net $N_\e$ such that $d(z_i) \leq B^2$ for all $i$.  Then
    \begin{align*}
      & \frac{1}{n}\sum_{i=1}^n \left(s_\beta(Y_i,z_i)\delta'(\beta,z_i)^2  - g(Y_i,z_i)d(z_i)\right)^2 \\
      \lesssim & \ \frac{1}{n}\sum_{i=1}^n s_\beta(Y_i,z_i)^2\left(\delta'(\beta,z_i)^2  - d(z_i)\right)^2 +  \frac{1}{n}\sum_{i=1}^n d(z_i)^2\left(s_\beta(Y_i,z_i) - g(Y_i,z_i)\right)^2 \\
      \leq & \ CB^2\lambda_n^6 \e^2 + B^4 \e^2,
    \end{align*}
    where we used the envelope bound on the last line. Thus,
    \begin{equation}
      N(CB\left(\lambda_n^3 + B^2 \right) \e ,\cl{F}^{(2)},\mathbb{L}_2(\Prob_n)) \leq  N(\e ,\cl{D}^{'2},\mathbb{L}_2(\Prob_n)) \cdot N(\e  ,\cl{S},\mathbb{L}_2(\Prob_n))
    \end{equation}
    for some absolute constant $C.$ Since $\sinc''$ is Lipschitz and $\sigma_i^{-1} \leq B$, we also have
    \[
      N(CB^2\lambda_n^4\e ,\cl{S},\mathbb{L}_2(\Prob_n)) \leq N(\e ,\cl{D},\mathbb{L}_2(\Prob_n)).
    \]
    Putting these bounds into Eq. \eqref{eq:second_derivative_sinc_dv_term_entropy_integral}, we have
    \begin{align*}
      \E \norm{\bb{G}}_{\cl{F}^{(2)}} & \lesssim \E\left[\int_0^{\eta_n} \sqrt{1 + \log N\left(\frac{\e}{CB^4\lambda_n^4(\lambda_n^3 + B^2)},\cl{D},\bL^2(\Prob_n) \right) } d\e \right] \\
      & +  \E\left[\int_0^{\eta_n} \sqrt{1 + \log N\left(\frac{\e}{CB^2(\lambda_n^3 + B^2)},\cl{D}^{'2},\bL^2(\Prob_n) \right) } d\e \right]
    \end{align*}
    Change of variables for the entropy integral gives as before 
    \begin{equation}
      \E \norm{\bb{G}}_{\cl{F}^{(2)}} \lesssim \E \norm{F^{(2)}}_n + \lambda_n^4(\lambda_n^3 + B^2)\sqrt{V(\cl{D})} + (\lambda_n^3 + B^2)\sqrt{V(\cl{D}^{'2})}.
    \end{equation}
    Using the definition of the envelope $F^{(2)}$, we conclude the proof.
  \end{proof}

  \subsubsection{Example decision rule class for \cref{thm:fast_rates}}
  \label{sec:example_fast_rate_class}

  \begin{example}
    \label{ex:fast_rate_class}

    Consider the class of simple threshold decision rules $\delta(z;\beta) = \beta$ with $\beta \in \bR$. It is unclear whether there are interesting examples which satisfy \crefrange{assmp:boundedness}{assmp:well_sep_max}. The following assumption gives such a class of $\mu_{1:n}$. Fix $\kappa_0 > 0, B > 0$ and let $\overline{\mu} := \frac{1}{n}\sum_{i=1}^n\mu_i, \overline{\sigma}^2 := \frac{1}{n}\sum_{i=1}^n(\mu_i - \overline{\mu})^2.$
    \begin{as}
      \label{as:lower_level_fast_rate_assn}
      \begin{asenum}
      \item \label{assmp:lower_level_mu_boundedness} $-B < \min_i \mu_i < 0$ and $0 < \max_i \mu_i < B$.
      \item \label{assmp:lower_level_bounded_decision}
        Letting $m_B := \inf_{[B,3B]} t\varphi(t)$,
        \[
          \left|\overline{\mu}\right| \leq \overline{\sigma}^2\frac{m_B}{\varphi(B)}.
        \]
      \item \label{assmp:lower_level_curvature} $\frac{1}{n} \sum_{i=1}^n \mu_i^2 > \kappa_0/\varphi(3B)$.
      \end{asenum}
    \end{as}

    \begin{proposition}
      \label{prop:assumptionBsatisfiesA}
      Fix constants $B,\kappa_0$ as above. Assume $\sigma_i^2 \equiv 1$ and \crefrange{assmp:lower_level_mu_boundedness}{assmp:lower_level_curvature}. Then \crefrange{assmp:boundedness}{assmp:well_sep_max} hold with $\kappa = \kappa_0$ and with some $\xi > 0$ depending only on $B,\kappa_0.$
    \end{proposition}
    \begin{proof}
      By \cref{assmp:lower_level_mu_boundedness}, it is easy to see that \cref{assmp:boundedness} holds. Now without loss, assume that $\mu_1 \leq \mu_2 \leq \dots \leq \mu_n$.
      In this setting, $W(\beta) = \frac{1}{n} \sum_{i=1}^n \mu_i \Phi(\mu_i - \beta), W'(\beta) =  -\frac{1}{n} \sum_{i=1}^n \mu_i \varphi(\mu_i - \beta)$, and $W''(\beta) =\frac{1}{n} \sum_{i=1}^n \mu_i(\beta - \mu_i) \varphi(\mu_i - \beta).$ A result of \cite{schonberg1948variation} shows that $\varphi$ is a totally positive kernel and, for any sequence of $c_1 \leq c_2 \leq ... \leq c_n$, the number of sign changes in the vector $(W'(c_1), \dots, W'(c_n))$ is at most that in the vector $(\mu_1, \dots, \mu_n)$. Note that since $(\mu_1, ..., \mu_n)$ changes sign only once, $(W'(c_1), \dots, W'(c_n))$ changes sign at most once. This implies that there cannot be multiple local maximizers.

      Next, we claim that \cref{assmp:lower_level_bounded_decision} implies that $W'(-2B) > 0$ and $W'(2B) < 0$. If so, this would show \cref{assmp:unimodal} holds and $\beta^* \in [-2B,2B]$. First, let $f_+(x) = \varphi(x- 2B)$. Let $\overline{f}_+ := \frac{1}{n}\sum_{i=1}^n f_+(\mu_i)$. Observe that
      \begin{equation}
        -W'(2B) = \frac{1}{n} \sum_{i=1}^n \mu_i f_+(\mu_i) = \overline{\mu}\overline{f}_+ + \frac{1}{n}\sum_{i=1}^n (\mu_i - \overline{\mu})\left(f_+(\mu_i) - f_+(\overline{\mu}) \right)
      \end{equation}
      By the mean-value theorem, we have for some $\xi_i$ between $\mu_i,\overline{\mu}$ that $f_+'(\xi_i)(\mu_i - \overline{\mu}) = f_+(\mu_i) - f_+(\overline{\mu})$. Plugging this into the equation above,
      \[
        \frac{1}{n} \sum_{i=1}^n \mu_i f_+(\mu_i) = \overline{\mu}\overline{f}_+ + \frac{1}{n}\sum_{i=1}^n f'_+(\xi_i)(\mu_i - \overline{\mu})^2.
      \]
      We have $f_+'(x) = (2B - x)\varphi(x-2B)$ and $f_+'(\xi_i) \geq m_B = \inf_{t \in [B,3B]} t \varphi(t) > 0.$ Moreover, $\overline{f}_+ \leq \varphi(B)$. Thus,
      \[
        -W'(2B) \geq -|\overline{\mu}|\varphi(B) + \bar \sigma^2  m_B > 0.
      \]
      We can repeat the argument in the other direction with $f_- := \varphi(x + 2B)$. Observe the analogous decomposition
      \begin{equation}
        \frac{1}{n} \sum_{i=1}^n \mu_i f_-(\mu_i) = \overline{\mu}\overline{f}_- + \frac{1}{n}\sum_{i=1}^n (\mu_i - \overline{\mu})\left(f_-(\mu_i) - f_-(\overline{\mu}) \right)
      \end{equation}
      By the mean-value theorem again, we have for some $\xi_i$ between $\mu_i,\overline{\mu}$ that $f_-'(\xi_i)(\mu_i - \overline{\mu}) = f_-(\mu_i) - f_-(\overline{\mu})$. Plugging this into the equation above,
      \[
        \frac{1}{n} \sum_{i=1}^n \mu_i f_-(\mu_i) = \overline{\mu}\overline{f}_- + \frac{1}{n}\sum_{i=1}^n f'_-(\xi_i)(\mu_i - \overline{\mu})^2.
      \]
      We have $f_-'(x) = -(x + 2B)\varphi(x+2B) \leq -m_B$. Moreover, $\overline{f}_- \leq \varphi(B)$. Thus,
      \[
        -W'(-2B) \leq |\overline{\mu}|\varphi(B) - \sigma^2  m_B < 0.
      \]
      This completes the claim.

      Next, notice that
      \[
        W''(\beta^*) = \frac{1}{n}\sum_i \mu_i (\beta^* - \mu_i) \varphi(\mu_i - \beta^*) = - \frac{1}{n}\sum_i \mu_i^2 \varphi(\mu_i - \beta^*).
      \]
      where the second equality follows from the first order condition. Then boundedness, $\beta^* \in [-2B,2B]$, and \cref{assmp:lower_level_curvature} imply
      \[
        W''(\beta^*) < -\frac{\kappa_0}{\varphi(3B)} \varphi(3B) = -\kappa_0.
      \]

      Now, for $\beta$ such that $|\beta - \beta^*| \leq R$, we have
      \begin{align*}
        W''(\beta) & = \beta\left(\frac{1}{n}\sum_{i=1}^n \mu_i \varphi(\mu_i - \beta) \right) - \frac{1}{n}\sum_{i=1}^n \mu_i^2\varphi(\mu_i - \beta) \\
        & \leq \left(2B + R\right) R\sup_{|\beta' - \beta^*|\leq R} \left|W''(\beta') \right|  -\kappa_0 \frac{\varphi\left(3B + R \right)}{\varphi(3B)}.
      \end{align*}
      Notice that for $\mu_i \in [-B,B]$,
      \[
        \left|\frac{1}{n}\sum_{i=1}^n \mu_i(\beta - \mu_i)\varphi(\mu_i - \beta) \right| \leq B\left| \sup_x x\varphi(x) \right| \leq B/4.
      \]
      Thus
      \[
        W''(\beta) \leq \left(2B + R\right) \frac{BR}{4}  -\kappa_0 \frac{\varphi\left(3B + R \right)}{\varphi(3B)}
      \]
      Now by the intermediate value theorem, for some $R$ small enough, less than $\kappa_0/B^2$, $W''(\beta) < -\kappa_0/2$ for all $|\beta - \beta^*| \leq R.$ This choice of $R$ only depends on $B$ and $\kappa_0$. By the uniqueness of local maximizer and a Taylor expansion argument, we conclude that there exists $\xi > 0$ such that
      \[
        \sup_{|\beta - \beta^*| > \kappa_0/B} W(\beta) < \sup_{|\beta - \beta^*| > R} W(\beta) < W(\beta^*) - \xi.
      \]
      This proves \cref{assmp:well_sep_max}.
    \end{proof}
  \end{example}

  \subsection{Auxiliary results for \cref{prop:gaussian_matching_lb_fast_rate}}

  \begin{lemma}
    \label{lemma:bayres_regret_as_classification}
    In the proof of \cref{prop:gaussian_matching_lb_fast_rate}, we have that Bayes regret
    is equal to
    \eqref{eq:bayes_regret_as_classification}.
  \end{lemma}

  \begin{proof}
    To see this, consider the case where $m_G(Y_i) > 0$ so that $a_G^*(Y_i) = 1$. Then
    \begin{align*}
      m_G(Y_i) \left( a_G^*(Y_i) - a_i(Y_{1:n})\right) & = |m_G(Y_i)| \left( 1 - a_i(Y_{1:n})\right) \\
      & = \left|m_G(Y_i)\right| \mathbf{1}\set{a_G^*(Y_i) \neq a_i(Y_{1:n})}.
    \end{align*}
    In the other case where $m_G(Y_i) < 0$ so that $a_G^*(Y_i) = 0$, then
    \begin{align*}
      m_G(Y_i) \left( a_G^*(Y_i) - a_i(Y_{1:n})\right) & = -m_G(Y_i) a_i(Y_{1:n}) \\
      &= \left|m_G(Y_i)\right|a_i(Y_{1:n}) \\
      & = \left|m_G(Y_i)\right| \mathbf{1}\set{a_G^*(Y_i) \neq a_i(Y_{1:n})}. \qedhere
    \end{align*}
  \end{proof}

  \begin{lemma}
    \label{lemma:rad_pair_overlap} In the proof of \cref{prop:gaussian_matching_lb_fast_rate},
    let
    \[
      G_0 := \mathsf{Rad}_{1/2 - 1/(4\sqrt{n})},
      \qquad
      G_1 := \mathsf{Rad}_{1/2 + 1/(4\sqrt{n})}.
    \]
    Let $p_{G_i}$ denote the marginal density of a single observation under $G_i$, and let
    $p_{G_i}^{(-j)}$ denote the joint density of $(Y_1,\dots,Y_{j-1},Y_{j+1},\dots,Y_n)$
    under $G_i^{\otimes n}$. Then, for all sufficiently large $n$,
    \[
      \alpha_{n-1} := \int_{\bR^{n-1}} \min\left(
      p_{G_0}^{(-j)}(\mathbf{y}_{-j}), p_{G_1}^{(-j)}(\mathbf{y}_{-j})\right)
      d\mathbf{y}_{-j} \geq \alpha,
    \]
    where $\alpha := 1 - \sqrt{2}/2 > 0$.
  \end{lemma}

  \begin{proof}
    Since $p_{G_i}^{(-j)}$ is the $(n-1)$-fold product of $p_{G_i}$, we have
    \[
      \alpha_{n-1} = 1 - \mathrm{TV}\left(p_{G_0}^{(-j)}, p_{G_1}^{(-j)}\right).
    \]
    By Pinsker's inequality and tensorization of the Kullback-Leibler divergence,
    \[
      \mathrm{TV}\left(p_{G_0}^{(-j)}, p_{G_1}^{(-j)}\right)
      \leq \sqrt{\frac{n-1}{2}\mathrm{KL}\left(p_{G_0}, p_{G_1}\right)}.
    \]
    By the data processing inequality,
    \[
      \mathrm{KL}\left(p_{G_0}, p_{G_1}\right)
      \leq \mathrm{KL}\left(\Bern\left(\frac{1}{2} - \frac{1}{4\sqrt{n}}\right),
      \Bern\left(\frac{1}{2} + \frac{1}{4\sqrt{n}}\right)\right).
    \]
    The Bernoulli KL-divergence equals
    \[
      \frac{1}{2\sqrt{n}} \log\frac{\frac{1}{2} + \frac{1}{4\sqrt{n}}}{
      \frac{1}{2} - \frac{1}{4\sqrt{n}}},
    \]
    which is at most $1/n$ for all sufficiently large $n$. Therefore
    \[
      \alpha_{n-1} \geq 1 - \sqrt{\frac{n-1}{2n}} \geq 1 - \frac{\sqrt{2}}{2} = \alpha.
    \]
  \end{proof}

  \begin{lemma}
    \label{lemma:rad_pair_large_bayes_regret}
    In the proof of \cref{prop:gaussian_matching_lb_fast_rate},
    let
    \[
      G_0 := \mathsf{Rad}_{1/2 - 1/(4\sqrt{n})},
      \qquad
      G_1 := \mathsf{Rad}_{1/2 + 1/(4\sqrt{n})}.
    \]
    Then there exists an absolute constant $c > 0$ such that, for every monotone decision
    rule $a_i(Y_{1:n}) = \one\set{Y_i > \delta_i(Y_{-i})}$,
    \[
      R(G_0,a_{1:n}) + R(G_1,a_{1:n}) \geq \frac{c}{n}
    \]
    for all sufficiently large $n$.
  \end{lemma}

  \begin{proof}
    Let $c_i^*$ denote the zero of the posterior mean $m_{G_i}(y)$, and define
    \[
      w(y) := \min\left(\left|m_{G_0}(y)\right|,\left|m_{G_1}(y)\right|\right),
      \qquad
      q(\mathbf{y}) := \min\left(p_{G_0}^{\otimes n}(\mathbf{y}),
      p_{G_1}^{\otimes n}(\mathbf{y})\right).
    \]
    Since $c_1^* < c_0^*$ and $a_i(Y_{1:n})$ is monotone in $Y_i$, for each fixed
    $Y_{-i}$ we have
    \[
      \one\set{a_{G_0}^*(Y_i) \neq a_i(Y_{1:n})}
      + \one\set{a_{G_1}^*(Y_i) \neq a_i(Y_{1:n})}
      \geq \one\set{c_1^* < Y_i < c_0^*}.
    \]
    Hence, by \eqref{eq:bayes_regret_as_classification},
    \begin{align*}
      R(G_0,a_{1:n}) + R(G_1,a_{1:n})
      &\geq \frac{1}{n}\sum_{i=1}^n \int_{\bR^n} w(y_i)\one\set{c_1^* < y_i < c_0^*}
      q(\mathbf{y})\, d\mathbf{y} \\
      &\geq \alpha_{n-1} \int_{c_1^*}^{c_0^*} w(y) \min\left(p_{G_0}(y),p_{G_1}(y)\right)\, dy,
    \end{align*}
    where $\alpha_{n-1}$ is defined in \cref{lemma:rad_pair_overlap}. By
    \cref{lemma:rad_pair_overlap}, $\alpha_{n-1} \geq \alpha := 1 - \sqrt{2}/2$ for all
    sufficiently large $n$.

    Next, write $p_- := \frac{1}{2} - \frac{1}{4\sqrt{n}}$ and
    $p_+ := \frac{1}{2} + \frac{1}{4\sqrt{n}}$. The posterior mean under
    $\mathsf{Rad}_p$ is
    \[
      m_{\mathsf{Rad}_p}(y)
      = \frac{pe^{2y} - (1-p)}{pe^{2y} + (1-p)}
      = \tanh\left(y - c_p^*\right),
      \qquad
      c_p^* := \frac{1}{2}\log\frac{1-p}{p}.
    \]
    Therefore, with $c_n := c_{p_-}^* = -c_{p_+}^*$, we have
    \[
      c_0^* = c_n,
      \qquad
      c_1^* = -c_n,
      \qquad
      c_n = \frac{1}{2}\log\frac{\frac{1}{2} + \frac{1}{4\sqrt{n}}}{
      \frac{1}{2} - \frac{1}{4\sqrt{n}}}.
    \]
    Since $c_n = \operatorname{arctanh}\left(\frac{1}{2\sqrt{n}}\right)$, we have
    $c_n \geq \frac{1}{2\sqrt{n}}$.

    For $y \in [-c_n,c_n]$,
    \[
      \left|m_{G_0}(y)\right| = \tanh(c_n - y),
      \qquad
      \left|m_{G_1}(y)\right| = \tanh(y + c_n).
    \]
    For all sufficiently large $n$ we have $2c_n \leq 1$, and $\tanh(t) \geq t/2$ for
    $t \in [0,1]$. Hence, on $[-c_n,c_n]$,
    \[
      w(y) \geq \frac{1}{2}\min(c_n - y, y + c_n) = \frac{1}{2}(c_n - |y|).
    \]
    Moreover, for all sufficiently large $n$ the interval $[-c_n,c_n]$ is contained in
    $[-1,1]$. Since both $p_-$ and $p_+$ lie in $[1/4,3/4]$,
    \[
      \min\left(p_{G_0}(y),p_{G_1}(y)\right)
      \geq \frac{1}{4}\varphi(2)
      \qquad \text{for all } y \in [-c_n,c_n].
    \]
    Combining the preceding displays, we obtain
    \begin{align*}
      R(G_0,a_{1:n}) + R(G_1,a_{1:n})
      &\geq \alpha \frac{\varphi(2)}{4} \int_{-c_n}^{c_n} w(y)\, dy \\
      &\geq \alpha \frac{\varphi(2)}{8} \int_{-c_n}^{c_n} (c_n - |y|)\, dy \\
      &= \alpha \frac{\varphi(2)}{8} c_n^2 \\
      &\geq \alpha \frac{\varphi(2)}{32} \frac{1}{n}.
    \end{align*}
    This proves the claim.
  \end{proof}

  \subsection{Auxiliary results for \cref{thm:assure_nongaussian_microdata}}
  \label{sec:proofs_for_sec_extensions}

  Throughout this subsection, we fix $\delta_1,\ldots,\delta_n$, set
  $h = 1/\sqrt{2\log n}$, and let $Z_i \sim \Norm(\mu_i,\sigma_i^2)$ independently. We
  record the three ingredients used in the proof of
  \cref{thm:assure_nongaussian_microdata}.

  \begin{lemma}
    \label{lemma:diff_in_variance_is_lipschitz}
    For
    \[
      F_x(s;\delta) := \frac{x}{2} + \frac{x}{\pi}\Si\left(\frac{x-\delta}{s h}\right) -
      \frac{s}{h}\sinc\left(\frac{x-\delta}{s h}\right),
    \]
    the map $s \mapsto F_x(s;\delta)$ is $C(1/h + |x|/\e)$-Lipschitz on $[\e,\infty)$,
    uniformly in $\delta$.
  \end{lemma}
  \begin{proof}
    Write $y := (x-\delta)/h$. Then
    \[
      F_x(s;\delta) = \frac{x}{2} + \frac{x}{\pi}\Si\left(\frac{y}{s}\right)
      - \frac{s^2}{h}\frac{\sin(y/s)}{\pi y}.
    \]
    Differentiating in $s$ gives
    \begin{align*}
      \frac{d}{ds}F_x(s;\delta)
      &= -\frac{x}{\pi s}\sin\left(\frac{x-\delta}{s h}\right)
      - \frac{2}{h}\sinc\left(\frac{y}{s}\right)
      + \frac{1}{\pi h}\cos\left(\frac{y}{s}\right).
    \end{align*}
    Using the uniform boundedness of $\sin$, $\cos$, and $\sinc$, we obtain
    \[
      \left|\frac{d}{ds}F_x(s;\delta)\right| \leq C\left(\frac{1}{h} + \frac{|x|}{s}\right)
      \leq C\left(\frac{1}{h} + \frac{|x|}{\e}\right)
    \]
    on $[\e,\infty)$. The mean value theorem completes the proof.
  \end{proof}

  \begin{lemma}
    \label{lemma:nongaussian_microdata_term1}
    Under the assumptions of \cref{thm:assure_nongaussian_microdata},
    \[
      \frac{1}{n}\sum_{i=1}^n \left|\E\left[\hat{w}_{h,i}(\sqrt{m}\bar{Y}_i;\delta_i) -
      w_{h,i}(\sqrt{m}\bar{Y}_i;\delta_i)\right]\right|
      \leq C\frac{\sqrt{\log n}}{\sqrt{m}}.
    \]
  \end{lemma}
  \begin{proof}
    Fix $i$ and write $\delta = \delta_i$. Since
    \[
      \hat{w}_{h,i}(x;\delta) = F_x(\hat{\sigma}_i;\delta),
      \qquad
      w_{h,i}(x;\delta) = F_x(\sigma_i;\delta),
    \]
    \cref{lemma:diff_in_variance_is_lipschitz} controls the sensitivity to variance
    estimation. Let $A_i := \set{\hat{\sigma}_i > \sigma_i/\sqrt{2}}$. By Chebyshev's
    inequality,
    \[
      \Prob(A_i^c)
      \leq \Prob\left(\left|\hat{\sigma}_i^2 - \sigma_i^2\right| > \frac{\sigma_i^2}{2}\right)
      \leq \frac{4\Var(\hat{\sigma}_i^2)}{\sigma_i^4}.
    \]
    On $A_i$, Lipschitzness gives
    \[
      \left|\hat{w}_{h,i}(\sqrt{m}\bar{Y}_i;\delta) - w_{h,i}(\sqrt{m}\bar{Y}_i;\delta)\right|
      \leq C\left(\frac{1}{h} + \frac{|\sqrt{m}\bar{Y}_i|}{\sigma_i}\right)
      \left|\hat{\sigma}_i - \sigma_i\right|.
    \]
    Also, using the boundedness of $\Si$ and $\sinc$,
    \[
      \left|\hat{w}_{h,i}(x;\delta) - w_{h,i}(x;\delta)\right|\one(A_i^c)
      \leq C\left(\frac{1}{h} + |x|\right)\one(A_i^c).
    \]
    Therefore,
    \begin{align*}
      & \E\left[\left|\hat{w}_{h,i}(\sqrt{m}\bar{Y}_i;\delta) -
      w_{h,i}(\sqrt{m}\bar{Y}_i;\delta)\right|\right] \\
      & \leq C \E\left[\left(\frac{1}{h} + \frac{|\sqrt{m}\bar{Y}_i|}{\sigma_i}\right)
      \left|\hat{\sigma}_i - \sigma_i\right|\right]
      + C \E\left[\left(\frac{1}{h} + |\sqrt{m}\bar{Y}_i|\right)\one(A_i^c)\right] \\
      & \leq C\left(
        \frac{1}{h}\sqrt{\E(\hat{\sigma}_i-\sigma_i)^2}
        + \frac{\sqrt{m}}{\sigma_i}\sqrt{\E(\bar{Y}_i^2)\E(\hat{\sigma}_i-\sigma_i)^2}
      \right) \\
      & \quad + C\left(
        \frac{1}{h}\Prob(A_i^c)
        + \sqrt{m}\sqrt{\E(\bar{Y}_i^2)\Prob(A_i^c)}
      \right).
    \end{align*}
    Now
    \[
      \E(\hat{\sigma}_i-\sigma_i)^2 \leq \frac{\Var(\hat{\sigma}_i^2)}{\sigma_i^2},
      \qquad
      \E(\bar{Y}_i^2) = \frac{\mu_i^2 + \sigma_i^2}{m},
    \]
    and $\Var(\hat{\sigma}_i^2) \lesssim \mu_{4,i}/m$. Substituting these bounds yields
    \begin{align*}
      \E\left[\left|\hat{w}_{h,i}(\sqrt{m}\bar{Y}_i;\delta) -
      w_{h,i}(\sqrt{m}\bar{Y}_i;\delta)\right|\right]
      & \lesssim \frac{\mu_{4,i}^{1/2}}{\sqrt{m}}
      \left(
        \frac{1}{\sigma_i h} + \frac{\sqrt{\mu_i^2+\sigma_i^2}}{\sigma_i^2}
      \right)
      + \frac{\mu_{4,i}}{m\sigma_i^4 h}.
    \end{align*}
    Since $h = 1/\sqrt{2\log n}$, the right-hand side is bounded by
    $C_i\sqrt{\log n}/\sqrt{m}$ for a polynomial $C_i$ in
    $|\mu_i|,\sigma_i,\mu_{4,i}$. Averaging over $i$ and absorbing the average polynomial
    into the constant $C$ proves the claim.
  \end{proof}

  \begin{lemma}
    \label{lemma:nongaussian_microdata_term2}
    Under the assumptions of \cref{thm:assure_nongaussian_microdata},
    \[
      \frac{1}{n}\sum_{i=1}^n \left|\E\left[w_{h,i}(\sqrt{m}\bar{Y}_i;\delta_i) -
      w_{h,i}(Z_i;\delta_i)\right]\right|
      \leq C\frac{(\log n)^2}{\sqrt{m}}.
    \]
  \end{lemma}
  \begin{proof}
    Fix $i$ and write $\delta = \delta_i$. Relabel the micro-data $Y_{ij}$ as
    $Y_1,\dots,Y_m$, and let $Z_1,\dots,Z_m$ be i.i.d.
    $\Norm(\mu_i/\sqrt{m},\sigma_i^2)$. Define
    \[
      S_k := \frac{1}{\sqrt{m}}\left(Z_1 + \cdots + Z_k + Y_{k+1} + \cdots + Y_m\right),
    \]
    for $k=0,\dots,m$, and
    \[
      S_k' := \frac{1}{\sqrt{m}}\left(Z_1 + \cdots + Z_k + Y_{k+2} + \cdots + Y_m\right),
    \]
    for $k=0,\dots,m-1$. Then $S_0 = \sqrt{m}\bar{Y}_i$, $S_m$ has the same distribution as
    $Z_i$, and
    \[
      S_k = S_k' + \frac{Y_{k+1}}{\sqrt{m}},
      \qquad
      S_{k+1} = S_k' + \frac{Z_{k+1}}{\sqrt{m}}.
    \]
    Taylor-expanding $x \mapsto w_{h,i}(x;\delta)$ around $S_k'$ to third order, and using
    that $Y_{k+1}$ and $Z_{k+1}$ are independent of $S_k'$ and have matching first two
    moments, we obtain
    \[
      \left|\E w_{h,i}(S_{k+1};\delta) - \E w_{h,i}(S_k;\delta)\right|
      \lesssim \frac{1}{m^{3/2}}
      \left(
        \E\left[|Y_{k+1}|^3 |w_{h,i}^{(3)}(\tilde{x})|\right]
        + \E\left[|Z_{k+1}|^3 |w_{h,i}^{(3)}(\tilde{x}')|\right]
      \right),
    \]
    where $\tilde{x}$ lies between $S_k$ and $S_k'$ and $\tilde{x}'$ lies between
    $S_{k+1}$ and $S_k'$.

    Since $k_i = 0$ under the assumptions of \cref{thm:assure_nongaussian_microdata},
    direct differentiation gives
    \[
      w_{h,i}^{(3)}(x;\delta)
      = \frac{3}{(\sigma_i h)^2}\sinc'\left(\frac{x-\delta}{\sigma_i h}\right)
      + \frac{x}{(\sigma_i h)^3}\sinc''\left(\frac{x-\delta}{\sigma_i h}\right)
      - \frac{1}{\sigma_i^2 h^4}\sinc'''\left(\frac{x-\delta}{\sigma_i h}\right),
    \]
    so the uniform boundedness of the derivatives of $\sinc$ yields
    \[
      |w_{h,i}^{(3)}(x;\delta)| \lesssim \frac{1}{\sigma_i^2 h^4}
      + \frac{|x|}{\sigma_i^3 h^3}.
    \]
    Also,
    \[
      |\tilde{x}| \leq \frac{|Y_{k+1}|}{\sqrt{m}} + |S_k'|,
      \qquad
      |\tilde{x}'| \leq \frac{|Z_{k+1}|}{\sqrt{m}} + |S_k'|,
    \]
    and
    \[
      \E|S_k'| \leq \sqrt{\E[(S_k')^2]} \lesssim \sqrt{\mu_i^2 + \sigma_i^2}.
    \]
    Finally,
    \[
      \E|Y_{k+1}|^3 \lesssim \frac{|\mu_i|^3}{m^{3/2}} + \mu_{3,i},
      \qquad
      \E[Y_{k+1}^4] \lesssim \frac{\mu_i^4}{m^2} + \mu_{4,i},
    \]
    and the Gaussian analogues satisfy
    \[
      \E|Z_{k+1}|^3 \lesssim \frac{|\mu_i|^3}{m^{3/2}} + \sigma_i^3,
      \qquad
      \E[Z_{k+1}^4] \lesssim \frac{\mu_i^4}{m^2} + \sigma_i^4.
    \]
    Substituting these moment bounds into the previous display gives
    \[
      \left|\E w_{h,i}(S_{k+1};\delta) - \E w_{h,i}(S_k;\delta)\right|
      \lesssim \frac{P_i}{m^{3/2}}\left(h^{-4} + h^{-3}\right),
    \]
    where $P_i$ is a polynomial in
    $|\mu_i|,\sigma_i,\mu_{3,i},\mu_{4,i},\sigma_i^{-1}$. Summing over
    $k=0,\dots,m-1$ yields
    \[
      \left|\E w_{h,i}(Z_i;\delta_i) - \E w_{h,i}(\sqrt{m}\bar{Y}_i;\delta_i)\right|
      \lesssim \frac{P_i}{\sqrt{m}}\left(h^{-4} + h^{-3}\right).
    \]
    Since $h = 1/\sqrt{2\log n}$, this is bounded by
    $C_i(\log n)^2/\sqrt{m}$ for another polynomial $C_i$ in
    $|\mu_i|,\sigma_i,\mu_{3,i},\mu_{4,i}$. Averaging over $i$ proves the claim.
  \end{proof}

  \begin{lemma}
    \label{lemma:nongaussian_microdata_term3}
    Under the assumptions of \cref{thm:assure_nongaussian_microdata},
    \[
      \left| \frac{1}{n}\sum_{i=1}^n \E\left[w_{h,i}(Z_i;\delta_i)\right] - W(\delta) \right|
      \leq C\left(\frac{1}{n\log n} + \frac{1}{\sqrt{m}}\right).
    \]
  \end{lemma}
  \begin{proof}
    Write
    \[
      W_G(\delta) := \frac{1}{n}\sum_{i=1}^n \mu_i \bar\Phi\left(\frac{\delta_i - \mu_i}{\sigma_i}\right).
    \]
    Then
    \[
      \left| \frac{1}{n}\sum_{i=1}^n \E\left[w_{h,i}(Z_i;\delta_i)\right] - W(\delta) \right|
      \leq \left| \frac{1}{n}\sum_{i=1}^n \E\left[w_{h,i}(Z_i;\delta_i)\right] - W_G(\delta)
      \right|
      + |W_G(\delta) - W(\delta)|.
    \]
    For the first term, \cref{thm:assure_bias}(1) with $k=0$ gives
    \[
      \left| \E\left[w_{h,i}(Z_i;\delta_i)\right] -
      \mu_i \bar\Phi\left(\frac{\delta_i - \mu_i}{\sigma_i}\right)\right|
      \leq |\mu_i| h^2 e^{-1/(2h^2)}.
    \]
    Averaging over $i$ and using $h = 1/\sqrt{2\log n}$ yields
    \[
      \left| \frac{1}{n}\sum_{i=1}^n \E\left[w_{h,i}(Z_i;\delta_i)\right] - W_G(\delta) \right|
      \lesssim \frac{1}{n\log n}.
    \]
    For the second term, Berry--Esseen gives
    \[
      \sup_x \left|F_i^{(m)}(x) - \Phi(x)\right| \leq \frac{\mu_{3,i}}{\sigma_i^3 \sqrt{m}},
    \]
    and therefore
    \[
      |W_G(\delta) - W(\delta)|
      \leq \left(\frac{1}{n}\sum_{i=1}^n |\mu_i| \sup_x \left|\bar F_i^{(m)}(x) - \bar\Phi(x)\right|\right)
      \leq \left(\frac{1}{n}\sum_{i=1}^n \frac{|\mu_i|\mu_{3,i}}{\sigma_i^3}\right)
      \frac{1}{\sqrt{m}}.
    \]
    Absorbing the average polynomial into the constant $C$ proves the claim.
  \end{proof}



{

\subsection{Minimax pointwise lower bounds for estimating welfare under square loss}

\begin{theorem}
\label{thm:minimax_welfare_lower_bound}
Let $\Theta_{2,M}$ denote the set of vectors $\mu_{1:n}$ such that $\frac{1}{n}\sum_{i=1}^n \mu_i^2 \leq M^2.$ Let $W(\beta) := \frac{1}{n}\sum_{i=1}^n \mu_i \Phi(\mu_i - \beta)$ be the welfare in \eqref{eq:general_welfare} with zero costs and decision rules $\delta(z_i;\beta) := \beta$.
Then for any fixed $\beta \in \bR,$
\begin{equation}
\inf_{\hat W} \sup_{\mu_{1:n} \in \Theta_{2,M}} \E \left[\left(\hat W  - W(\beta)\right)^2\right]\geq C_{M,\beta}\frac{\sqrt{\log n}}{n}.
\end{equation}
\end{theorem}
\begin{proof}
By \cref{lemma:truncated_exp}, recall that 
\[
\mu \Phi(\mu - \beta) = \E_\mu[Y \mathbf{1}\set{Y > \beta} ] - \phi(\mu - \beta)
\]
and define $T(\beta) := \frac{1}{n}\sum_{i=1}^n \phi(\mu_i - \beta)$.
Given any estimator $\hat{W}$ of $W(\beta)$, we may obtain an estimator of $T$ by considering
\begin{equation}
\hat{T} := \frac{1}{n}\sum_{i=1}^n Y_i \mathbf{1}\set{Y_i > \beta} - \hat{W}. 
\end{equation}
Then,
\begin{align*}
\E\left[\left(\hat{T} - T(\beta) \right)^2 \right] & \leq 2\E\left[\left(\hat{W} - W(\beta) \right)^2 \right] + \frac{2}{n^2}\sum_{i=1}^n \Var\set{Y_i \mathbf{1}\set{Y_i > \beta}} \\
& \leq 2\E\left[\left(\hat{W} - W(\beta) \right)^2 \right] + \frac{2(1+M^2)}{n}.
\end{align*}
Therefore, showing that 
\begin{equation}\label{eq:lower_bound_MSE}
\inf_{\hat T} \sup_{\mu_{1:n} \in \Theta_{2,M}} \E \left[\left(\hat T  - T(\beta)\right)^2\right]\geq C_{M, \beta}\frac{\sqrt{\log n}}{n}
\end{equation}
will establish the claim. \\



For this lower bound, will adapt a lower bound construction from \cite{kim2014minimax} and apply the multivariate Van Trees inequality \citep{Tsybakov2009IntroNonparametricEstimation}. The idea is the following. We will define a collection of priors $\tilde{G}_\alpha$ on $\mu_{1:n}$ such that 1) with high probability, $\mu_{1:n}$ drawn iid from $\tilde{G}_\alpha$ belongs to $\Theta_{2,M}$ for all $\alpha$, and 2) the worst case Bayes risk over this set of priors is at least $O(\sqrt{\log n}/n)$
\begin{equation}
\label{eq:minimax_lower_bound_tau}
\inf_{\hat \tau} \sup_{\alpha \in [0,1]^K} \E_\alpha \left[\left(\hat \tau(Y_{1:n}) - \tau(\alpha) \right)^2 \right] \geq c_{M,\beta}\frac{\sqrt{\log n}}{n}.
\end{equation}
In \eqref{eq:minimax_lower_bound_tau}, $Y_1,\dots,Y_n$ are iid with density $\textsf{N}(0,1) * \tilde{G}_\alpha,$ and $\tau(\alpha)$ is the Bayesian analog of the target $T(\beta)$ under the prior $\tilde{G}_\alpha$, given by 
\begin{equation}
\tau(\alpha) = \int \phi(\beta - u) d\tilde{G}_\alpha(u).
\end{equation}
Note that under the Bayes model $\mu_{i} \sim \tilde{G}_\alpha$ iid, $\E_{\alpha}[T(\beta)] = \tau(\alpha)$. The worst case Bayes risk \eqref{eq:minimax_lower_bound_tau} will serve as a lower bound, up to $O(1/n)$ terms, for the minimax risk in \eqref{eq:lower_bound_MSE}.\\

We turn to the construction of these priors. Let $\phi(u)$ be the standard normal density with variance $1$ and let $K = \set{1,3,5,\dots,2m-1}$. It is shown in the proof of Theorem 1.3 of \citet{kim2014minimax} the existence of priors $\Pi_\alpha$, indexed by the hypercube $\alpha \in \set{0,1}^K$, with densities
\begin{equation}
\pi_\alpha(u) = \phi(u) + \e_n \sum_{j \in K}\alpha_j v_j(u),
\end{equation}
for some functions $v_j(u)$. In Kim's construction, $\e_n$ is taken small enough as $\Theta(n^{-1/2})$ and $m \sim \log n$. Define the observation densities
\begin{align*}
g_j(u) & := \phi * v_j \\    
f_\alpha(u) & := \phi * \pi_\alpha.
\end{align*}
These priors satisfy the following properties: 1) they are probability measures, 2) $\frac{1}{2}\phi(u) \leq \pi_\alpha(u) \leq \frac{3}{2}\phi(u)$, 3) $g_j(u)/\sqrt{f_0(u)}$ are equal to $\sqrt{\pi}i^j \psi_j(u/2)$ where $\psi_j(u)$ are the Hermite functions, defined in (7) therein as:
\[\psi_j(u) = \sqrt{-1}^{-j}\sqrt{2\phi(2u)}\frac{H_j(2u)}{\sqrt{j!}},\]
where $H_j$ is the Hermite polynomial of order $j$, the polynomial for which the $j$-th derivative of $\phi(u)$ is $(-1)^j H_j(u) \phi(u)$.
Recall Hermite polynomials are orthogonal under the Gaussian measure $\phi$:
\[\int H_j(u)H_{j'}(u)\phi(u)du = j! \cdot 1(j=j').\]
Then 
\begin{align*}
 \int \psi_j(u)\overline{\psi_{j'}(u)}du 
 & = i^{-j}(-i)^{-j'}\frac{1}{\sqrt{j!(j')!}}\int H_j(2u)H_{j'}(2u)\phi(2u)d(2u)\\
 & = \frac{1}{\sqrt{j!(j')!}}\int H_j(u)H_{j'}(u)\phi(u)du\\
 & = 1(j=j').
\end{align*}
In Kim's paper, $g_j(u)/\sqrt{f_0(u)}$ is chosen such that its Fourier transform is the Hermite function, but taking the inverse Fourier transform gives $\sqrt{\pi}i^j \psi_j(u/2)$ since the Hermite functions are Fourier eigenfunctions.
Now by orthonormality,
\begin{equation}
\label{eq:orthogonality_kim}
    \int \frac{\left(\sum_{j\in K} s_j g_j(u) \right)^2}{f_0(u)} du = 2 \pi \sum_{j \in K} s_j^2
\end{equation}
for every real vector $s \in \bR^K$. Extend the construction so that now, $\alpha$ ranges over $[0,1]^K$. By convexity, it is not difficult to see that the three properties above continue to hold. Moreover, by 2),
\begin{equation}
\label{eq:marginal_density_sandwich}
\frac{1}{2}f_0(u) \leq f_\alpha(u) \leq \frac{3}{2}f_0(u).
\end{equation}

Next, we will adapt the priors $\Pi_\alpha, \alpha \in [0,1]^K$ so they have second moment bounded by $M^2$ and shift them so $\tau(\alpha)$ are different enough across $\alpha \in [0,1]$. First, define $\tilde{\Pi}_{\alpha}$ with density $\pi_\alpha(u - \beta - 2)$ and set $\tilde{f}_\alpha(u) = f_\alpha(u - \beta - 2)$. In addition, define $\tilde{g}_{k}(u) := g_k(u - \beta - 2), k \in K.$ Then
\[
\tilde{f}_\alpha = \tilde{f}_0 + \e_n \sum_{j \in K} \alpha_j \tilde g_j(u).
\]
Define the following priors and their corresponding observation densities
\begin{align*}
\tilde{G}_\alpha & := w \delta_0 + (1-w) \tilde{\Pi}_\alpha \\
\tilde{q}_\alpha(u) & := \phi * \tilde{G}_\alpha =  w \phi(u) + (1-w)\tilde{f}_\alpha(u).
\end{align*}
Where $w$ is a mixing weight chosen sufficiently close to one, depending only on $M$ and $\beta$, so that the second moments satisfy $\int u^2 d\tilde{G}_\alpha(u) \leq M^2$ for all $\alpha \in [0,1]^K$. By \eqref{eq:marginal_density_sandwich}, $\tilde{q}_\alpha(u) \geq (1-w)\tilde{f}_0(u)/2$. \\

We will now show the worst-case Bayes regret in \cref{eq:minimax_lower_bound_tau} is bounded below by $\sqrt{\log n}/n$. This will be done by the multivariate van Trees inequality \citep{gassiat2024van}. Note that $\tau(\alpha) = \tilde{q}_\alpha(\beta)$ and the partial derivative $\partial_k q_\alpha(u)$ is given by $(1-w)\e_n \tilde{g}_k(u)$. The $n$-sample Fisher information matrix $\cl{I}_n(\alpha)$, for any $s \in \bR^K$ satisfies, using \eqref{eq:marginal_density_sandwich} and \eqref{eq:orthogonality_kim},
\begin{align*}
    s^\top \cl{I}_n(\alpha) s & = n(1-w)^2\e_n^2\int \frac{\left(\sum_{j\in K} s_j \tilde{g}_j(u) \right)^2}{\tilde{q}_\alpha(u)} du \\
    & \leq 2n(1-w)\e_n^2 \int \frac{\left(\sum_{j\in K} s_j \tilde{g}_j(u) \right)^2}{\tilde{f}_0(u)} du \\
    & = 4\pi n(1-w)\e_n^2 \norm{s}_2^2 \\
    & \leq 4\pi(1-w) \norm{s}_2^2.
\end{align*}
The last line holds by the choice of $\e_n$. Therefore, $\cl{I}_n(\alpha) \preceq 4\pi(1-w) I_{m\times m}.$  Let $\mathbb{Q}$ be the product prior on $[0,1]^K$ with one-dimensional density $\rho$ proportional to $t^2(1-t)^2$. This prior vanishes at the boundary as is required by the van Trees inequality and has its one-dimensional marginals have finite Fisher information $\cl{I}_{\rho} < \infty$. Next, note that 
\[
\nabla \tau(\alpha) = (1-w)\e_n \left(\tilde{g}_j(\beta) \right)_{j \in K} = (1-w)\e_n \left(g_j(-2) \right)_{j \in K}.
\]
which is constant in $\alpha$. By multivariate van Trees labelled (vTm) in \citet{gassiat2024van}, we have
\begin{align*}
    \inf_{\hat \tau} \sup_{\alpha \in [0,1]^K} \E_\alpha \left[\left(\hat \tau(Y_{1:n}) - \tau(\alpha) \right)^2 \right] & \geq \inf_{\hat \tau} \E_{\mathbb{Q}}\left[ \E_\alpha \left[\left(\hat \tau(Y_{1:n}) - \tau(\alpha) \right)^2 \right]\right] \\
    & \geq (\nabla \tau)^\top \left(\E_{\mathbb{Q}} I_n(\alpha) + \cl{I}_\rho I_{m \times m} \right)^{-1} \nabla \tau \\
    & \geq \frac{\norm{\nabla \tau}_2^2}{4\pi(1-w) + \cl{I}_\rho}.
\end{align*}
We have $\norm{\nabla \tau}_2^2 = (1-w)^2\e_n^2 \sum_{j \in K} g_j(-2)^2 = (1-w)^2\e_n^2 2\pi f_0(-2) \sum_{j \in K} |\psi_j(-1)|^2 \asymp n^{-1}(\sum_{j \in K} |\psi_j(-1)|^2)$, where $\psi_j$ are the Hermite functions. The second to last scaling holds since $\sqrt{\pi}i^j\psi_j(u/2) = g_j(u)/\sqrt{f_0(u)}$.

We claim that
\[
\sum_{j \in K} |\psi_j(-1)|^2 \asymp \sqrt{m} \asymp \sqrt{\log n}.
\]
We give a proof using the Mehler formula; there are alternative proofs using the Christoffer-Darboux formula and the Plancherel-Rotach asymptotic for Hermites. By definition, 
\[\sum_{j \in K} |\psi_j(-1)|^2 \asymp \sum_{k=1}^{m} \frac{1}{(2k-1)!}H_{2k-1}(-2)^2.\]
By Mehler's formula \citep[e.g.][Proposition 2.2]{liang2022mehler} with $x=y=-2$, for any $\rho \in (-1,1)$, 
\[\frac{1}{\sqrt{1-\rho^2}}\exp \left\{-\frac{4\rho}{\rho+1}\right\} = \sum_{j\ge 0}\frac{\rho^j}{j!}H_j(-2)^2.\]
Replacing $\rho$ with $-\rho$, we have
\[\frac{1}{\sqrt{1-\rho^2}}\exp \left\{-\frac{4\rho}{\rho-1}\right\} = \sum_{j\ge 0}(-1)^j\frac{\rho^j}{j!}H_j(-2)^2.\]
Taking the difference implies
\[ 2\sum_{k\ge 1}\frac{H_{2k-1}(-2)^2}{(2k-1)!}\rho^{2k-1} = \frac{1}{\sqrt{1-\rho^2}}\left(\exp \left\{-\frac{4\rho}{\rho+1}\right\} - \exp \left\{-\frac{4\rho}{\rho-1}\right\}\right).\]
Letting $z = \rho^2$, we have 
\[2\sum_{k\ge 0}\frac{H_{2k-1}(-2)^2}{(2k-1)!}z^k = \frac{1}{z^{1/2}\sqrt{1-z}}\left(\exp \left\{-\frac{4\sqrt{z}}{\sqrt{z}+1}\right\} - \exp \left\{-\frac{4\sqrt{z}}{\sqrt{z}-1}\right\}\right).\]
As $z\rightarrow 1^{-}$, 
\[\sum_{k\ge 0}\frac{H_{2k-1}(-2)^2}{(2k-1)!}z^k \sim (1-z)^{-1/2}.\]
Since all coefficients are non-negative, \cite{karamata1931neuer} implies 
\[\sum_{k=1}^{m}\frac{H_{2k-1}(-2)^2}{(2k-1)!}\asymp \sqrt{m}.\]

This establishes the minimax lower bound in \cref{eq:minimax_lower_bound_tau}.\\

Finally, we transfer this lower bound to the compound case. This is standard. Define the event $A$ where the prior $\tilde{G}_\alpha$ has $\frac{1}{n}\sum_{i=1}^n \mu_i^2 \leq M^2$. By the density sandwich \eqref{eq:marginal_density_sandwich}, the priors $\tilde{G}_\alpha$ are sub-Gaussian with a uniformly controlled sub-Gaussian parameter. Therefore the probability $\Prob_\alpha(A^c)$ is exponentially small. Without loss, any estimator $\hat T$ can be clipped to $[0,\max_x \phi(x)]$ so it has uniformly bounded MSE. Then
\[
\E_{\alpha}\left[\left(\hat{T} - T(\beta)\right)^2\right] \leq \sup_{\Theta_{2,M}} \E \left[\left(\hat T  - T(\beta)\right)^2\right] + Ce^{-cn}.
\]
Lastly,
\[
\E_{\alpha}\left[\left(\hat{T} - \tau(\alpha)\right)^2\right] \lesssim  \E_{\alpha}\left[\left(\hat{T} - T(\beta)\right)^2\right] + \E_{\alpha}\left[\left(\tau(\alpha) - T(\beta)\right)^2\right]
\]
Note that $\E_{\alpha}[T(\beta)] = \frac{1}{n}\sum_{i=1}^n \E_{\mu_i \sim \tilde{G}_\alpha}[ \phi(\mu_i - \beta)] = \tau(\alpha)$. Then, $\sup_\alpha \E_{\alpha}\left[\left(\tau(\alpha) - T(\beta)\right)^2\right] = \frac{1}{n}\Var_{\mu \sim \tilde{G}_\alpha}(\phi(\beta - \mu)) = O(1/n)$ uniformly in $\alpha.$ Therefore, 
\[
\sup_{\Theta_{2,M}} \E \left[\left(\hat T  - T(\beta)\right)^2\right] \geq c_{M,\beta}\frac{\sqrt{\log n}}{n} - O(1/n) - Ce^{-cn},
\]
as desired.
\end{proof}

}

  \section{Results for \assurecb{}}
  \label{appx sec:cb}
  In this section, we derive the counterpart of \cref{thm:assure_bias} and \cref{thm:main_regret_bound} for \assurecb{} to corroborate the rate comparison in \cref{sec:heuristic}, further assuming that $\mu_{1:n}$ and $k_{1:n}$ are uniformly bounded. Define the \assurecb{} estimator as
  \begin{align*}
    \hat{W}_{\text{CB}}\left(\beta;\epsilon\right) = \frac{1}{n} \sum_{i=1}^{n} \left[ \left( Y_i -k_i \right) \cdot \Phi \left( \frac{Y_i - \delta\left(z_i;\beta\right)}{\epsilon\sigma_i}  \right) - \frac{ \sigma_i}{\epsilon} \cdot \varphi \left( \frac{Y_i - \delta\left(z_i;\beta\right)}{\epsilon\sigma_i}  \right) \right].
  \end{align*}
  By arguments in \cref{sec:heuristic},  $ \hat{W}_{\text{CB}}\left(\beta;\epsilon\right)$ is an unbiased estimator for a welfare function
  \begin{align*}
    W(\beta;\epsilon) &= \E\left[ \frac{1}{n}\sum_{i=1}^{n} \left(\mu_i -k_i \right) \mathbf{1}\left\{ Y_i - \epsilon
    \sigma_i Q_i \geq \delta\left(z_i;\beta\right) \right\}\right]= \E\left[ V_n(\beta;\epsilon) \right],
  \end{align*}
  where $Q_i \sim \Norm(0,1)$ and $Q_i \indep Y_i, z_i$.
  Define the welfare function of our interest.
  \begin{align*}
    W(\beta;0) &= \E\left[  \frac{1}{n}\sum_{i=1}^{n} \left(\mu_i -k_i \right) \mathbf{1}\left\{ Y_i \geq \delta\left(z_i;\beta\right) \right\} \right]= \E\left[ V_n(\beta;0) \right].
  \end{align*}

  \begin{proposition}[MSE for \assurecb{}]
      Under the same bounded-parameter assumptions as in \Cref{thm:main_regret_bound}, further assuming uniformly bounded $\mu_{1:n}$ and $k_{1:n}$, for $\epsilon_n = n^{-\frac{1}{5}}$ and fixed $\beta$, the squared error of \assurecb{} satisfies
      \begin{align*}
          \E\left[\left(\hat W_{\text{CB}} \left(\beta;\epsilon_n\right) - W\left(\beta;0\right)\right)^2 \right] = O\left(n^{-\frac{4}{5}}\right).
      \end{align*}
  \end{proposition}

  \begin{proof}
    As $\hat W_{\text{CB}}\left(\beta;\epsilon_n\right)$ is an unbiased estimator of $W\left(\beta;\epsilon_n\right)$, by \Cref{prop:omega_true_hetero}, the bias is
    \begin{align*}
       \abs{\E\left[\hat W_{\text{CB}}\left(\beta;\epsilon_n\right) \right] - W(\beta;0)}
    &= O(\epsilon_n^2) = O\left(n^{-\frac{2}{5}}\right).
    \end{align*}
    By independence across $i$, the variance is
    \begin{align*}
      \Var\left(\hat W_{\text{CB}}\left(\beta;\epsilon_n\right)\right) &\leq  \frac{1}{n^2} \sum_{i=1}^{n}  \E\left[\left(Y_i - k_i\right)^2  \right] + \frac{1}{n^2}\sum_{i=1}^{n} \frac{\sigma_i^2}{\epsilon_n^2} \E\left[\left(\varphi \left( \frac{Y_i - \delta\left(z_i;\beta\right)}{\epsilon_n\sigma_i}  \right)\right)^2 \right]\\
      &= \frac{1}{n} \left(\frac{1}{n} \sum_{i=1}^{n}\left( \sigma_i^2 + \left(\mu_i - k_i\right)^2\right) \right) + \frac{1}{n^2\epsilon_n^2}\sum_{i=1}^{n} \sigma_i^2 \E\left[\left(\varphi \left( \frac{Y_i - \delta\left(z_i;\beta\right)}{\epsilon_n\sigma_i}  \right)\right)^2 \right].
    \end{align*}
    By \Cref{lemma:second_monent}, the second term can be bounded as
    \begin{align*}
      \frac{1}{n^2\epsilon_n^2}\sum_{i=1}^{n} \sigma_i^2 \E\left[\left(\varphi \left( \frac{Y_i - \delta\left(z_i;\beta\right)}{\epsilon_n\sigma_i}  \right)\right)^2 \right] &\lesssim \frac{1}{n\epsilon_n}\frac{1}{n}\sum_{i=1}^{n}\sigma_i^2 = O\left(n^{-\frac{4}{5}}\right),
    \end{align*}
    which shows that the variance is $O\left(n^{-\frac{4}{5}}\right)$. Therefore, the result follows from the bias-variance decomposition.
  \end{proof}

  \begin{theorem}[Main regret bound for \assurecb{}]
    \label{thm:cb_regret_bound} Let $\hat\beta (\epsilon)$ be the maximizer of $\hat{W}_{\text{CB}}(\beta;\epsilon)$, $\beta^\ast(\epsilon)$ be the maximizer of $W(\beta;\epsilon)$ and $\beta_0$ be the maximizer of $W(\beta;0)$.  Under the same bounded-parameter assumptions as in \cref{thm:main_regret_bound}, further assuming uniformly bounded $\mu_{1:n}$ and $k_{1:n}$, with $\epsilon_n\to0$ the regret satisfies
    \begin{align*}
      W( \beta_0; 0 ) - \E\left[ V_n(\hat \beta(\epsilon_n);0) \right] =  O(\gamma(n)),
    \end{align*}
    where
    \begin{align*}
      \gamma(n) = \max \left\{\epsilon_n^2, \frac{1}{\sqrt{n}}, \frac{\epsilon_n^{-\frac{1}{2}} \left(-\log\epsilon_n \right)^{\frac{1}{2}}}{\sqrt{n}}, \frac{\epsilon_n^{-1}\left( -\log \epsilon_n\right)}{n}\right\}.
    \end{align*}
    Therefore the optimal choice of $\epsilon_n$ is $\epsilon_n = n^{-\frac{1}{5}}$, and then $\gamma(n) = n^{-\frac{2}{5}}\left( \log n \right)^{\frac{1}{2}}$.
  \end{theorem}

  \begin{proof}[Proof of \cref{thm:cb_regret_bound}]
    The regret can be bounded by
    \begin{align*}
      W(\beta_0;0) - \E\left[V_n(\hat\beta(\epsilon_n) ;0)\right] &= \E\left[ V_n(\beta_0 ;0) -  V_n(\hat\beta(\epsilon_n) ;0) \right]\\
      &\leq \E\left[ V_n(\beta_0 ;0) -  \hat W_{\text{CB}}(\beta_0;\epsilon_n) \right] + \E\left[ \hat W_{\text{CB}}(\hat \beta(\epsilon_n);\epsilon_n) -  V_n(\hat\beta(\epsilon_n) ;0) \right] \\
      &\leq 2 \E\left[ \sup_\beta \left| \hat W_\text{CB} (\beta;\epsilon_n) - V_n (\beta;0) \right|  \right]\\
      &\leq 2 \E\left[ \sup_\beta \left| \hat W_{\text{CB}}(\beta;\epsilon_n) - W(\beta;\epsilon_n) \right|  \right]\\
      &\quad + 2\E\left[ \sup_\beta \left| V_n (\beta;0) - \E\left[ V_n (\beta;0) \right] \right| \right]\\
      &\quad + 2 \sup_\beta \left|  W(\beta;\epsilon_n)  - W (\beta;0)   \right|.
    \end{align*}
    With $\epsilon_n \to 0$, from \cref{prop:omega_true_hetero}, we have
    \begin{align*}
      \sup_\beta \left|  W(\beta;\epsilon_n)  - W (\beta;0)   \right| =  O( \epsilon_n^2).
    \end{align*}
    From \cref{prop:omega_ure_hetero}, we have
    \begin{align*}
      &\quad  \E\left[ \sup_\beta \left| \hat W_{\text{CB}}(\beta;\epsilon_n) - W(\beta;\epsilon_n) \right|  \right] + \E\left[ \sup_\beta \left| V_n (\beta;0) - \E\left[ V_n (\beta;0) \right] \right| \right] \\
      &=\max \left\{ O\left(\frac{1}{\sqrt{n}}\right), O\left(\frac{\epsilon_n^{-\frac{1}{2}} \left(-\log \epsilon_n \right)^{\frac{1}{2}}}{\sqrt{n}} \right) , O\left(\frac{\epsilon_n^{-1}\left( -\log \epsilon_n\right)}{n}\right)\right\}.
    \end{align*}
    Combine the two results and we conclude the proof.
  \end{proof}

  \begin{proposition} \label{prop:omega_true_hetero}
    In the proof of \Cref{thm:cb_regret_bound}, if we let $\epsilon_n \to0$, then we have
    \begin{align*}
      \sup_\beta \left| W(\beta ;\epsilon_n) - W( \beta;0 ) \right| = \lVert D\rVert_{n} O( \epsilon_n^2).
    \end{align*}
  \end{proposition}
  \begin{proof}

    For any $\epsilon$, since $Y_i - \epsilon \sigma_i \omega_i \mid z_i \sim \Norm(\mu_i, \left( 1 + \epsilon^2\right)\sigma_i^2)$,
    \begin{align*}
      W(\beta;\epsilon) =  \frac{1}{n} \sum_{i=1}^{n} \left(\mu_i -k_i \right) \Phi \left(\frac{\mu_i -\delta(z_i;\beta)}{\sqrt{1 + \epsilon^2}\sigma_i} \right) .
    \end{align*}

    By the mean value theorem,
    \begin{align*}
      & \left| W(\beta; \epsilon) - W(\beta;0) \right| =   \left|  \frac{\partial}{\partial \epsilon} W(\beta; \epsilon) \big|_{\epsilon= \epsilon_1}  \right| \left|\epsilon-0\right|\\
      &= \left| \frac{1}{n}\sum_{i=1}^{n} \left(\mu_i - k_i\right) \varphi \left(
      \frac{\mu_i -\delta(z_i;\beta)}{\sqrt{1+\epsilon_1^2}\sigma_i} \right)
      \frac{\mu_i -\delta(z_i;\beta)}{\sigma_i}
      \right| \frac{1}{(1+\epsilon_1^2)^{3/2}}\epsilon_1\epsilon\\
      &\lesssim \lVert D \rVert_{n} \frac{1}{(1+\epsilon_1^2)^{3/2}}\epsilon_1\epsilon \leq \epsilon^2 \lVert D \rVert_{n}.
    \end{align*}

  \end{proof}

  \begin{proposition} \label{prop:omega_ure_hetero}
    In the proof of \Cref{thm:cb_regret_bound}, if we let $\epsilon_n \to 0$, then we have
    \begin{align*}
      &\quad \E\left[ \sup_\beta \left| \hat W_{\text{CB}}(\beta;\epsilon_n) - W(\beta;\epsilon_n) \right|  \right]+ \E\left[ \sup_\beta \left| V_n (\beta;0) - \E\left[ V_n (\beta;0) \right] \right| \right]\\
      &= \max \left\{ O\left(\frac{1}{\sqrt{n}}\right), O\left(\frac{\epsilon_n^{-\frac{1}{2}} \left(-\log \epsilon_n \right)^{\frac{1}{2}}}{\sqrt{n}}\right) , O\left(\frac{\epsilon_n^{-1}\left( -\log \epsilon_n\right)}{n}\right)\right\}.
    \end{align*}
  \end{proposition}
  \begin{proof}
    Define $\mathcal{D} = \left\{\delta\left(z;\beta\right) \right\}$ and $V(\mathcal{D})$ as its VC-dimension. With $\epsilon_n$, first we have
    \begin{align*}
      &\quad \E\left[\sup_{\beta} \left|  \hat{W}_{\text{CB}} ( \beta;\epsilon_n) - W(\beta;\epsilon_n)\right| \right]\\
      &= \E\left[\sup_{\beta} \left|   \hat{W}_{\text{CB}} ( \beta;\epsilon_n) -\E\left[  \hat{W}_{\text{CB}} ( \beta;\epsilon_n)\right] \right| \right] & \\
      &\leq \E\left[\sup_\beta \left| \frac{1}{n}\sum_{i=1}^{n}\left(Y_i - k_i \right) \Phi\left(\frac{Y_i - \delta\left(z_i;\beta\right)}{\epsilon_n \sigma_i} \right) - \E\left[\left(Y_i - k_i \right) \Phi\left(\frac{Y_i - \delta\left(z_i;\beta\right)}{\epsilon_n \sigma_i} \right) \right] \right| \right]\\
      &\quad + \frac{1}{\epsilon_n}\E\left[\sup_\beta \left| \frac{1}{n}\sum_{i=1}^{n} \sigma_i \varphi\left(\frac{Y_i - \delta\left(z_i;\beta\right)}{\epsilon_n \sigma_i} \right) - \E\left[\sigma_i \varphi\left(\frac{Y_i - \delta\left(z_i;\beta\right)}{\epsilon_n \sigma_i} \right) \right] \right| \right]\\
      &=  T_{1,n} + \frac{1}{\epsilon_n} T_{2,n},
    \end{align*}
    and
    \begin{align*}
      \E\left[\sup_\beta \left| V_n(\beta;0) - W(\beta;0) \right| \right] = \E\left[\sup_\beta \left| V_n(\beta;0) - \E\left[V_n(\beta;0)\right] \right| \right] = T_{3,n}.
    \end{align*}
    It suffices to bound $T_{1,n}$, $\frac{1}{\epsilon_n} T_{2,n}$ and $T_{3,n}$. By the symmetrization lemma for the i.n.i.d case, we derive the bound as follows.
    Conditional on $Y, Z$, denote the empirical measure of as $\mathbb{P}_n$. First notice that the function class below is a VC-subgraph class:
    \begin{align*}
      \mathcal{F}_1 = \left\{ f_{\beta}\left(y,z\right) =  \left(y -k \right) \Phi \left( \frac{y - \delta\left(z;\beta\right)}{\epsilon_n \sigma}\right) \right\}.
    \end{align*}
    The result is from the following. First, $\left\{ \frac{y - \delta\left(z;\beta\right)}{\epsilon_n \sigma} \right\}$ is a VC-subgraph class the VC dimension $V(\mathcal{D})$ by Lemma 2.6.18.vi of \citet{vaart1996weak} from multiplicity. Since $t\mapsto \Phi\left(t\right)$ is a monotone function, by Lemma 2.6.18.viii of \citet{vaart1996weak} and by Lemma 2.6.18.vi of \citet{vaart1996weak} with multiplicity, $\mathcal{F}_1$ is a VC-subgraph class with the VC dimension $V(\mathcal{D})$.

    Since $\mathcal{F}_1$ has an envelope function: $\left| f_{\beta} \left(y,z\right) \right| \leq \left|y-k \right|$, from Theorem 2.6.7 of \citet{vaart1996weak}, we have for the covering number
    \begin{align*}
      N\left(\eta \lVert y-k \rVert_{L_2, \mathbb{P}_n}, \mathcal{F}_1, L_2(\mathbb{P}_n)
      \right) \leq K V(\mathcal{D}) \left(16 e\right)^{V(\mathcal{D})} \left(\frac{1}
      {\eta}\right)^{2(V(\mathcal{D})-1)},
    \end{align*}
    where $K$ is a universal constant. Applying the maximal inequality similar to Theorem 2.14.1 of \citet{vaart2023empirical}, we obtain the following:
    \begin{align*}
      T_{1,n}
      &\lesssim \E \left[ \frac{1}{\sqrt{n}} \lVert y-k \rVert_{L_2, \mathbb{P}_n} \int_{0}^{1} \sqrt{1 + \log N\left(\eta\lVert y-k \rVert_{L_2, \mathbb{P}_n}, \mathcal{F}_1, L_2\left(\mathbb{P}_n\right) \right) }d\eta \right] \\
      &\leq \E\left[\frac{1}{\sqrt{n}}\lVert y-k \rVert_{L_2, \mathbb{P}_n}  \int_0^{1} \left(1 + \log \left(A_1 \left(\frac{1}{\eta}\right)^{A_2}\right)\right)^{\frac{1}{2}} d\eta\right]\\
      &\leq A_3\frac{1}{\sqrt{n}} \E\left[\lVert y-k \rVert_{L_2, \mathbb{P}_n}\right]\\
      &\leq A_3 \frac{1}{\sqrt{n}}\left(  \E\left[ \lVert y-\mu \rVert_{L_2, \mathbb{P}_n}\right] +  \lVert \mu-k \rVert_{L_2, \mathbb{P}_n}\right)\\
      &\leq A_3 \frac{1}{\sqrt{n}} \left( \sqrt{ \E\left[ \frac{1}{n}\sum_{i=1}^{n} \left(Y_i - \mu_i\right)^2 \right] } + m_2\right)  = A_3 \frac{1}{\sqrt{n}}\left(s_2 + m_2 \right)
    \end{align*}
    where $s_q, m_q$ are defined before \cref{prop:realized_regret_term_empirical_process_bound}, and $A_1, A_2, A_3$ are universal constants.

    Similarly, first notice that the function class below is a VC-subgraph class:
    \begin{align*}
      \mathcal{F}_2 = \left\{ f_\beta\left(y,z\right) =   \sigma \varphi \left( \frac{y -
      \delta\left(z;\beta\right)}{\epsilon_n\sigma} \right) \right\}.
    \end{align*}
    The result is from the following. First, $\left\{ \left| \frac{y-\delta\left(z;\beta\right)}{\epsilon_n \sigma} \right|\right\}$ is a VC subgraph class with dimension $V(\mathcal{D})$ by multiplicity (Lemma 2.6.18.vi of \citet{vaart1996weak}) and part 5 of \cref{lemma:vc_indexes}. Since $t\mapsto t^2$ is a monotone mapping for non-negative $t$, and $t \mapsto e^{-t}$ is a monotone mapping, by Lemma 2.6.18.viii of \citet{vaart1996weak} and by Lemma 2.6.18.vi of \citet{vaart1996weak} with multiplicity, $\mathcal{F}_2$ is a VC-subgraph class with dimension $V(\mathcal{D})$.

    Since $\mathcal{F}_2$ is bounded $\left| f_{\beta} \left(y,z\right) \right| \leq M$, from Theorem 2.6.7 of \citet{vaart1996weak}, we have for the covering number
    \begin{align*}
      N\left(\eta M, \mathcal{F}_2, L_2(\mathbb{P}_n) \right) \leq K V(\mathcal{D}) \left(16 e\right)^{V(\mathcal{D})} \left(\frac{1}{\eta}\right)^{2(V(\mathcal{D})-1)},
    \end{align*}
    where $K$ is a universal constant. For the entropy  defined in \Cref{sec:inid_theory},
    \begin{align*}
      J\left(b, \mathcal{F}_2 \mid M, L_2 \right) &= \sup_Q \int_{0}^{b} \sqrt{1 + \log N\left(\eta M, \mathcal{F}_2, L_2(Q) \right)} d\eta,
    \end{align*}
    we have $J\left(b, \mathcal F_2\mid M, L_2\right) = O(b\sqrt{\log\left
    (1/b\right)})$ as $b
    \downarrow 0$ \citep[before Theorem 2.14.1, p. 330]{vaart2023empirical}

    Since \cref{lemma:second_monent} shows that $\E f^2 \leq b_n^2 M^2$, with $b_n = \epsilon_n^{1/2}$ for $\forall f
    \in \mathcal{F}_2$, applying \Cref{cor:inid_localized_donsker_bound}, we have
    \begin{align*}
      T_{2,n}
      &\lesssim \frac{1}{\sqrt{n}} J\left(b_n, \mathcal{F}_2 \mid M , L_2\right) \left( 1 + \frac{J\left(b_n, \mathcal{F}_2 \mid M, L_2\right)}{b_n^2 \sqrt{n} }\right) M\\
      &\lesssim \frac{1}{\sqrt{n}}b_n \sqrt{ \log \frac{1}{b_n}} + \frac{1}{n} \log \left(\frac{1}{b_n}\right)\lesssim \max \left\{ \frac{\epsilon_n^{\frac{1}{2}}\sqrt{\pr{-\log \epsilon_n}}}{\sqrt{n}}, \frac{\pr{-\log \epsilon_n}}{n} \right\}.
    \end{align*}

    For $T_{3,n}$, define the function class
    \begin{align*}
      \mathcal{F}_3 = \left\{ f_\beta\left(y,z,\mu\right) =  \left(\mu - k\right)\mathbf{1}\left\{y  \geq \delta \left(z;\beta\right) \right\}  \right\}.
    \end{align*}
    By the additive, multiplicative and monotone composition properties of VC-subgraph classes \citep[Lemma 2.6.18]{vaart1996weak}, $\mathcal{F}_3$ is a VC-subgraph class with dimension $V(\mathcal{D})$. Since $\mathcal{F}_3$ has an envelope function: $\left| f_{\beta} \left(y,z,\mu\right) \right| \leq \left|\mu-k \right|$, from Theorem 2.6.7 of \citet{vaart1996weak} and the maximal inequality similar to Theorem 2.14.1 of \citet{vaart2023empirical}, we have
    \begin{align*}
      T_{3,n} &\lesssim \E \left[ \frac{1}{\sqrt{n}} \lVert \mu-k \rVert_{L_2, \mathbb{P}_n} \int_{0}^{1} \sqrt{1 + \log N\left(\eta\lVert \mu-k \rVert_{L_2, \mathbb{P}_n}, \mathcal{F}_3, L_2\left(\mathbb{P}_n\right) \right) }d\eta \right] \\
      &\leq \E\left[\frac{1}{\sqrt{n}}m_2  \int_0^{1} \left(1 + \log \left(A_1 \left(\frac{1}{\eta}\right)^{A_2}\right)\right)^{\frac{1}{2}} d\eta\right]\leq A_3\frac{1}{\sqrt{n}} m_2,
    \end{align*}
    where $s_q, m_q$ are defined before \cref{prop:realized_regret_term_empirical_process_bound}, and $A_1, A_2, A_3$ are universal constants.

    As a conclusion, we have the bound:
    \begin{align*}
      T_{1,n} + \frac{1}{\epsilon_n} T_{2,n} + T_{3,n} &\lesssim \frac{1}{\sqrt{n}} + \frac{1}{\epsilon_n} \max \left\{ \frac{\epsilon_n^{\frac{1}{2}}\sqrt{\pr{-\log \epsilon_n}}}{\sqrt{n}}, \frac{\pr{-\log \epsilon_n}}{n} \right\} + \frac{1}{\sqrt{n}}\\
      & \lesssim \max \left\{ O\pr{\frac{1}{\sqrt{n}}}, O\pr{\frac{\epsilon_n^{-\frac{1}{2}} \pr{-\log \epsilon_n}^{\frac{1}{2}}}{\sqrt{n}}}, O\pr{\frac{\epsilon_n^{-1}\pr{-\log \epsilon_n}}{n}} \right\}.
    \end{align*}
  \end{proof}

  \begin{lemma}
    \label{lemma:second_monent}
    In the proof of \Cref{prop:omega_ure_hetero}, for $\epsilon_n\to0$, we have
    \begin{align*}
      \sup_{\beta} \E\left[\frac{1}{n}\sum_{i=1}^{n} \varphi^2 \left(\frac{Y_i - \delta(z;\beta)}{\epsilon_n\sigma_i} \right) \right] \leq \frac{1}{2\pi\sqrt{2}}\epsilon_n.
    \end{align*}
  \end{lemma}
  \begin{proof}
    We have
    \begin{align*}
      &\quad \E\left[ \varphi^2 \left( \frac{Y_i - \delta(z_i;\beta)}{\epsilon_n \sigma_i} \right) \right]= \frac{1}{\sqrt{2\pi}} \E\left[ \varphi\left( \sqrt{2}\frac{Y_i - \delta\left(z_i;\beta\right)}{\epsilon_n \sigma_i} \right) \right]\\
      &= \frac{1}{\sqrt{2\pi}} \int \frac{1}{\sqrt{2\pi}} \frac{1}{\sigma_i}\exp \left\{ -\frac{1}{2\sigma_i^2} \left[ 2\frac{\left( y - \delta(z_i;\beta)\right)^2}{\epsilon_n^2} + \left(y-\mu_i\right)^2\right]\right\}  dy\\
      &= \frac{1}{\sqrt{2\pi}} \int \frac{1}{\sqrt{2\pi}} \frac{1}{\sigma_i}\exp \bigg\{ -\frac{1}{2\sigma_i^2} \bigg[ \left(1+ \frac{2}{\epsilon_n^2}\right)\left( y - \frac{2}{2+\epsilon_n^2} \delta(z_i;\beta) - \frac{\epsilon_n^2}{2+\epsilon_n^2}\mu_i \right)^2 \\
      & \hspace{1cm} + \frac{2}{2+\epsilon_n^2}\left(\delta(z_i;\beta) - \mu_i \right)^2 \bigg]\bigg \}  dy\\
      &= \frac{1}{\sqrt{2\pi}}\frac{1}{\sqrt{1+\frac{2}{\epsilon_n^2}}} \varphi \left( \frac{
      \sqrt{2}}{\sigma_i\sqrt{2+\epsilon_n^2}} \left( \delta(z_i;\beta) - \mu_i\right) \right)\leq \frac{1}{2\pi\sqrt{2}}\epsilon_n, \forall \beta.
    \end{align*}
  \end{proof}

  \section{Empirical processes for i.n.i.d. data}
  \label{sec:inid_theory}

  The next lemma records how the VC index changes under composition for the results in Lemma 2.6.20 of \cite{vaart2023empirical}, alongside some additional useful results. Throughout this section, let $\psi_2 = e^{x^2} - 1$.
  \begin{lemma}
    \label{lemma:vc_indexes}
    Let $\cl{F}$ be a VC subgraph class with index $V(\cl{F})$ and let $g:\cl{X} \rightarrow \bb{R}$, $\phi: \cl{X} \rightarrow \cl{Y}.$ Then
    \begin{enumerate}
      \item $g + \cl{F} = \set{g + f : f \in \cl{F}}$ is a VC subgraph class with VC index equal to $V(\cl{F})$.
      \item $g\cl{F} = \set{gf : f\in \cl{F}}$ is a VC subgraph class with VC index  $\leq 2V(\cl{F})+1$.
      \item Let $f(x)$ be a VC subgraph class with index $V(\cl{F})$. Then the class of functions $\mathbf{1}\set{f(x) > 0}$ is also VC subgraph with index at most $V(\cl{F}).$
      \item Let $\cl{G}$ be the class of functions $g(x,y): X \times Y \rightarrow \bb{R}$ such that $g(x,y) = f(x)$ for some $f\in \cl{F}$. Then $\cl{G}$ is VC subgraph with index equal to $V(\cl{F}).$
      \item $|\cl{F}| = \set{|f| : f \in \cl{F}}$ is a VC subgraph class.
    \end{enumerate}
  \end{lemma}
  \begin{proof}

    Items (1) through (3) follow from inspecting the proofs of Lemma 2.6.20 of \cite
    {vaart2023empirical} parts i, vi, viii in conjunction with Problem 2.6.10. For item
    (4), this follows immediately by noting that the subgraphs of $\cl{G}$ shatter a set
    $\set{((Y_i,y_i),t_i)}_{i=1}^n$ if and only if the subgraphs of $\cl{F}$ shatter $\set{
    (Y_i,t_i)}_{i=1}^n$. For item 5, notice that the subgraph of $|f|$ is the union of
    the subgraphs of $f,-f$. Combining Lemma 2.6.20 parts ii and iv completes the proof.
  \end{proof}

  \begin{lemma}[VC function covering, Theorem 2.6.7 of \cite{vaart2023empirical}]
    \label{lemma:vc_subgraph_class_covering_number_bound}
    For a VC-subgraph class of functions with measurable envelope function $F$ and $r \geq 1$ we have for any probability measure $Q$ that
    \[
      N(\e \norm{F}_{Q,r}, \cl{F}, L_r(Q)) \leq K V(\cl{F})(16e)^{V(\cl{F})}\left( \frac{1}{\e}\right)^{r V(\cl{F})}
    \]
    for a universal constant $K$ and $0 < \e < 1$.
  \end{lemma}

  Following \citet{vaart2023empirical} (p.330), for $0 < \delta \leq 1$, $r \geq 1$, a
  function class $\cl{F}$ with envelope $F$, and a probability measure $Q$, write
  \[
    \norm{F}_{Q,r}^r := \int |F|^r\, dQ,
    \qquad
    J(\delta,\cl{F}\mid F,L_r) := \sup_Q \int_0^\delta \sqrt{1 + \log N\left(\e \norm{F}_{Q,r}, \cl{F}, L_r(Q)\right)}\, d\e,
  \]
  where the supremum is over all  discrete probability measures $Q$ with $\norm{F}_{Q, r}
  > 0$.

  Write $\bb{G}_n$ to be the i.n.i.d. empirical process
  \[
    \bb{G}_n f = \frac{1}{\sqrt{n}} \sum_{i=1}^n \left(f(V_i) - P_i f \right),
  \]
  where $V_1,\dots,V_n$ are independent random elements taking values in a common measurable space and $P_i f := \E f(V_i)$. The next theorem records a uniform entropy bound from \citet{vaart2023empirical} adapted to this setting. In the applications below, we will take $V_i = (Y_i,z_i)$ or $V_i = (Y_i,\mu_i,z_i)$, depending on the function class.

  \begin{theorem}[Theorem 2.14.1 of \citet{vaart2023empirical}]
    \label{thm:inid_donsker_bound}
    Let $V_1,\dots,V_n$ be independent random elements taking values in a common measurable space, and let $\cl{F}$ be a measurable class of functions on this space with measurable envelope function $F$. Then for $p \geq 1$,
    \[
      \norm{\norm{\bb{G}_n}^*_{\cl{F}}}_{p} \lesssim \norm{J(\eta_n, \cl{F} \ | \ F, L_2)\norm{F}_n}_{p} \lesssim J(1,\cl{F} \ | \ F,L_2) \left(\frac{1}{n}\sum_{i=1}^n \norm{F(V_i)}_{2 \vee p}^2 \right)^{1/2}
    \]
    where $\norm{g}_n^2 = \frac{1}{n} \sum_{i=1}^n g(V_i)^2$ is the $\bb{L}^2(\Prob_n)$-seminorm and $\eta_n = \norm{\norm{f}_n}_{\cl{F}}/\norm{F}_n$.
  \end{theorem}
  \begin{proof}
    The proof is verbatim from \cite{vaart2023empirical}. Denote by $\bb{G}^o_n$ the symmetrized process given by
    \[
      \bb{G}^o_n f := \frac{1}{\sqrt{n}}\sum_{i=1}^n \e_i f(V_i).
    \]
    The same argument as Theorem 2.14.1 of \cite{vaart2023empirical} shows that conditional on $V_1,\dots,V_n$, the sub-Gaussian norm of the symmetrized process satisfies
    \[
      \norm{\norm{\bb{G}_n^o}_{\cl{F}}}_{\psi_2|V} \lesssim \int_0^{\eta_n} \sqrt{1 + \log N(\e,\cl{F},L_2(\bb{P}_n))} d\e
    \]
    where $\bb{P}_n f = \frac{1}{n} \sum_{i=1}^n f(V_i)$ and
    \[
      \eta_n := \norm{\norm{f}_n}_{\cl{F}} = \sup_{f \in \cl{F}} \sqrt{\frac{1}{n} \sum_{i=1}^n f(V_i)^2}.
    \]
    A change of variable in the integral by $\e \leftarrow \e \norm{F}_n$ gives
    \begin{align*}
      \norm{\norm{\bb{G}_n^o}_{\cl{F}}}_{\psi_2|V} & \lesssim \norm{F}_n \int_0^{\eta_n/\norm{F}_n} \sqrt{1 + \log N(\e\norm{F}_n,\cl{F},L_2(\bb{P}_n))} d\e \\
      & \leq \norm{F}_n J( \eta_n, \cl{F} \ | \ F, L_2)
    \end{align*}
    The rest of the proof follows as is.
  \end{proof}

  \begin{theorem}[Localized i.n.i.d. maximal inequality, Theorem 2.14.2 of \cite{vaart2023empirical}]
    \label{thm:inid_localized_donsker_bound}
    Let $V_1,\dots,V_n$ be independent random elements taking values in a common measurable
    space, and let $\cl{F}$ be a measurable class of functions with measurable envelope
    $F$. Write
    \[
      \bar P f := \frac{1}{n}\sum_{i=1}^n P_i f,
      \qquad
      \norm{F}_{\bar P,2}:=(\bar P F^2)^{1/2}.
    \]
    Suppose that $F \leq 1$ and that, for some
    $\delta\in(0,1)$,
    \[
      \sup_{f\in\cl{F}}\bar P f^2 \leq \delta^2 \bar P F^2.
    \]
    If $\norm{F}_{\bar P,2}>0$, then
    \[
      \Estar\norm{\bb{G}_n}_{\cl{F}}
      \lesssim
      J(\delta,\cl{F}\mid F,L_2)
      \left(
      1+\frac{ J(\delta,\cl{F}\mid F,L_2)}
      {\delta^2\sqrt n\,\norm{F}_{\bar P,2}}
      \right)
      \norm{F}_{\bar P,2}.
    \]
    If $\norm{F}_{\bar P,2}=0$, then $\norm{\bb{G}_n}_{\cl{F}}=0$ almost surely.
  \end{theorem}
  \begin{proof}
Let
    \[
      S:=\Estar\norm{\bb{G}_n}_{\cl{F}},\qquad
      \sigma_n^2:=\sup_{f\in\cl{F}}\bb{P}_n f^2,\qquad
      \tau_n^2:=\bb{P}_n F^2.
    \]
    The same symmetrization and conditional maximal-inequality step used in
    \cref{thm:inid_donsker_bound} gives
    \[
      S\lesssim \Estar\left[
      J\left(\frac{\sigma_n}{\tau_n},\cl{F}\mid F,L_2\right)\tau_n
      \right],
    \]
    with the convention that the integrand is zero when $\tau_n=0$. Since
    $t\mapsto J(t,\cl{F}\mid F,L_2)$ is concave, Jensen's inequality gives
    \[
      S\lesssim
      J\left(\frac{\sqrt{\Estar\sigma_n^2}}{\norm{F}_{\bar P,2}},\cl{F}\mid F,L_2\right)
      \norm{F}_{\bar P,2}.
    \]
    Also,
    \[
      \Estar\sigma_n^2
      \leq \sup_{f\in\cl{F}}\bar P f^2
      + n^{-1/2}\Estar\norm{\bb{G}_n}_{\cl{F}^2}
      \leq \delta^2\norm{F}_{\bar P,2}^2 + C n^{-1/2}S,
    \]
    where $\cl{F}^2:=\{f^2:f\in\cl{F}\}$ and the last inequality follows from the
    contraction inequality, because $x\mapsto x^2$ is $2$-Lipschitz on $[-1,1]$.
    Let $z:=\sqrt{\Estar\sigma_n^2}/\norm{F}_{\bar P,2}$. If $z\leq\delta$, the claimed
    bound is immediate. If $z>\delta$, concavity implies
    \[
      J(z,\cl{F}\mid F,L_2)
      \leq \frac{z}{\delta}J(\delta,\cl{F}\mid F,L_2).
    \]
    Combining the preceding displays yields
    \[
      z^2 \leq \delta^2
      + C\frac{ J(\delta,\cl{F}\mid F,L_2)}
      {\delta\sqrt n\,\norm{F}_{\bar P,2}}z.
    \]
    The elementary implication $x^2\leq a^2+bx\implies x\leq a+b$ when $a,b > 0$ gives
    \[
      z\leq \delta
      + C\frac{ J(\delta,\cl{F}\mid F,L_2)}
      {\delta\sqrt n\,\norm{F}_{\bar P,2}},
    \]
    and substituting this into the bound for $S$ proves the theorem.
  \end{proof}

{
The previous theorem is most useful when the envelope $F$ is bounded by some other constant $M$, which may depend on $n$. The next corollary records this.
  \begin{corollary}
  \label{cor:inid_localized_donsker_bound}
    Under the same setting as \cref{thm:inid_localized_donsker_bound}, consider a function class $\cl{F}$ with envelope $F$ uniformly bounded by a constant $M$. Then,    
    \[
      \Estar\norm{\bb{G}_n}_{\cl{F}}
      \lesssim
      J(\delta,\cl{F}\mid F,L_2)
      \left(
      1+\frac{ M J(\delta,\cl{F}\mid F,L_2)}
      {\delta^2\sqrt n\,\norm{F}_{\bar P,2}}
      \right)
      \norm{F}_{\bar P,2}.
    \]
    where the implicit constant in $\lesssim$ does not depend on $M$. 
  \end{corollary}
  \begin{proof}
  Apply \cref{thm:inid_localized_donsker_bound} to the class of functions $\tilde{\cl{F}} := \set{f/M:f \in \cl{F}}.$ Define $\tilde{F} := F/M$ as the envelope function for $\tilde{\cl{F}}$. We obtain 
    \[
      \Estar\norm{\bb{G}_n}_{\tilde{\cl{F}} }
      \lesssim
      J(\delta,\tilde{\cl{F}}\mid \tilde{F},L_2)
      \left(
      1+\frac{ J(\delta,\tilde{\cl{F}}\mid \tilde{F},L_2)}
      {\delta^2\sqrt n\,\norm{ \tilde{F}}_{\bar P,2}}
      \right)
      \norm{\tilde{F}}_{\bar P,2}.
    \]
    First, note covering numbers remain unchanged after scaling. For any probability measure $Q$, 
    \[
    N\left(\e \norm{\tilde{F}}_{Q,r},\tilde{\cl{F}}, L_r(Q) \right) =
    N\left(\e \norm{F}_{Q,r},\cl{F}, L_r(Q) \right). 
    \]
    Thus, 
    \[
      \Estar\norm{\bb{G}_n}_{\tilde{\cl{F}} }
      \lesssim
      J(\delta,\cl{F}\mid F,L_2)
      \left(
      1+\frac{ MJ(\delta,\cl{F}\mid F,L_2)}
      {\delta^2\sqrt n\,\norm{F}_{\bar P,2}}
      \right)
      \frac{\norm{F}_{\bar P,2}}{M}.
    \]
    Multiplying by $M$ gives the result.
  \end{proof}
}
  \begin{proposition}[Orlicz norm bounds, i.n.i.d. analogue of Theorem 2.14.23
    \cite{vaart2023empirical}]
    \label{prop:orlicz_inid_bounds}

    Let $V_1,\dots,V_n$ be independent random elements taking values in a common measurable space, and let $\cl{F}$ be a class of measurable functions on this space with measurable envelope function $F$. Then
    \begin{align*}
      & \norm{\norm{\bb{G}_n}_{\cl{F}}^*}_{\psi_2} \lesssim \Estar\norm{\bb{G}_n}_\cl{F} +
      \left(\frac{1}{n} \sum_{i=1}^n \norm{F(V_i)}_{\psi_2}^2 \right)^{1/2}
    \end{align*}
  \end{proposition}
  \begin{proof}
    We recall Proposition A.1.6 of \cite{vaart2023empirical}. Let $U_1,\dots,U_n$ be independent mean zero stochastic processes indexed by an arbitrary set $T$. Set $S_n = \sum_{i=1}^n U_i$. Then we have
    \[
      \norm{\norm{S_n}^*}_{\psi_2} \leq K \left[ \norm{\norm{S_n}^*}_{1} + \left( \sum_{i=1}^n \norm{\norm{U_i}^*}_{\psi_2}^2 \right)^{1/2} \right],
    \]
    where $\norm{X}^*$ is shorthand for $\sup_t |X_t|$. Take $U_i(f) = \frac{1}{\sqrt{n}}\left(f(V_i) - \E f(V_i) \right)$, indexed by $f \in \cl{F}$ so that $\norm{S_n}^* = \norm{\bb{G}_n}_{\cl{F}}$. Plugging this into the proposition, we obtain
    \begin{align}
      \norm{\norm{\bb{G}_n}_{\cl{F}}^*}_{\psi_2} \lesssim E^*\norm{\bb{G}_n}_{\cl{F}} + n^{-1/2}\left( \sum_{i=1}^n \norm{\sup_{f \in \cl{F}} |f(V_i) - \E[f(V_i)] |}_{\psi_2}^2\right)^{1/2} \\
      \leq E^*\norm{\bb{G}_n}_{\cl{F}} + n^{-1/2}\left( \sum_{i=1}^n \norm{F(V_i) + \E[F(V_i)]}_{\psi_2}^2\right)^{1/2}
    \end{align}
    Where we used the fact that if $|f| \leq |g|$ then $\norm{|f(X)|}_{\psi_2} \leq \norm{|g(X)|}_{\psi_2}$ for any random variable. Using the triangle inequality for the Orlicz norm and the fact that $\norm{X}_{\psi_2} \geq \E|X|$ and $\norm{C}_{\psi_2} \leq C/\sqrt{\ln 2}$ for any constant $C,$ we can bound
    \begin{align*}
      \norm{F(V_i) + \E[F(V_i)]}_{\psi_2} & \leq \norm{F(V_i)}_{\psi_2} + \norm{\E[F(V_i)]}_{\psi_2} \\
      & \leq \norm{F(V_i)}_{\psi_2} + \E[F(V_i)]/\sqrt{\ln 2} \\
      & \lesssim \norm{F(V_i)}_{\psi_2}.
    \end{align*}
    This gives a final bound by a constant multiple of
    \[
      E^*\norm{\bb{G}_n}_{\cl{F}} + n^{-1/2}\left( \sum_{i=1}^n \norm{F(V_i)}_{\psi_2}^2\right)^{1/2},
    \]
    from which the desired claim holds.
  \end{proof}

  \begin{lemma}
    For any constant random variable $C$, $\norm{C}_{\psi_2} \leq |C|/\sqrt{\ln 2}$. For any random variable $X$, $\norm{X}_{\psi_2} \geq \E|X|$.
  \end{lemma}
  \begin{proof}
    Recall that for any random variable $X$,
    \[
      \norm{X}_{\psi_2} = \inf \set{c: \E\psi_2\left( \frac{|X|}{c}\right) \leq 1}.
    \]
    For the first claim, if $C = 0$ there is nothing to prove. Otherwise, taking $c = |C|/\sqrt{\ln 2}$, we see that $\psi_2(|C|/c) \leq \psi_2(\sqrt{\ln 2}) = 1$. For the second claim, notice that for any $c$ such that $\E\psi_2\left( \frac{|X|}{c}\right) \leq 1$, convexity of $\psi_2$ and Jensen's inequality implies that $\psi_2(\E|X|/c) \leq 1$. The explicit form of $\psi_2$ implies that $\E|X|/c \leq 1$. Thus $c > \E|X|$ as desired.
  \end{proof}

  \subsection{Examples of VC classes}
  \label{asub:vc_checks}

  \begin{lemma}
    For fixed finite $p$, the decision rule classes in (1)--(3) in \cref{sec:examples} are VC-subgraph classes.
  \end{lemma}

  \begin{proof}
    For the first class,
    \[
      \cl{D}_{\text{threshold}}
      = \set{z \mapsto k(z) + \beta : \beta \in [-M,M]}
      \subset k + \set{\beta \cdot 1 : \beta \in \bR},
    \]
    where $1(z) \equiv 1$. The class $\set{\beta \cdot 1 : \beta \in \bR}$ is a
    one-dimensional vector space of measurable functions, hence VC-subgraph by Lemma 2.6.16
    of \citet{vaart2023empirical}. By \cref{lemma:vc_indexes}(1), adding the fixed function
    $k$ preserves the VC-subgraph property, so $\cl{D}_{\text{threshold}}$ is VC-subgraph.

    For the second class,
    \[
      \cl{D}_{\text{$t$-stat}}
      = \set{z \mapsto k(z) + \beta \sigma(z) : \beta \in [-M,M]}
      \subset k + \set{\beta \sigma : \beta \in \bR}.
    \]
    The class $\set{\beta \sigma : \beta \in \bR}$ is again a one-dimensional vector space of
    measurable functions, so it is VC-subgraph by the same lemma. Another application of
    \cref{lemma:vc_indexes}(1) shows that $\cl{D}_{\text{$t$-stat}}$ is VC-subgraph.

    For the third class, let
    \[
      \mathcal V := \mathrm{span}\set{\delta^{(1)}, \dots, \delta^{(p)}}.
    \]
    Because $p < \infty$, $\mathcal V$ is a finite-dimensional vector space of measurable
    functions, hence VC-subgraph by Lemma 2.6.16 of \citet{vaart2023empirical}. Moreover,
    \[
      \mathcal D_{\text{finite}}
      = \mathrm{conv}\bigl(\set{\delta^{(1)}, \dots, \delta^{(p)}}\bigr)
      \subset \mathcal V.
    \]
    Any subclass of a VC-subgraph class is again VC-subgraph, so $\mathcal D_{\text{finite}}$
    is VC-subgraph as well.
  \end{proof}

  \begin{proposition}[VC-subgraph properties in \cref{sec:examples}]
    \label{prop:vc_subgraphness_examples}
    The decision threshold classes in \cref{tab:close-gauss-fam}
    are all VC-subgraph.
  \end{proposition}
  \begin{proof}[Proof of \cref{prop:vc_subgraphness_examples}]
    Let us consider the first claim. Notice that the linear shrinkage threshold rules are
    contained in the function class
    \[
      d_{a,b}(k,\sigma^2) := k + a \sigma^2 k  + b \sigma^2
    \]
    for $a,b \in \bR$. This is the fixed function $(k,\sigma^2) \mapsto k$ plus a
    two-dimensional vector space of functions. Combining Lemma 2.6.16 of
    \citet{vaart2023empirical} and \cref{lemma:vc_indexes}(1) establishes that this is VC
    subgraph. For the second claim, notice that Fay--Herriot decisions are contained in the
    function class
    \[
      d_{b_0,\beta}(k,\sigma^2,x) := k + b_0\sigma^2k - \sigma^2 x' \beta
    \]
    indexed by $(b_0,\beta) \in \bR^{1+p}$. The same argument as before establishes the claim.

    Finally for the third claim, we will use the results of \citet{onshuus2015metric}, which
    discuss a connection between VC dimension and the theory of $o$-minimality from model
    theory, a subfield of mathematical logic. See also \citet{aschenbrenner2013vapnik} and
    \citet{van1998tame} for textbook references on this topic. Consider the decision class in
    parametric \closegauss{} as functions of the cost $k$ and log standard deviation
    $\ell$:
    \[
      f_\beta(k,\ell) = k + e^{2\ell} e^{-b_1 - b_2 \ell} (k - a_1 - a_2 \ell)
    \]
    indexed by $\beta = (a_1,a_2,b_1,b_2).$ The subgraph sets
    \[
      S_{\beta} := \set{(k,\ell,t): t < k + e^{2\ell} e^{-b_1 - b_2 \ell} (k - a_1 - a_2 \ell)}
    \]
    are defined by a first-order formula in the structure $(\bR,+,\cdot,0,1,<,e^x)$, which is $o$-minimal: see Fact 2 of \citet{onshuus2015metric} or \citet{wilkie1996model}. Theorem 2 of \citet{onshuus2015metric} implies that for any set of points $\set{x_1,\dots,x_n}$ in $\bR^3$, the collection of sets $S_\beta$ pick out at most $O(n^4)$ subsets of $\set{x_1,\dots,x_n}$. Thus the function class $\set{f_\beta(k,\ell)}_\beta$ is VC subgraph.
  \end{proof}

  \section{Extension to Poisson observations}
  \label{sec:extensions_other_distributions}

  \assure{} can be extended to settings where the observation distribution is non-Gaussian. For certain exponential families, one can find an exactly unbiased estimate for the welfare $W$.
  The idea is to exploit a generalization of Stein's identity, called Hudson's Lemma \citep{hudson1978natural}, which holds for both continuous and discrete exponential families.

  We focus on the Poisson setting as it
  frequently appears in empirical Bayes and compound decisions.\footnote{See Chapter 6 of
    \cite{efron2021computer} for an overview of classic applications of the Poisson model to
    insurance claims, the missing species problems, and medical applications. See
    \cite{jana2022optimal, jana2023empirical} for modern methods to this problem and also
  \citet{montiel2021empirical} for econometric applications.} In this setting, let $\mu_i \geq 0$ and $Y_i \sim \Poi(\mu_i)$. Associated with
  each decision problem is a cost $k_i$ and auxiliary information $z_i$.\footnote{For
    example, analogous to $\sigma_i$ in the Gaussian case, $z_i$ might include the known
    length $t_i$ of the observation window for the Poisson outcome $Y_i$. If the unknown
    parameter $\mu_i$ denotes the Poisson rate per unit time, then $Y_i \sim \Poi(\mu_i t_i)$.
  The decision rule might incorporate the observation window $t_i$. } Consider any
  integer-valued threshold $\delta(z_i;\beta)$ parametrized by $\beta$. Then the
  estimator
  \begin{equation}
    \label{eq:sure_poisson}
    \hat{W}_{\mathsf{Poi}}(\beta) := \frac{1}{n}\sum_{i=1}^n \left( Y_i \mathbf{1}\set{Y_i \geq \delta(z_i;\beta) + 1} - k_i\mathbf{1}\set{Y_i \geq \delta(z_i;\beta)} \right)
  \end{equation}
  is the analog for the \assure{} estimator:

  \begin{restatable}{proposition}{proppoissonestimator}
    \label{prop:poisson_estimator}
    The Poisson \assure{} estimator $\hat{W}_{\mathsf{Poi}}(\beta)$ is unbiased for the welfare
    \begin{equation}
      \label{eq:poisson_welfare}
      W(\beta;\mu_{1:n}) :=  W(\beta) := \frac{1}{n}\sum_{i=1}^n \left(\mu_i - k_i \right) \Prob(Y_i \geq \delta(z_i;\beta)).
    \end{equation}
  \end{restatable}

  This follows as a direct consequence of Hudson's Lemma, so the proof is omitted. Under mild conditions, \assure{} achieves $1/\sqrt{n}$ regret, which
  is optimal. In contrast to the Gaussian case, the discreteness of the Poisson distribution
  suggests that no fast-rate regime is available.

  \begin{restatable}[Regret upper bounds for the Poisson Case]{theorem}{thmpoissonub}
    \label{thm:poisson_ub}
    Let $m_q := \left(\frac{1}{n}\sum_{i=1}^n \mu_i^q \right)^{1/q}$ for $q \ge 1$. Suppose the class of decision rules $\cl{D}$ is a VC subgraph class with index $V(\cl{D})$. Then \assure{} for Poisson observations has regret at most an absolute constant times
    \[
      \frac{1}{\sqrt{n}} \sqrt{V(\cl{D})} (m_2 + \sqrt{m_1}).
    \]
  \end{restatable}

  We prove the upper bound by appealing to the similar empirical process theory arguments as
  in the Gaussian case. We suspect some of the assumptions of \cref{thm:poisson_ub} may be
  further relaxed. This rate is the best possible, as the following lower bound shows.

  \begin{restatable}[Matching Regret Lower Bounds]{theorem}{thmmatchingpoissonlb}
    \label{thm:matching_lb_pois}
    Fix costs $k_i = k > 0$, and let $\Theta$ be a compact interval containing a
    neighborhood of $k$. We have
    \[
      \inf_{a(\cdot) : a_i \in \set{0,1}} \sup_{\mu_{1:n} \in \Theta^n} \regret_n(a_{1:n}) = \Omega(n^{-1/2}),
    \]
    where $\regret_n(a_{1:n})$ is defined analogously to \eqref{eq:gen_regret} except with
    $\mu_i$ replaced with $\mu_i - k$.
  \end{restatable}

  \begin{remark}
    Consider the case $k_i = k$ and the threshold class of decisions $Y_i \geq \delta, \delta \in \bR$. An alternative empirical Bayes approach to the Poisson selection problem is to threshold the \textit{Robbins estimator} given by
    \begin{equation}
      \hat{\mu}(y) := (y+1) \frac{N_{y+1}}{N_y}
    \end{equation}
    where $N_y := \#\set{i : Y_i = y}$. The estimator can also be used for selection by selecting all indices $i$ for which the Robbins estimator is
    greater than $k$. By rearranging $\hat{\mu}(Y_i) \geq k$, one can see that truncating
    the Robbins estimator bears a resemblance to the \assure{}
    decision. However, truncating the Robbins estimator does not respect monotonicity. The
    estimator \eqref{eq:sure_poisson} regularizes and enforces this monotonicity, giving a
    more stable and intuitive alternative to decisions based on the Robbins estimator. Related issues of the Robbins estimator have been discussed in the literature
    \citep{brown2013poisson, jana2022optimal}.
  \end{remark}

  \subsection{Proofs for Poisson extension}

  \begin{proof}[Proof of \cref{thm:poisson_ub}]
    Define $V(\beta) = \frac{1}{n}\sum_{i=1}^n \left(\mu_i - k_i\right) \mathbf{1}\set{Y_i \geq \delta(z_i;\beta)}$ be the in-sample welfare. The same decomposition as in \cref{thm:main_regret_bound} shows that the regret is bounded as
    \begin{align*}
      W(\beta^*) - \E[V(\hat{\beta})] \leq 2\E[\sup_\beta |V(\beta) - \hat{W}(\beta)|].
    \end{align*}
    Thus it suffices to control $\frac{1}{\sqrt{n}}\E \sup_\beta \left| \frac{1}{\sqrt{n}}\sum_{i=1}^n f_\beta(Y_i,\mu_i,z_i) \right|$ where
    \begin{equation}
      f_\beta(y,\mu,z) = y \mathbf{1}\set{y \geq \delta(z;\beta)+1} - \mu \mathbf{1}\set{y \geq \delta(z;\beta)}.
    \end{equation}
    Equivalently, $f_\beta(y,\mu,z) = (y - \mu) \mathbf{1}\set{y \geq \delta(z;\beta)+1} - \mu \mathbf{1}\set{y = \delta(z;\beta)}.$ Note that $\E f_\beta(Y_i,\mu_i,z_i) = 0.$ Let $\cl{F}$ denote this function class, indexed by $\beta$, with envelope function $F := y + \mu$. By the argument of \cref{thm:inid_donsker_bound} with $V_i = (Y_i,\mu_i,z_i)$, we may apply symmetrization and a maximal inequality to obtain
    \begin{equation}
      \E \sup_\beta \left| \frac{1}{\sqrt{n}}\sum_{i=1}^n f_\beta(Y_i,\mu_i,z_i) \right|\leq \E \left[  \int_0^{\eta_n} \sqrt{1 + \log N(\e ,\cl{F},\mathbb{L}_2(\Prob_n))} d\e \right],
    \end{equation}
    with $\eta_n := \sup_\beta \norm{f_\beta}_n$. Next, let $\cl{F}_1$ be the function class indexed by $\beta$ consisting of the functions $y\mathbf{1}\set{y \geq \delta(z;\beta)+1}$ and $\cl{F}_2$ be the function class consisting of $\mu \mathbf{1}\set{y \geq \delta(z;\beta)}$. By \cref{lemma:vc_indexes} these are VC subgraph classes with indices $O(V(\cl{D}))$. These have envelope functions $F_1 := y$ and $F_2 := \mu,$ respectively. Next, note that
    \[
      \log N(\e ,\cl{F},\mathbb{L}_2(\Prob_n)) \leq \log N(\e ,\cl{F}_1,\mathbb{L}_2(\Prob_n)) + \log N(\e ,\cl{F}_2,\mathbb{L}_2(\Prob_n))
    \]
    Using this and a change of variables, we arrive at
    \begin{align*}
      \int_0^{\eta_n} \sqrt{1 + \log N(\e,\cl{F},\mathbb{L}_2(\Prob_n))} d\e  & \lesssim  \norm{F_1}_n \int_0^{\frac{\norm{F}_n}{\norm{F_1}_n}} \sqrt{1 + \log N(\e \norm{F_1}_n ,\cl{F}_1,\mathbb{L}_2(\Prob_n))} d\e \\
      & + \norm{F_2}_n \int_0^{\frac{\norm{F}_n}{\norm{F_2}_n}} \sqrt{1 + \log N(\e \norm{F_2}_n ,\cl{F}_2,\mathbb{L}_2(\Prob_n))} d\e \\
      & \lesssim \norm{F}_n + \norm{F_1}_n J(1,\cl{F}_1 \ | \ F_1, \bL_2 ) + \norm{F_2}_n J(1,\cl{F}_2 \ | \ F_2, \bL_2 ) \\
      & \lesssim \norm{F}_n + \sqrt{V(\cl{D})}\left(\norm{F_1}_n + \norm{F_2}_n \right).
    \end{align*}
    Next, by Jensen's inequality observe the bounds
    \begin{align*}
      \E \norm{F_1}_n & \leq \left( \frac{1}{n}\sum_{i=1}^n \mu_i^2 + \mu_i \right)^{1/2}
    \end{align*}
    and $\E \norm{F}_n \leq \E \norm{F_1}_n + \norm{F_2}_n$ with $\norm{F_2}_n = \left( \frac{1}{n}\sum_{i=1}^n \mu_i^2 \right)^{1/2}$. Combining the pieces, we have
    \[
      \E \sup_\beta \left| \frac{1}{\sqrt{n}}\sum_{i=1}^n f_\beta(Y_i,\mu_i,z_i) \right| \lesssim \sqrt{V(\cl{D})} \left( \frac{1}{n}\sum_{i=1}^n \mu_i^2 + \mu_i \right)^{1/2}.
    \]
    This gives the result.
  \end{proof}

  \begin{proof}[Proof of \cref{thm:matching_lb_pois}]
    We can lower bound the compound total by restricting to two sets of parameters $\mu_1,\dots,\mu_n$ given by $\mu_i = z_0 := k + hn^{-1/2}$ for all $i$, and $\mu_i = z_1 := k - hn^{-1/2}$ for all $i=1,\dots,n.$ The parameter $h$ can be chosen later. In the first case, notice that the optimal decision is to select all indices $i$ (take $C_{\opt} = 0$). In the second case, the optimal decision is to select no indices (take $C_{\opt} = \infty$). Then:
    \begin{align*}
      \textsf{RegretComp}_n(\Theta) & \geq \inf_{a_i} \max \Bigg( hn^{-1/2} - hn^{-1/2}\frac{1}{n}\sum_{i=1}^n  \E_{z_0}(a_i(Y_1,\dots,Y_n)), \\
      & hn^{-1/2}\frac{1}{n}\sum_{i=1}^n  \E_{z_1}(a_i(Y_1,\dots,Y_n)) \Bigg)
    \end{align*}
    This can be written as
    \[
      \frac{h}{\sqrt{n}} \max\left(\E_{z_0}\left[1 - \frac{1}{n} \sum_{i=1}^n a_i(Y_1,\dots,Y_n)\right], \E_{z_1}\left[\frac{1}{n} \sum_{i=1}^n a_i(Y_1,\dots,Y_n)\right] \right)
    \]
    Notice that this quantity is $h/\sqrt{n}$ times the maximum of the Type I and Type II errors of the randomized test $(1/n)\sum_i a_i(Y_1, \ldots, Y_n)$. But for any test, the minimum total error for a binary hypothesis test by Le Cam's lemma is lower bounded by
    \[
      \frac{1}{2}\left\{1 - d_{TV}\left(\Poi(z_1)^{\otimes n},\Poi(z_0)^{\otimes n} \right)\right\}.
    \]
    Now, applying Pinsker's inequality and tensorization yields
    \begin{align*}
      d_{TV}\left(\Poi(z_1)^{\otimes n},\Poi(z_0)^{\otimes n} \right) & \leq \frac{1}{\sqrt{2}} \sqrt{d_{KL}\left(\Poi(z_1)^{\otimes n} | \Poi(z_0)^{\otimes n} \right)} \\
      & \leq \sqrt{\frac{n}{2}}\sqrt{d_{KL}\left(\Poi(z_1) | \Poi(z_0) \right)} \\
      & = \sqrt{n/2} \left[z_1 \ln \frac{z_1}{z_0} - (z_1 - z_0) \right]^{1/2}
    \end{align*}
    But
    \begin{align*}
      z_1 \ln \frac{z_1}{z_0} - (z_1 - z_0) & = z_1 \ln\left(1 + \frac{z_1 - z_0}{z_0}\right) - (z_1 - z_0)\\
      & = z_1 \left[\left(\frac{z_1 - z_0}{z_0}\right) - \frac{(z_1 - z_0)^2}{2z_0^2} + O\left((z_1-z_0)^3\right) \right] - (z_1 - z_0) \\
      & = \frac{(z_1 - z_0)^2}{z_0} - \frac{z_1(z_1 - z_0)^2}{2z_0^2} + O\left((z_1-z_0)^3\right) \\
      & = \frac{2z_0 - z_1}{2}\frac{(z_1 - z_0)^2}{z_0^2}  + O\left((z_1-z_0)^3\right) \\
      & = O\left(\frac{h^2}{n} \right).
    \end{align*}
    Thus
    \[
      d_{TV}\left(\Poi(z_1)^{\otimes n},\Poi(z_0)^{\otimes n} \right) \leq O(h),
    \]
    where $h$ can be taken small enough so that the TV distance is less than $1/2$. Then the total risk is bounded below by $\Omega(n^{-1/2})$ as desired.
  \end{proof}

  \section{Additional details for \cref{sec:oa_application}}
  \label{sec:additional_numerical_results}

  We start from the replication code in \citet{chen2022empirical} and follow
  \citet{chen2022empirical} to generate the raw data and covariates for each target
  variable (mean rank and top-20 probability), for the largest 20 Commuting Zones in the
  US. Both Opportunity Atlas (OA, \citet{chetty2018opportunity}) outcomes used in the
  paper retain $n=10{,}056$ tracts.

  \paragraph{Cost calibration and target shares.} Let $\tau \in \set{10,33}$ denote the
  target top share. For each outcome and each
  $\tau \in \set{10,33}$, we first fit the plug-in Fay--Herriot rule to the OA
  data and denote the resulting posterior means by $\hat m_i^{\mathrm{FH}}$.  
  The constant cost used in that scenario is then
  \[
    k^{(\tau)}
    := \text{the $\left\lceil \frac{\tau}{100} n \right\rceil$th largest element of }
    \set{\hat m_i^{\mathrm{FH}}}_{i=1}^n.
  \]
  Equivalently, $k^{(\tau)}$ is the weakest Fay--Herriot posterior mean among the tracts in
  the calibrated top $\tau\%$. We then set $k_i = k^{(\tau)}$ for every tract in that
  scenario.

  For completeness, \cref{fig:bergman_app_top10} reports the coupled-bootstrap
  comparisons for the stricter top-$10\%$ calibration. The broad ranking of methods
  is similar to the main-text top-third exercise.

  \begin{figure}[!htb]
    \centering

    \includegraphics[width=\textwidth]
    {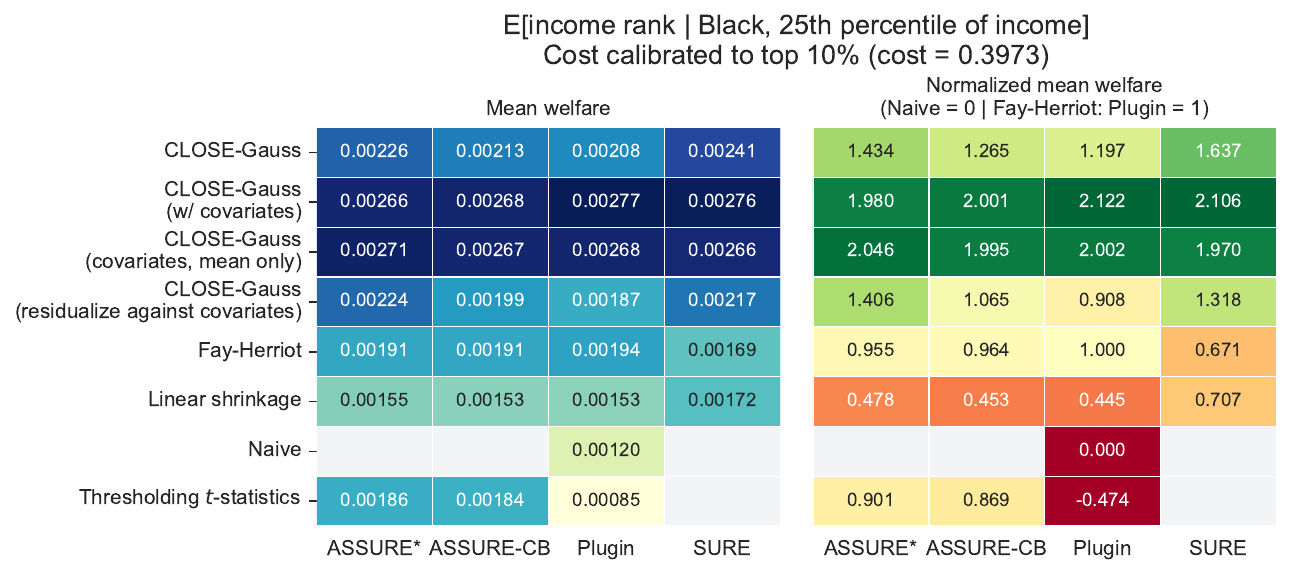}

    \includegraphics[width=\textwidth]
    {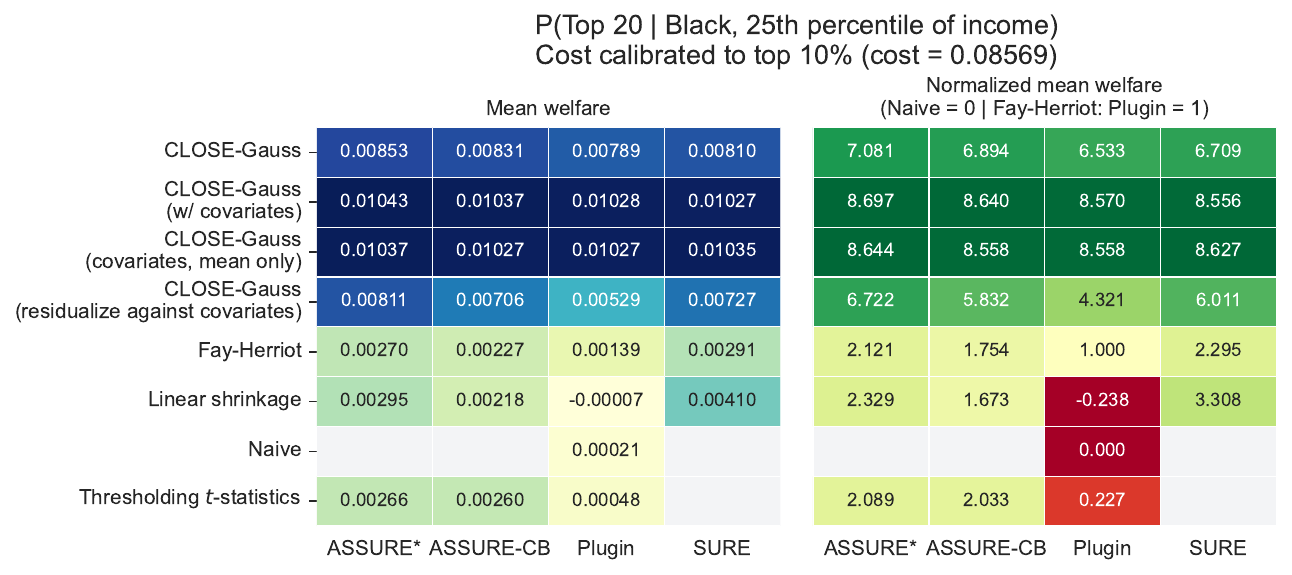}
    \caption{Coupled-bootstrap comparisons for the top-$10\%$ OA calibration.}
    \label{fig:bergman_app_top10}

    \begin{proof}[Notes]
      This is the analogue of \cref{fig:bergman_app}, but with costs calibrated
      to the weakest Fay--Herriot posterior mean among the top $10\%$ of tracts rather than
      the top third. Results are averaged over 1000 coupled-bootstrap draws.
    \end{proof}
  \end{figure}

  \paragraph{Coupled bootstrap.} Fix one outcome and one
  target share $\tau$.
  Starting from the observed OA estimates $\set{(Y_i,\sigma_i,x_i)}_{i=1}^n$, one Monte
  Carlo draw generates i.i.d. perturbations $Q_i \sim \Norm(0,1)$ and forms
  $
    Y_i^{\mathrm{tr}}
    = Y_i + \sqrt{w}\,\sigma_i Q_i,
    Y_i^{\mathrm{val}}
    = Y_i - \sqrt{\frac{1}{w}}\,\sigma_i Q_i,
  $  with
  \[
    \sigma_i^{\mathrm{tr}} = \sigma_i \sqrt{1+w},
    \qquad
    \sigma_i^{\mathrm{val}} = \sigma_i \sqrt{1+\frac{1}{w}}.
  \]
  The default implementation uses $w=1/9$, which mimics a $90/10$ train/validation split.
  Any decision rule fitting only sees the training sample
  $\set{(Y_i^{\mathrm{tr}},\sigma_i^{\mathrm{tr}},x_i,k^{(\tau)})}$; the validation sample
  is used only for evaluation.

  On draw $b$ of the coupled bootstrap, if a fitted rule returns thresholds
  $\delta_i^{(b)}$ on draw $b$, the reported validation welfare is
  \[
    \hat W_b^{\mathrm{val}}
    :=
    \frac{1}{n}
    \sum_{i=1}^n
    \one\pr{Y_i^{\mathrm{tr}} > \delta_i^{(b)}}
    \pr{Y_i^{\mathrm{val}} - k^{(\tau)}}.
  \]

  \paragraph{Decision rule classes.} The coupled-bootstrap OA exercise compares rules in
  \cref{sec:examples}. The naive rule is simply $\one(Y_i > k^{(\tau)})$. The $t$-statistic
  class uses thresholds of the form
  \[
    \delta_i = k^{(\tau)} + \beta \sigma_i,
  \]
  so it reduces to choosing a scalar cutoff $\beta$. The linear-shrinkage class is the
  constant-mean, constant-variance Gaussian-prior class from \cref{tab:close-gauss-fam},
  \[
    \delta_i
    =
    k^{(\tau)} + \frac{\sigma_i^2}{\tau^2}\pr{k^{(\tau)}-\mu_0}.
  \]
  Under constant costs, threshold-based criteria depend on $(\mu_0,\tau^2)$ only through
  the normalized slope
  \[
    \beta = \frac{k^{(\tau)}-\mu_0}{\tau^2},
    \qquad
    \delta_i = k^{(\tau)} + \beta \sigma_i^2.
  \]
  Fay--Herriot uses the conjugate-normal threshold implied by a prior mean
  $m_{0i}=x_i' \beta$ and a constant signal variance $A$. 

  Finally,
  {\footnotesize\closegauss} class uses
  \[
    \delta_i
    =
    k_i + \frac{\sigma_i^2}{s_{0i}^2}\pr{k_i - m_{0i}},
  \]
  where the prior mean $m_{0i}$ and prior variance $s_{0i}^2$ are modeled as functions of
  tract characteristics. In the plain {\footnotesize\closegauss} specification,
  $m_{0i}$ and $\log s_{0i}^2$ are both linear in $(1,\log \sigma_i)$. In the covariate
  version they are linear in $(1,\log \sigma_i,x_i)$. In the ``mean-only'' version the
  prior mean is linear in $(1,\log \sigma_i,x_i)$ but the signal variance is constrained to
  be constant across tracts. In the ``prefit'' version we first regress $Y_i$ on $x_i$,
  replace $Y_i$ and $k_i$ by residualized counterparts, fit the
  {\footnotesize\closegauss} rule using only $(1,\log \sigma_i)$ on those
  residuals, and then add the fitted covariate component back when interpreting thresholds
  in the original outcome scale.

  \paragraph{Plug-in fitting.} For each decision class, we may
  fit via plug-in:
  \begin{itemize}[wide]
    \item \emph{Naive rule.} There is no statistical fit: the plug-in rule is simply
      \[
        \delta_i = k^{(\tau)}.
      \]

    \item \emph{Thresholding $t$-statistics.} There is again no estimated nuisance
      parameter. The plug-in benchmark fixes
      \[
        \delta_i = k^{(\tau)} + 1.96\,\sigma_i,
      \]
      so the conventional two-sided $5\%$ rule is treated as the default member of this class.

    \item \emph{Linear shrinkage.}
      Plug-in estimates the prior mean by the sample average
      \[
        \hat \mu_{0,\mathrm{plugin}} = \bar Y,
        \qquad
        \bar Y = \frac{1}{n}\sum_{i=1}^n Y_i,
      \]
      and the signal variance by the average residual variance proxy
      \[
        \hat \tau^2_{\mathrm{plugin}}
        =
        \max\pr{
          \frac{1}{n}\sum_{i=1}^n (Y_i-\bar Y)^2
          -
          \frac{1}{n}\sum_{i=1}^n \sigma_i^2,\;
          \epsilon
        },
      \]
      truncating at a small $\epsilon$.

    \item \emph{Fay--Herriot.} Plug-in first
      estimates $\beta$ by ordinary least squares in the regression of $Y_i$ on $x_i$. It then
      forms the residual variance proxy
      \[
        \hat V_i = (Y_i - x_i' \hat \beta)^2 - \sigma_i^2
      \]
      and sets the signal variance equal to the cross-tract average of $\hat V_i$, truncated
      below at a numerically negligible positive floor so that the fitted variance remains
      strictly positive:
      \[
        \hat A_{\mathrm{plugin}}
        =
        \max\pr{\frac{1}{n}\sum_{i=1}^n \hat V_i,\; \epsilon}.
      \]

    \item \emph{{\footnotesize\closegauss} with covariates, mean only.} This is the
      same plug-in construction as Fay--Herriot, but with the enlarged design
      $z_i = (1,\log \sigma_i,x_i')'$ in place of $x_i$. Thus
      \[
        \hat m_{0i} = z_i' \hat \beta,
        \qquad
        \hat s_{0i}^2 = \hat A_{\mathrm{plugin}},
      \]
      where $\hat \beta$ is the least-squares coefficient on $z_i$.

    \item \emph{{\footnotesize\closegauss} and {\footnotesize\closegauss}
      with covariates.} These classes use
      \[
        \delta_i
        =
        k_i + \frac{\sigma_i^2}{s_{0i}^2}\pr{k_i - m_{0i}},
      \]
      with $m_{0i}$ linear in  $z_i$ and variance that is
      exponential-linear, up to a small variance floor $\underline s^2$:
      \[
        s_{0i}^2 = \underline s^2 + \exp(z_i' \gamma).
      \]
      Plug-in first estimates $\beta$,
      \[
        \hat \beta_{\mathrm{plugin}}
        \in
        \arg\min_\beta \frac{1}{n}\sum_{i=1}^n (Y_i - z_i' \beta)^2.
      \]
      Writing $\hat m_{0i} = z_i' \hat \beta_{\mathrm{plugin}}$, it then forms
      \[
        \hat V_i = (Y_i - \hat m_{0i})^2 - \sigma_i^2,
        \qquad
        \underline s^2
        =
        \frac{1}{10}\max\pr{\frac{1}{n}\sum_{i=1}^n \hat V_i,\;10^{-10}}.
      \]
      The variance parameter $\gamma$ is then chosen by nonlinear least squares:
      \[
        \hat \gamma_{\mathrm{plugin}}
        \in
        \arg\min_\gamma
        \frac{1}{n}\sum_{i=1}^n
        \pr{\hat V_i - \underline s^2 - \exp(z_i' \gamma)}^2.
      \]
      
      There are two additional truncation choices to prevent dividing by zero or over/underflow. First, the exponential index $z_i' \gamma$
      is clipped to the interval $[-10,10]$ before exponentiation, which rules out
      extremely small or large fitted signal variances. Second, the fitted prior mean can
      be clipped to the observed range of the estimation sample, $[\min_i Y_i,\max_i
      Y_i]$; the coupled-bootstrap OA runs leave this clipping turned off by default.

    \item \emph{Prefit {\footnotesize\closegauss}.} In this class the data are first
      residualized on $(1,x_i')$:
      \[
        R_i = Y_i - \hat g(x_i),
        \qquad
        k_i^R = k^{(\tau)} - \hat g(x_i),
      \]
      where $\hat g$ is the least-squares projection of $Y_i$ on $(1,x_i')$. Plug-in then
      applies the same procedure as in the plain {\footnotesize\closegauss} class, but
      to $(R_i,\sigma_i,k_i^R)$ with design $(1,\log \sigma_i)$. 
  \end{itemize}

  \paragraph{Methods for fitting decision classes.} We state the fitting problem
  separately for each rule class. For a generic threshold family
  $\delta_i(\theta)$, \assure{*} estimates \eqref{eq:sure_type}. The \assurecb{} fitter
  chooses
  \[
    \hat \theta_{\mathrm{cb}}
    \in
    \argmax_{\theta}
    \frac{1}{n}\sum_{i=1}^n
    (Y_i - k_i)\Phi\pr{\frac{Y_i-\delta_i(\theta)}{\epsilon_n \sigma_i}}
    -
    \frac{\sigma_i}{\epsilon_n}
    \phi\pr{\frac{Y_i-\delta_i(\theta)}{\epsilon_n \sigma_i}},
  \]
  with $\epsilon_n = n^{-0.2}$. For the Gaussian-prior classes, SURE minimizes
  \[
    \mathrm{SURE}(\theta)
    :=
    \frac{1}{n}\sum_{i=1}^n
    \pr{
      \hat \mu_i(\theta) - Y_i
    }^2
    +
    \frac{2}{n}\sum_{i=1}^n \sigma_i^2 w_i(\theta),
  \]
  where
  $
    \hat \mu_i(\theta)
    =
    w_i(\theta) Y_i + \pr{1-w_i(\theta)} m_{0i}(\theta)$ and $
    w_i(\theta)
    =
    \frac{s_{0i}^2(\theta)}{s_{0i}^2(\theta)+\sigma_i^2}.
  $
  Some notable implementation choices include:
  \begin{itemize}[wide]

    \item \emph{{\footnotesize\closegauss}.} Let
      \[
        z_i = \pr{1,\log \sigma_i}'.
      \]
      For some small floor $\underline s^2$, this class uses
      \[
        m_{0i}(\beta) = z_i' \beta,
        \qquad
        s_{0i}^2(\gamma) = \underline s^2 + \exp\pr{z_i' \gamma},
      \]
      so the threshold rule is
      \[
        \delta_i(\beta,\gamma)
        =
        k^{(\tau)} +
        \frac{\sigma_i^2}{\underline s^2 + \exp\pr{z_i' \gamma}}
        \pr{k^{(\tau)} - z_i' \beta}.
      \]
      Plug-in estimates $(\beta,\gamma)$ by least squares on the mean side and nonlinear
      least squares on the residual variance proxy; \assure{*}, \assurecb, and SURE then tune
      the same parameter vector.

    \item \emph{{\footnotesize\closegauss} with covariates.} This class enlarges the
      previous design to
      $
        z_i = \pr{1,\log \sigma_i,x_i'}',
      $

    \item \emph{Prefit {\footnotesize\closegauss}.} In this variant we first regress
      $Y_i$ on $(1,x_i')$ to obtain fitted values $\hat g(x_i)$ and residuals
      $
        R_i = Y_i - \hat g(x_i),
        k_i^{R} = k^{(\tau)} - \hat g(x_i).
      $
      We then fit the plain {\footnotesize\closegauss} class to
      \[
        R_i \sim \Norm\pr{\mu_i-\hat g(x_i),\sigma_i^2}
      \]
      using only $(1,\log \sigma_i)$ in the Gaussian-prior model:
      \[
        \delta_i^{R}(\beta,\gamma)
        =
        k_i^{R}
        +
        \frac{\sigma_i^2}{\underline s^2 + \exp\pr{\gamma_0 + \gamma_1 \log \sigma_i}}
        \pr{k_i^{R} - (\beta_0 + \beta_1 \log \sigma_i)}.
      \]
      The threshold in the original outcome scale is then
      \[
        \delta_i(\beta,\gamma) = \hat g(x_i) + \delta_i^{R}(\beta,\gamma).
      \]
      Plug-in, \assure{*}, \assurecb, and SURE are all applied to this residualized class.
  \end{itemize}

  \paragraph{Redundant variance intercept under constant costs when fitting \assure} 
  Under constant costs, the
  thresholds on the ratio $(k_i - m_{0i})/s_{0i}^2$, so the intercept in
  the variance parametrization for $s_{0i}^2$ is redundant.\footnote{This is exactly
  true if the small floor
  $\underline s^2$ were zero.} In the reported implementation, the optimizer fixes that
  intercept at its plug-in value and
  optimizes the remaining coordinates only.

  \paragraph{Warm starts and optimizer fallback.} The reported Gaussian-prior fits also use
  the following convention throughout. By default, \assure{*} is initialized at the
  \assurecb{} optimum rather than at the raw plug-in fit. 
  Conversely, the
  \assurecb{} fitter itself falls back to the plug-in fit whenever the optimized
  coupled-bootstrap objective is worse than the plug-in objective, so the fitter never
  accepts a numerically inferior point purely because the optimizer moved there.

  \paragraph{Optimization heuristics.} For the scalar $t$-statistic rule, both
  \assure{*} and \assurecb{} are tuned by deterministic one-dimensional grid search rather
  than by generic nonlinear optimization. 


  For the Gaussian-prior classes, \assure{*}, \assurecb, and SURE are all optimized
  multiple times, with added Gaussian noise, around the plug-in fitted parameters.
  In the OA analysis the default is three restarts
  for both
  \assure{*} and \assurecb.

  \section{Miscellany}

  \begin{lemma}[Properties of the sinc kernel]
    \label{lemma:sinc_properties}
    The function $x \mapsto \frac{1}{h}\sinc\left( \frac{x}{h}\right)$ is bounded above by $Ch^{-1},$ and further has a derivative bounded uniformly by $Ch^{-2}$ for an absolute constant $C.$
  \end{lemma}
  \begin{proof}
    The first claim follows since $|\sinc(y)| \leq \pi^{-1}$ per our definition. For the second, differentiating gives $h^{-2} \sinc'\left( \frac{x}{h}\right)$ where
    \[
      \sinc'(y) = \frac{y\cos y - \sin y}{\pi y^2}.
    \]
    An application of l'H\^opital's rule shows that this function is zero at zero. From this and continuity, it is not difficult to see that $\sinc'(y)$ is uniformly bounded by some absolute constant. This gives the second claim.
  \end{proof}

  \begin{lemma}[Derivatives of $\Psi$]
    \label{lemma:derivatives of Psi}
    Let $\Psi_C, \Psi_{CC}$ denote the first and second partial derivatives in C. They are given by the following formulas.
    \begin{align}
      \label{eq:assure_summand_first_derivative}
      \Psi_C(Y_i,z_i,C) & = -\frac{(Y_i - k_i)}{h\sigma_i}\sinc\left(\frac{Y_i - C}{h\sigma_i} \right) + \frac{1}{h^2} \sinc'\left(\frac{Y_i - C}{h\sigma_i} \right) \\
      \label{eq:assure_summand_second_derivative}
      \Psi_{CC}(Y_i,z_i,C) & = \frac{(Y_i - k_i)}{h^2\sigma_i^2}\sinc' \left(\frac{Y_i - C}{h\sigma_i}\right) - \frac{1}{\sigma_i h^3} \sinc'' \left(\frac{Y_i - C}{h\sigma_i} \right)
    \end{align}
  \end{lemma}

\end{appendices}

\end{document}